\documentclass[12pt]{article}

\usepackage[utf8]{inputenc}

\usepackage[
letterpaper,
margin=1in,
includehead,
includefoot
]{geometry}

\usepackage{setspace}
\usepackage{amsmath}
\usepackage{amssymb}
\usepackage{mathtools}
\usepackage[title]{appendix}
\usepackage{graphicx,psfrag,epsf}
\usepackage{enumerate}
\usepackage{natbib}
\usepackage{bm}
\usepackage{amsthm}
\usepackage{xcolor}
\usepackage{url}
\usepackage[pdfencoding=auto]{hyperref}
\usepackage{comment}

\newcommand{\blind}{1}

\newcommand{\tr}{\mathrm{tr}}
\allowdisplaybreaks[2]
\newcommand{\E}{\mathbb E}

\newtheorem{theorem}{Theorem}
\newtheorem{claim}{Claim}
\newtheorem{lemma}{Lemma}
\newtheorem{definition}{Definition}
\newtheorem{corollary}{Corollary}
\newtheorem{proposition}{Proposition}
\newtheorem*{theorem*}{Theorem}
\newtheorem{case*}{Case}
\newtheorem{case**}{Case}
\newtheorem{assumption}{Assumption}[section]

\newtheorem{assumptionprime}{Assumption}[section]

\newtheorem{assumptionprimeprime}{Assumption}[section]

\newtheorem{caseprime}{Case}
\newtheorem{caseprimeprime}{Case}

\newtheorem{remark}{Remark}
\begin{document}

	\if1\blind
	{
		\title{\bf Global identification of dynamic panel models with interactive effects}
		\author{Jushan Bai and Pablo Mones\thanks{
				Bai: Department of Economics,
				Columbia University, jb3064@columbia.edu;\\ \indent
				\hspace{0.3em}Mones: Department of Economics,
				Columbia University, pm3257@columbia.edu.\\
				We are grateful to conference participants at the NBER-NSF Time Series Conference, the International Panel Data Conference, the Annual Conference of the IAAE, the RCEA International Conference on Economics, Econometrics and Finance, and the New York Camp Econometrics, as well as seminar participants at Columbia University, Universidad de Montevideo, and Universidad de la Rep\'ublica for their insightful comments and suggestions. All remaining errors are our own. Pablo Mones thanks the Program for Economic Research (PER) at Columbia University for financial support.}
		}
		\maketitle
	} \fi
	
	\if0\blind
	{
		\bigskip
		\bigskip
		\bigskip
		\begin{center}
			{\LARGE\bf Global identification of dynamic panel models with interactive effects}
		\end{center}
		\medskip
	} \fi

	\bigskip
	\begin{abstract}
We investigate the problem of global identification in dynamic panel models with interactive effects, in the large-$N$, fixed-$T$ setting. While local identification, typically established via the Jacobian matrix, is well understood, global identification has remained a more elusive and challenging issue. It is commonly believed to be unachievable in this context. However, we demonstrate that the model is, in fact, globally identified for almost all configurations of the factors. Our analysis also covers models with additive fixed effects, including unit-root cases in which previous studies have reported non-identification from differenced moments. We show that, even in these settings, the level covariance structure delivers global identification.
	\end{abstract}
	
	\noindent%
	
	\noindent%
	
	\vfill
	
	\newpage
	\onehalfspacing
	\section{Introduction}\label{sec1}

Consider the dynamic panel data model with interactive effects
\begin{equation}\label{eq1:dgp}
    y_{it} = \alpha y_{it-1}+\delta_t+\lambda_i'f_t+\varepsilon_{it},
\end{equation}
where $i=1,2,\ldots,N$ and $t=1,2,\ldots,T$. Here, $y_{it}$ is the outcome
variable, $\alpha$ is the autoregressive parameter, and $\delta_t$ captures
time-specific fixed effects. The vector $f_t$ is $\bar r\times 1$ and contains
the values of the $\bar r$ latent factors at time $t$, while $\lambda_i$ is the
corresponding $\bar r\times 1$ vector of factor loadings for individual $i$.
The term $\varepsilon_{it}$ is an idiosyncratic error term. In this model,
$y_{it}$ is the only observable variable.

This framework provides a flexible way to model unobserved common trends and heterogeneous responses. The model allows the common shocks (modeled by $f_t$) to have a heterogeneous effect across individuals (modeled by $\lambda_i$) on the outcome variable $y_{it}$. A notable special case is the standard additive fixed-effects model, where $y_{it} = \delta_t + \gamma_i + \varepsilon_{it}$, which corresponds to $\alpha=0$ and a single constant factor $f_t=1$. Thus, this model generalizes the additive framework by allowing for richer dynamics and heterogeneity in the factor structure. Parameter estimation in this model can be conducted via quasi-maximum-likelihood based methods as in \cite{hayakawa2023short} and \cite{bai2024likelihood}.

This paper studies the problem of global identification of structural parameters when the number of cross-sectional units $(N)$ is large, but the time dimension $(T)$ is fixed. This case is particularly relevant in applied economics because many panel datasets, such as those involving countries, firms, or households, span a large number of cross-sectional units but only cover a limited time horizon. For instance, macroeconomic panels often include many countries but only a few decades of annual data, making the short-$T$, large-$N$ framework empirically realistic.

Following \cite{rothenberg1971identification}, we say that a parameter point $\theta^0$ in the parameter space $\Theta$ is globally identifiable if there is no other $\tilde{\theta}\in\Theta$ that is observationally equivalent. Establishing identifiability is crucial, as it is a prerequisite for estimation procedures to have a meaningful population target. Despite the apparent simplicity of the model, demonstrating parameter identifiability is a nontrivial task. The seminal work of \cite{anderson1956statistical} established sufficient conditions for identification in the absence of the autoregressive term ($\alpha=0$). In this same setting, \cite{bekker1997generic} showed that the parameters of the static factor model are generically globally identified (i.e. globally identified except for a subset of the parameter space of Lebesgue measure zero) when $\bar r$ is strictly below the \cite{ledermann1937rank} bound, up to the usual orthogonal $\bar r\times \bar r$ rotation matrix. This generic formulation is the standard benchmark in factor analysis and,
by extension, in interactive-effects models, where identification results
typically rule out lower-dimensional exceptional configurations.

To the best of our knowledge, no existing study has established the global identifiability of the parameters in the model described by \eqref{eq1:dgp}. \cite{bai2024likelihood} demonstrates that the parameters of this model are locally identified when $\bar{r}$ does not exceed the modified Ledermann bound. Relatedly, \cite{hayakawa2023short}, working in a transformed likelihood framework, argues that global identification cannot be guaranteed in general in dynamic panel data models with interactive effects.
	
We show that, under standard assumptions from the factor-analysis literature, the parameters of the model are globally identifiable when $T\geq 2(\bar r+1)$ for Lebesgue-almost every configuration of the factor realizations $(f_1,\ldots,f_T)$. The genericity requirement is only with respect to the factors: fixing any admissible value of the autoregressive coefficient and of the remaining parameters, global identification can fail only for a Lebesgue-null set of factor configurations.\footnote{This is a stronger result than generic identification in the full parameter space, which requires the set of unidentified parameter values to have Lebesgue measure zero jointly in all parameters. Here, by contrast, the non-factor parameters are arbitrary: for each admissible value of those parameters, identification may fail only on a Lebesgue-null set of factor configurations. The null set need not be the same for all values of the non-factor parameters.} Consequently, if the factors $f_t$ are viewed as realizations from an absolutely continuous stochastic process, then, conditional on any admissible values of the remaining parameters, the probability that global identification fails is zero. In this auxiliary probabilistic interpretation, we say that the parameters are almost surely globally identified.

Identification is achieved from the first two population moments of the outcome. Since the Gaussian quasi-likelihood depends on the model only through these moments, injectivity of the moment map implies uniqueness of the population Gaussian quasi-likelihood criterion, and hence almost sure Gaussian QML identification. This result is of particular importance, as it establishes a one-to-one correspondence between the population moments and the model parameters, despite the high degree of nonlinearity of the problem.

This paper contributes to the literature on global identification and panel data models. The question of global identification in parametric models has been the focus of extensive research (see, for example, \cite{shapiro1985identifiability}, \cite{komunjer2011dynamic}, \cite{komunjer2012global}, \cite{gorgens2019moment}, and \cite{kociecki2023solution}). A comprehensive review of the literature on identification is provided in \cite{lewbel2019identification}.

Panel data models have been a core focus in econometrics for many years, with extensive study and application over the past few decades. Notable monographs and textbooks on panel data include works by \cite{arellano2003panel}, \cite{baltagi2008econometric}, \cite{hsiao2022analysis}, and \cite{wooldridge2010econometric}, among others. Moreover, the study of panel data models with interactive fixed effects has gained notable attention during recent years.
 While there is an extensive literature in this area, some of the most noteworthy papers include \cite{holtz1988estimating}, \cite{ahn2001gmm}, \cite{pesaran2006estimation}, \cite{bai2009panel}, and \cite{ahn2013panel}. More recently, there has been significant research on dynamic panel models with interactive effects, including studies by \cite{moon2017dynamic}, \cite{hayakawa2023short}, and \cite{bai2024likelihood}, to name a few. However, these studies have not established the global identifiability of the model parameters for the short-$T$ case, a gap this paper aims to fill.

We also argue that the techniques developed can be used to prove identification in a broad class of dynamic models featuring time, individual, and interactive effects. For example, we show that in a dynamic panel with individual effects (\cite{arellano1991some}), the autoregressive coefficient can be globally identified (in the sense of \cite{rothenberg1971identification}) when $T\geq 5$ even when $\alpha=1$, a case previously reported to be non-identifiable (see \cite{blundell1998initial}, \cite{alvarez2022robust}, and \cite{sentana2024finite}). The comparison highlights that the unit-root identification problem is moment-set specific. While conventional moment conditions based on differenced data may fail to identify $\alpha$, the covariance restrictions embedded in the level representation contain additional information sufficient for global identification.

The remainder of the paper is organized as follows: Section \ref{sec2} formalizes the identification problem in our context. Section \ref{sec3} derives intermediate results and provides the main proof of identification under the stated assumptions. Section \ref{sec4} discusses identification in the presence of individual fixed effects. Finally, Section \ref{sec5} provides some concluding remarks and discusses potential extensions of the paper.
	
\section{Framework}\label{sec2}

\subsection{The Model in Matrix Notation}\label{sec2.1}

Before formally defining global identification in this context, it is useful to express the model in matrix notation. Suppose we observe $y_i = [y_{i1},y_{i2},\dots, y_{iT}]'$. We begin by projecting the initial observation $y_{i1}$ onto $[1,\lambda_{i}]$, which gives
\[y_{i1} = \delta_1^*+\lambda_{i}'f_1^*+\varepsilon_{i1}^*\]	
where $\delta_1^*$ and $f_1^*$ are the projection coefficients, and $\varepsilon_{i1}^*$ is the projection residual. This step should be understood as a reduced-form treatment of the initial condition. Let
\[
\varepsilon_i^*
=
(\varepsilon_{i1}^*,\varepsilon_{i2},\ldots,\varepsilon_{iT})'.
\]
A key requirement is that $\varepsilon_i^*$ has diagonal covariance matrix. The following martingale differences assumption on the structural error $\varepsilon_{it}$ will guarantee
this diagonality.
Let $
\mathbf F=(f_1,\ldots,f_T)'$. 
 Define the initial filtration ($\sigma$-field)\footnote{We could also include $\delta=(\delta_1,\dots,\delta_T)'$ in the filtration. Since $\delta$ is modeled as a vector of fixed constants, and more importantly, the sample covariance is invariant to $\delta$, it is omitted from the filtration.}
\[
\mathcal F_{i0}
=
\sigma(\lambda_i,y_{i0},\mathbf F).
\]
For $t\geq 1$, define
\[
\mathcal F_{it}
=
\sigma(\lambda_i,y_{i0},\mathbf F,\varepsilon_{i1},\ldots,\varepsilon_{it}).
\]
Assume that the structural innovations satisfy the martingale-difference
condition (MDC)
\begin{equation}\label{eq:mdc}
\E(\varepsilon_{it}\mid \mathcal F_{i,t-1})=0,
\qquad
t=1,\ldots,T.
\end{equation}
Then 
\[
\E(\varepsilon_i^*\varepsilon_i^{*'})
=
\operatorname{diag}
\left(
\E[(\varepsilon_{i1}^*)^2],
\E(\varepsilon_{i2}^2),
\ldots,
\E(\varepsilon_{iT}^2)
\right).
\]
To see this, we first argue $\varepsilon_{i1}^*$ is in $\mathcal F_{i1}$. 
The first observed period satisfies
\[
y_{i1}
=
\delta_1+\alpha y_{i0}+\lambda_i'f_1+\varepsilon_{i1}.
\]
Thus $y_{i1}\in \mathcal F_{i1}$. Next from $\varepsilon_{i1}^*= y_{i1} -\delta_1^*-\lambda_{i}'f_1^*$, 
since $y_{i1}$ and $\lambda_i$ are in $\mathcal F_{i1}$, and $\delta_1^*$ and $f_1^*$ are projection coefficients,
we have $\varepsilon_{i1}^*\in \mathcal F_{i1}$. 
Since $\mathcal F_{i1}\subseteq \mathcal F_{i,t-1}$ for $t\ge2$, the MDC implies
\[
\E(\varepsilon_{i1}^*\varepsilon_{it})
=
\E\left[
\varepsilon_{i1}^*\E(\varepsilon_{it}\mid \mathcal F_{i,1})
\right]
=
0.
\]
For $2\leq s<t\leq T$, $\varepsilon_{is}$ is
$\mathcal F_{i,t-1}$-measurable. Therefore
\[
\E(\varepsilon_{is}\varepsilon_{it})
=
\E\left[
\varepsilon_{is}\E(\varepsilon_{it}\mid \mathcal F_{i,t-1})
\right]
=
0.
\]
This proves that $\varepsilon_i^*$ has a diagonal covariance matrix. 

Since $\delta_1^*$ and $f_1^*$ are free parameters, and $\varepsilon_{it}$ in the model are not required to have identical distributions (especially, no homoskedasticity assumption), we can simply relabel $\delta_1^*, f_1^*$, and $\varepsilon_{i1}^*$ as $\delta_1$, $f_1$, and $\varepsilon_{i1}$, respectively, to simplify notation.
Thus we can rewrite
$y_{i1} = \delta_1+\lambda_{i}'f_1 +\varepsilon_{i1}$.

Next, stacking observations over $t$, we obtain
\[
    \mathbf{B} y_i = \delta + \mathbf{F}\lambda_i + \varepsilon_i,
\]
where $y_i=(y_{i1},y_{i2},\ldots,y_{iT})'$ is a $T\times 1$ vector of
individual observations, $\delta=(\delta_1,\ldots,\delta_T)'$ is a $T\times 1$
vector of time fixed effects, $\mathbf{F}=(f_1,\ldots,f_T)'$ is a
$T\times \bar r$ matrix containing the $\bar r$ factors, and
$\varepsilon_i=(\varepsilon_{i1},\varepsilon_{i2},\ldots,\varepsilon_{iT})'$
is a $T\times 1$ vector of idiosyncratic errors. The loading vector
$\lambda_i=(\lambda_{i1},\lambda_{i2},\ldots,\lambda_{i\bar r})'$ is assumed
to satisfy $\E(\lambda_i)=0$ when it is random, without loss of generality. If
$\lambda_i$ is random and $\E(\lambda_i)\neq 0$, its mean can be absorbed into
the time fixed effects by redefining
$\delta_t^\star=\delta_t+\E(\lambda_i)'f_t$ and
$\lambda_i^\star=\lambda_i-\E(\lambda_i)$. One can then work with the
observationally equivalent pair $(\delta_t^\star,\lambda_i^\star)$ instead of
$(\delta_t,\lambda_i)$. The case of deterministic loadings can also be
accommodated; see Remark~1 of \cite{bai2024likelihood}.

Furthermore, $\mathbf{B}$ is a nonsingular lower-triangular $T\times T$ matrix, with inverse given by $\mathbf{\Gamma}\equiv\mathbf{B}^{-1}$.
\begin{small}\[\mathbf{B}= \begin{bmatrix}
		1 & 0 & 0 & \dots  & 0 \\
		-\alpha & 1 & 0 & \dots  & 0 \\
		0 & -\alpha & 1& \dots  & 0 \\
		\vdots & \vdots &\ddots& \ddots & \vdots \\
		0 & \dots &0& -\alpha  & 1
	\end{bmatrix}, \qquad  \mathbf{\Gamma} = \begin{bmatrix}
		1 & 0 & 0&\dots  & 0 \\
		\alpha & 1 & 0&\dots  & 0 \\
		\alpha^2 & \alpha & 1 & \dots  & 0 \\
		\vdots & \vdots & \ddots & \ddots& \vdots \\
		\alpha^{T-1} & \alpha^{T-2} & \dots & \alpha  & 1
	\end{bmatrix}\] \end{small}
Pre-multiplying by $\mathbf{\Gamma}$ allows us to rewrite the model as
\begin{equation}\label{eq2:dgp_mat}
	y_{i} = \mathbf{\Gamma} \delta + \mathbf{\Gamma} \mathbf{F} \lambda_i + \mathbf{\Gamma}\varepsilon_i
\end{equation}
The error vector $\varepsilon_i$ has a diagonal covariance matrix denoted by $\mathbf{D} =\text{diag}(d_1,d_2,\dots,d_T)$. After transformation, the idiosyncratic error $\mathbf{\Gamma}\varepsilon_i$ has covariance matrix $\mathbf{\Gamma}\mathbf{D}\mathbf{\Gamma}'$, which is no longer diagonal. Let
\[
\mathbf{\Psi}_{N} = \frac{1}{N-1}\sum_{i=1}^{N}(\lambda_{i}-\bar{\lambda})(\lambda_{i}-\bar{\lambda})'
\]
denote the $\bar r\times \bar r$ sample covariance matrix of the loadings, where
$\bar{\lambda}=N^{-1}\sum_{i=1}^{N}\lambda_i$.
Introducing $\mathbf\Psi_N$ allows us to accommodate both random and
nonrandom factor loadings in a unified way.
We assume that
$\mathbf{\Psi}_{N}\to\mathbf{\Psi}$ as $N\to\infty$, with convergence
in probability when the loadings are random and deterministic convergence when
they are treated as fixed constants. Thus, $\mathbf{\Psi}$ is the limiting
cross-sectional covariance matrix of the factor loadings. Similarly, define
\[
\mathbf{S}_N=\frac{1}{N-1}\sum_{i=1}^{N}(y_i-\bar{y})(y_i-\bar{y})',
\]
where $\bar{y}=N^{-1}\sum_{i=1}^{N}y_i$. Following \cite{bai2024likelihood}, the likelihood is written in terms of the realized cross-sectional covariance of the loadings, $\mathbf{\Psi}_{N}$, rather than the individual loadings themselves (which are not identifiable in a short $T$ panel). Let
\[ \theta_N^0 \equiv \left(\alpha,\mathbf F,\mathbf\Psi_N,d_1,\ldots,d_T\right) \quad\text{and}\quad \theta^0 \equiv \left(\alpha,\mathbf F,\mathbf\Psi,d_1,\ldots,d_T\right).\]
Thus, $\theta_N^0$ denotes the finite-$N$ second-moment parameter vector, which uses the realized cross-sectional covariance matrix $\mathbf\Psi_N$, whereas $\theta^0$ denotes its population limit.
The time fixed effects $\delta$ are not included in these parameter vectors because they do not enter the cross-sectional covariance matrix.

Expectations below are taken with respect to the idiosyncratic errors, treating the realized factors and loadings as given. When the loadings are random, the same expressions are understood conditionally on their realization.
\[
\E[\mathbf{S}_N]
=
\mathbf{\Gamma}(\mathbf F\mathbf{\Psi}_{N}\mathbf F' + \mathbf D)\mathbf{\Gamma}'
\equiv \mathbf{\Sigma}(\theta^{0}_{N}).
\]
Under $\mathbf{\Psi}_N\to\mathbf{\Psi}$,
\begin{equation}\label{eq3:Sigma}
\mathbf{\Sigma}(\theta^{0}_{N})\to
	\mathbf{\Sigma}(\theta^0)
	\equiv
	\mathbf{\Gamma}(\mathbf F\mathbf{\Psi}\mathbf F' + \mathbf D)\mathbf{\Gamma}',
\end{equation}
where $\mathbf{\Sigma}(\theta^0)$ is a $T\times T$ symmetric positive definite (PD) matrix, $\mathbf{\Psi}$ is $\bar r\times\bar r$ symmetric PD, $\mathbf F\mathbf{\Psi}\mathbf F'$ is $T\times T$ symmetric
positive semidefinite (PSD) with rank $\bar r$, and $\mathbf D$ is a $T\times T$
diagonal PD matrix. Also, define:
\[
\mathbf{\Omega}=\mathbf{F}\mathbf{\Psi} \mathbf{F}' + \mathbf{D}
\]
as a $T\times T$ symmetric PD matrix corresponding to the expected value of the sample covariance matrix of $y_i$ in a model without an autoregressive
term $\alpha=0$ (\cite{anderson1956statistical}).

\subsection{Global identification of the model}\label{sec2.2}

This paper studies global identification through the Gaussian quasi-log-likelihood. Since the population criterion depends on the structural parameters only through the expected sample
mean and the limiting expected sample covariance of the outcome, identification reduces to showing that these two population objects uniquely determine the parameters. Without loss of generality, we center the factor loadings by absorbing their cross-sectional mean into the time effects. Under this normalization, define $\delta^\dag\equiv\mathbf{\Gamma}\delta$.\footnote{Once $\alpha$ is identified, we can use $\E[\bar{y}]$ to identify the correspondingly normalized time effects $\delta$.} The Gaussian quasi-log-likelihood function, evaluated at a candidate parameter $\theta$, is
\[-\frac N 2 \ln |\mathbf{\Sigma}(\theta)| -\frac 1 2 \sum_{i=1}^N (y_i -\delta^\dag)' \mathbf{\Sigma}(\theta)^{-1} (y_i-\delta^\dag)\]
Concentrating out $\delta^\dag$, which is estimated by $\bar y$, and dividing by $N$, we obtain\footnote{In the MLE, the definition of $\mathbf{S}_N $ needs a slight modification, namely, $1/(N-1)$ is replaced by $1/N$.}
\[\frac 1 N \ell_N(\theta) =-\frac 1 2 \ln |\mathbf{\Sigma}(\theta)| -\frac 1 2 \tr[ \mathbf{\Sigma}(\theta)^{-1} \mathbf{S}_N]\]
The concentrated criterion is equivalent to Stein's loss between the sample covariance matrix $\mathbf{S}_N$ and the model-implied covariance matrix $\mathbf{\Sigma}(\theta)$. This loss-function interpretation does not require normality, nor does it require the factor loadings to be random (see \cite{bai2024likelihood} for a discussion). That is, one can directly work with this loss function.

For a fixed $T$ and as $N\to\infty$, by the law of large numbers, $ \mathbf S_N-\mathbf\Sigma(\theta_N^0)\overset{p}{\to}0$.
Since $\mathbf\Sigma(\theta_N^0)\to\mathbf\Sigma(\theta^0)$, and hence $\mathbf S_N\overset{p}{\to}\mathbf\Sigma(\theta^0)$.
Therefore,
\begin{equation}\label{eq4:QML}
	\frac{2}{N}\ell_N(\theta) \overset{p}{\to} -\ln\vert \mathbf{\Sigma}(\theta)\vert - \tr\left[\mathbf{\Sigma}(\theta^0)\mathbf{\Sigma}(\theta)^{-1}\right].
\end{equation}
Furthermore, it is known (e.g. Example 3.23 of \cite{boyd2004convex}) that this function has a unique maximizer at $\mathbf{\Sigma}(\theta)=\mathbf{\Sigma}(\theta^0)$. While this result ensures that $\theta^0$ is a global maximizer of the objective function, it does not immediately establish uniqueness---hence, global identification---because it does not preclude the possibility that other parameter values $\tilde{\theta}\neq\theta^0$ satisfy $\mathbf{\Sigma}(\tilde{\theta})=\mathbf{\Sigma}(\theta^0)$.
In the following sections, we investigate whether any alternative parameter set $\tilde{\theta}\neq \theta^0$ satisfies $\mathbf{\Sigma}(\tilde{\theta})=\mathbf{\Sigma}(\theta^0)$. We demonstrate that, almost surely, no such alternative parameterization exists. Consequently, the Gaussian quasi-log-likelihood population function \eqref{eq4:QML} is (almost surely) uniquely maximized at the true parameter value $\theta^0$, establishing identification.

To prove this claim, consider an alternative parameterization: $\tilde{\mathbf{F}} = \left[\tilde{f}_{1},\tilde{f}_{2}, \dots, \tilde{f}_{T}\right]'$ a $T\times \bar{r}$ matrix, $\tilde{\mathbf{\Psi}}$ a symmetric PD $\bar{r}\times\bar{r}$ matrix, $\tilde{\mathbf{D}} = \text{diag}(\tilde{d}_1,\dots \tilde{d}_{T})$ a diagonal PD $T\times T$ matrix, and $\tilde{\mathbf{\Gamma}}$ the analogue of $\mathbf{\Gamma}$ under an alternative parameter $\tilde{\alpha}$. Note that $\tilde{\mathbf{F}} \tilde{\mathbf{\Psi}} \tilde{\mathbf{F}}'$ must be a symmetric and PSD $T\times T$ matrix of rank $\bar{r}$.

We say the model is globally identified at the true parameter value $\theta^0$ if every admissible parameterization satisfying $\mathbf{\Sigma}(\tilde{\theta})=\mathbf{\Sigma}(\theta^0)$ also satisfies $\tilde{\theta}=\theta^0$. Equivalently, establishing global identification requires showing that any admissible solution to
\begin{equation}\label{eq5:objective}
	\mathbf{\Gamma}(\mathbf{F}\mathbf{\Psi}\mathbf{F}'+\mathbf{D})\mathbf{\Gamma}' = \tilde{\mathbf{\Gamma}}(\tilde{\mathbf{F}} \tilde{\mathbf{\Psi}} \tilde{\mathbf{F}}'+\tilde{\mathbf{D}})\tilde{\mathbf{\Gamma}}'
\end{equation}
must satisfy $\tilde{\alpha}=\alpha$, $\tilde{\mathbf D}=\mathbf D$, and $\tilde{\mathbf F}\tilde{\mathbf\Psi}\tilde{\mathbf F}'=\mathbf F\mathbf\Psi\mathbf F'$.
For ease of exposition, we write this as
$\tilde{\theta}=\theta^0$. It is important to keep in mind that the matrices $\mathbf{F}$ and $\mathbf{\Psi}$ are not separately identified without restrictions. In particular, it is well-known (see Section 5 of \cite{anderson1956statistical}) that factor models (and thus interactive effects) are identifiable up to a $\bar{r}\times\bar{r}$ matrix rotation, so $\bar{r}^2$ restrictions must be imposed for identification. Once these restrictions are imposed, both $\mathbf{F}$ and $\mathbf{\Psi}$ are separately identified.

\subsection{Assumptions, definitions, and useful results}

The following assumptions are imposed throughout the paper:

\begin{assumption}\label{as:1}
	The number of distinct factors is known and equal to $\bar{r}\geq 1$.
\end{assumption}

\begin{assumption}\label{as:2}
	$\mathbf{\Psi}$ is an unrestricted symmetric positive definite matrix, and
	$\mathbf{F}=(\mathbf{I}_{\bar r},\mathbf{F}_2')'$, where $\mathbf{F}_2$
	is an unrestricted $(T-\bar r)\times\bar r$ matrix.
\end{assumption}

\begin{assumption}\label{as:3}
	
	$T \geq 2(\bar{r}+1)$.
	
\end{assumption}

\begin{assumption}\label{as:4}

For each $t$, the errors $\varepsilon_{it}$ in model \eqref{eq1:dgp} are
i.i.d.\ across $i$. For each $i$, they satisfy the MDC over $t$ with respect to
the filtration $\mathcal F_{it}$ defined above, and
$\mathbb V(\varepsilon_{it})=d_t>0$.

\end{assumption}

\begin{assumption}\label{as:5}
	The factor loadings $\lambda_{i}$ are either:
	\begin{enumerate}
		\item Random coefficients satisfying $\E[\lambda_i]=0$ and $\mathbf{\Psi}_{N} \xrightarrow{p} \mathbf{\Psi}$ as $N\to\infty$.
		\item Deterministic coefficients satisfying $\mathbf{\Psi}_{N} \to \mathbf{\Psi}$ as $N \to \infty$.
	\end{enumerate}
\end{assumption}

\begin{assumption}\label{as:6} The unrestricted factor coordinates $\operatorname{vec}(\mathbf F_2)$ are generated from an absolutely continuous distribution with respect to Lebesgue measure on $\mathbb R^{(T-\bar r)\bar r}$.
\end{assumption}

\noindent Given Assumptions \ref{as:1}-\ref{as:5}, we introduce the identification concept used in Section \ref{sec3}. Formally, we say that the parameters of the model are globally identified for almost all factor configurations if for any fixed admissible value of $(\alpha,\delta,\mathbf{D},\mathbf{\Psi})$ global identification (as defined in Subsection \ref{sec2.2}) holds except for a Lebesgue-zero subset of the unrestricted factors $\mathbf{F}_{2}$.

If Assumption \ref{as:6} is additionally imposed, this statement translates into almost sure identification, where the almost sure qualifier is with respect to the distribution of $\mathbf{F}_{2}$.\footnote{This convention in language will be maintained throughout the paper.} This means that if one is willing to interpret the unrestricted factors as the realization of some arbitrary stochastic process drawn from an absolutely continuous distribution with respect to the Lebesgue measure in $\mathbb{R}^{(T-\bar{r})\bar{r}}$, the probability that the factors lie in the potentially underidentified region is zero. The plausible practical relevance of this probability-one statement depends on the context (see Section \ref{sec4}).

Strictly speaking, the paper shows global identification for almost all factor configurations for any fixed admissible value of $(\alpha,\delta,\mathbf{D},\mathbf{\Psi})$, which is stronger than generic identification (which requires a genericity argument over the whole parameter space $\Theta$) and almost sure identification (which additionally requires factors to be drawn from an absolutely continuous distribution). However, we present the results as almost sure identification, because this concept conveys the economic intuition of the result in a clearer way.

Remark \ref{rem:1} clarifies the role of Assumption \ref{as:6} and the scope of the identification results established in this paper.

\begin{remark}\label{rem:1}
For identification purposes, the factors $\mathbf{F}$ are treated throughout as fixed population parameters. Assumption \ref{as:6} does not give the factors a structural random interpretation. Rather, it specifies the measure with respect to which the genericity statement in the identification theorem is made.

More precisely, the argument in Section \ref{sec3} shows that, fixing any admissible values of the non-factor parameters, global identification is guaranteed for Lebesgue-almost every admissible configuration of the unrestricted factors $\mathbf{F}_{2}$. Hence, any failure of the identification argument can occur only on a Lebesgue-null subset of the factor-configuration space.

The probabilistic wording is only an equivalent formulation of this deterministic genericity statement. If the unrestricted factor coordinates are viewed as realizations from an absolutely continuous distribution, then the exceptional set has probability zero. In this sense, the parameters are almost surely globally identified. Without this probabilistic interpretation (i.e., after removing Assumption \ref{as:6}), the result should be read as follows: for each fixed admissible configuration of the remaining parameters, global identification holds for Lebesgue-almost every configuration of the unrestricted factors. Thus the genericity statement is imposed on the factor configuration, rather than on the full parameter vector, although the exceptional null set of factor configurations may depend on the fixed values of the remaining parameters.
\end{remark}

Additionally, throughout the paper, we will reference some established results from linear algebra, probability theory and algebraic geometry. To ensure they are easily accessible to the reader, we present them as lemmas in this section.	

\begin{definition}
	Let $\mathbf{A}$ be a $T \times T$ matrix, and let $
	R = \{r_1, r_2, \dots, r_k\}, \;
	C = \{c_1, c_2, \dots, c_k\}$
	be index sets for rows and columns, respectively, with $|R| = |C| = k$.
	We write $\mathbf{M}_{R,C}^{\mathbf{A},k}$ for the $k \times k$ submatrix of $\mathbf{A}$ obtained by restricting $\mathbf{A}$ to the rows in $R$ and the columns in $C$.
	Its determinant, $\det\left(\mathbf{M}_{R,C}^{\mathbf{A},k}\right)$, is called the minor associated with $(R,C)$.
\end{definition}

Throughout the paper, the following convention applies:

\begin{remark}\label{rem:2}
	Whenever a submatrix $\mathbf{M}_{R,C}^{\mathbf{A},k}$ is formed, the rows and columns are listed in increasing order of their indices so that $r_1<\dots<r_k$ and $c_1<\dots<c_k$. This convention fixes the sign of $\det\left(\mathbf{M}_{R,C}^{\mathbf{A},k}\right)$. A different ordering of the same row or column set may change the determinant by a sign, but it does not affect whether the determinant is zero. Whenever an element is removed from $R$ or $C$, the remaining set is also listed in increasing order.
\end{remark}

Now consider the following results.

\begin{lemma}\label{lemma:1}
	Let $\mathbf{A}$ be a $T \times T$ matrix with rank $\bar r < T$. Then the determinant of every $k \times k$ submatrix of $\mathbf{A}$, with $k>\bar{r}$, is equal to zero.
\end{lemma}

\begin{lemma}\label{lemma:2}
	Let $\mathbf{X}$ be an absolutely continuous random vector in $\mathbb{R}^q$, and let $\Upsilon \subset \mathbb{R}^q$ be a set of Lebesgue measure zero. Then $P(\mathbf{X} \in \Upsilon)=0$.
\end{lemma}

\begin{lemma}\label{lemma:3}
	Let $p:\mathbb{R}^q \to \mathbb{R}$ be a nonzero polynomial. Then the zero set $
	\{x \in \mathbb{R}^q : p(x)=0\}$ has Lebesgue measure zero in $\mathbb{R}^q$.
\end{lemma}

See Theorem 4.5 of \cite{ersel2011linear}, Definition 4.1.1 of \cite{athreya2006measure}, and \cite{caron2005zero} for the respective references.

\section{Global identification}\label{sec3}

\subsection{Overview of the proof}

In this section, we demonstrate that, under Assumptions \ref{as:1}-\ref{as:6}, the parameters of the dynamic panel data model with interactive effects described in \eqref{eq1:dgp} are almost surely globally identified. The proof we propose draws from Theorem 5.1 of \cite{anderson1956statistical}. In that theorem, the authors establish a sufficient condition for identification of a factor model that did not include an autoregressive component. In their simpler model, global identification is obtained by showing that any admissible decomposition satisfying
\begin{equation}\label{eq6:AR}
	\mathbf{F}\mathbf{\Psi} \mathbf{F}' + \mathbf{D}
	=
	\tilde{\mathbf{F}} \tilde{\mathbf{\Psi}} \tilde{\mathbf{F}}'
	+\tilde{\mathbf{D}}
\end{equation}
must have $\tilde{\mathbf D}=\mathbf D$ and $\tilde{\mathbf F}\tilde{\mathbf\Psi}\tilde{\mathbf F}'=\mathbf F\mathbf\Psi\mathbf F'$. Their proof relies on the fact that the off-diagonal elements of $\tilde{\mathbf{F}} \tilde{\mathbf{\Psi}} \tilde{\mathbf{F}}'$ must match those of $\mathbf{F}\mathbf{\Psi} \mathbf{F}'$ for \eqref{eq6:AR} to hold, because both $\mathbf{D}\ \text{and} \ \tilde{\mathbf{D}}$ are diagonal matrices. Thus, the problem reduces to showing that the diagonal elements of these matrices are also equal. To establish this, they leverage the fact that all $(\bar r+1)\times(\bar r+1)$ minors of $\tilde{\mathbf F}\tilde{\mathbf\Psi}\tilde{\mathbf F}'$ must vanish for \eqref{eq6:AR} to hold. This follows because any admissible alternative decomposition must satisfy that
$\tilde{\mathbf F}\tilde{\mathbf\Psi}\tilde{\mathbf F}'$ also has rank
$\bar r$. Lemma \ref{lemma:1} then translates these rank restrictions into
restrictions on the relevant minors. Thus, since
\[\tilde{\mathbf{F}} \tilde{\mathbf{\Psi}} \tilde{\mathbf{F}}' = \mathbf{F}\mathbf{\Psi} \mathbf{F}' + \mathbf{D}-\tilde{\mathbf{D}}\]
this permits them to derive conditions that allow them to pin down the diagonal terms of $\tilde{\mathbf{F}} \tilde{\mathbf{\Psi}} \tilde{\mathbf{F}}'$, thereby establishing that $\tilde{\mathbf{F}} \tilde{\mathbf{\Psi}} \tilde{\mathbf{F}}'=\mathbf{F}\mathbf{\Psi} \mathbf{F}'$. Once this is established, it follows immediately that $\tilde{\mathbf{D}}=\mathbf{D}$. Also, note that once the rotation is fixed (see Assumption \ref{as:2}) $\tilde{\mathbf{F}} \tilde{\mathbf{\Psi}} \tilde{\mathbf{F}}'=\mathbf{F}\mathbf{\Psi} \mathbf{F}'$ implies $\tilde{\mathbf{F}}= \mathbf{F}$, and $\tilde{\mathbf{\Psi}}=\mathbf{\Psi}$.

In their approach, \cite{anderson1956statistical} focus exclusively on minors involving diagonal elements of the matrix. In contrast, our proof takes a different approach: we rearrange the terms of our original problem—Equation \eqref{eq5:objective}—in a way that allows us to exploit the information contained in the \textit{diagonal exclusion minors} (minors that do not include diagonal terms of the original matrix) of this transformed problem. This novel approach enables us to show that any admissible solution to \eqref{eq5:objective} must satisfy $\tilde{\alpha}=\alpha$ for almost all configurations of the unrestricted factors, thereby establishing almost sure global identification of $\alpha$.

Once we establish that $\alpha$ is almost surely globally identified, it follows that our identification condition can be rewritten in the form proposed by \cite{anderson1956statistical}, ensuring that all other parameters can also be almost surely globally identified. Thus, the primary focus of our analysis is to show that any admissible parameterization satisfying \eqref{eq5:objective} must have $\tilde{\alpha}=\alpha$. This task requires some intermediate results, which we articulate and prove throughout the paper. The main contribution of this work is to formally establish the following Theorem.

\begin{theorem}\label{thm:1} Let $y_{it}$ be generated from \eqref{eq1:dgp}. Then, under Assumptions \ref{as:1}-\ref{as:6}, $\alpha$ is almost surely globally identified.
\end{theorem}

Once this is shown, it is immediate to conclude that all the parameters of the model are almost surely globally identified as well, which is formally stated in the following theorem.

\begin{theorem}\label{thm:2} Let $y_{it}$ be generated by \eqref{eq1:dgp}. Then, under Assumptions \ref{as:1}--\ref{as:6}, $\{d_{t}\}_{t=1}^{T},\mathbf{\Psi}$, and $\mathbf{F}$ are almost surely globally identified.
\end{theorem}

\begin{remark}\label{rem:3}
	The qualifier ``almost surely'' is with respect to the distribution of $\mathbf{F}_2$. For any fixed $(\alpha,\delta,\mathbf{D},\mathbf{\Psi})$, identification holds for almost all realizations of $\mathbf{F}_2$. The practical implication is that as long as there is sufficient variation in $f_t$, both $\alpha$ and other parameters are globally identifiable. If Assumption \ref{as:6} is removed, the results can be interpreted as holding almost-everywhere in the unrestricted factor space for each fixed value of the remaining parameters.
\end{remark}

\subsection{Intermediate results}\label{sbs:32}

In order to prove Theorem \ref{thm:1}, we rearrange the problem in a way that resembles the proof of \cite{anderson1956statistical} and then establish some properties of this modified problem. Start from \eqref{eq5:objective}. As $\tilde{\mathbf{\Gamma}}$ is a unit-diagonal lower-triangular matrix, it is invertible so we can pre-multiply and post-multiply by $\tilde{\mathbf{\Gamma}}^{-1}$ and $\tilde{\mathbf{\Gamma}}^{-1'}$, and subtract $\tilde{\mathbf{D}}$ on both sides to rewrite
\begin{equation}\label{eq7:mod_objective}
	\left(\tilde{\mathbf{\Gamma}}^{-1}\mathbf{\Gamma}\right)\Bigl(\mathbf{F}\mathbf{\Psi}\mathbf{F}'+\mathbf{D}\Bigr)\left(\tilde{\mathbf{\Gamma}}^{-1}\mathbf{\Gamma}\right)'-\tilde{\mathbf{D}} = \tilde{\mathbf{F}} \tilde{\mathbf{\Psi}} \tilde{\mathbf{F}}'
\end{equation}
The product $\tilde{\mathbf{\Gamma}}^{-1}\mathbf{\Gamma}$ is a lower-triangular matrix that can be decomposed into
\[\tilde{\mathbf{\Gamma}}^{-1}\mathbf{\Gamma} = \mathbf{I}_T + (\alpha-\tilde{\alpha})\mathbf{L},\qquad \mathbf{L}=\begin{bmatrix}
	0 & 0 & 0 & \dots  & 0 \\
	1& 0 & 0 & \dots  & 0 \\
	\alpha & 1 & 0 & \dots &0\\
	\vdots & \vdots & \ddots & \ddots & \vdots \\
	\alpha^{T-2} &\alpha^{T-3}& \dots & 1  & 0
\end{bmatrix}\]
where $\mathbf{L}$ is a strictly lower-triangular Toeplitz matrix. To simplify notation, denote the left-hand side of \eqref{eq7:mod_objective} by
\begin{equation}\label{eq8:O}
	\mathbf{O} \equiv \left(\tilde{\mathbf{\Gamma}}^{-1}\mathbf{\Gamma}\right)\Bigl(\mathbf{F}\mathbf{\Psi}\mathbf{F}'+\mathbf{D}\Bigr)\left(\tilde{\mathbf{\Gamma}}^{-1}\mathbf{\Gamma}\right)'-\tilde{\mathbf{D}}
\end{equation}
We next establish some preliminary results.
\begin{proposition} \label{prop:1}
	Let $\mathbf{M}_{R,C}^{\mathbf{O},k}$ be any square submatrix of $\mathbf{O}$ of dimension $k\times k$, with $k>\bar{r}$. Then $\det\left(\mathbf{M}_{R,C}^{\mathbf{O},k}\right)=0$ is a necessary condition for \eqref{eq5:objective} to hold.
\end{proposition}
\begin{proof}
	We have shown that \eqref{eq5:objective} is equivalent to $\mathbf{O} = \tilde{\mathbf{F}} \tilde{\mathbf{\Psi}} \tilde{\mathbf{F}}'$. Furthermore, $\tilde{\mathbf{F}} \tilde{\mathbf{\Psi}} \tilde{\mathbf{F}}'$ is a matrix of rank $\bar{r}$ by assumption. Then, it must be the case that $\mathbf{O}$ is also a symmetric PSD matrix of rank $\bar{r}$ in order for \eqref{eq5:objective} to hold. By Lemma \ref{lemma:1} we know that if $\mathbf{O}$ is of rank $\bar{r}$ then all minors of $\mathbf{O}$ associated with submatrices of dimension $k>\bar{r}$ must be equal to zero.
\end{proof}
Furthermore, we can show that $\mathbf{O}$ can be decomposed in the following way.
\begin{claim} \label{claim:1}
	$\mathbf{O}$ can be decomposed into $\mathbf{O} = \mathbf{\Omega}-\tilde{\mathbf{D}}+(\alpha-\tilde{\alpha})\mathbf{J}(\tilde{\alpha}, \theta^0)$, where $\mathbf{J}(\tilde{\alpha}, \theta^0)$ is given by:
	\begin{equation}\label{eq9:J}
		\mathbf{J}(\tilde{\alpha},\theta^0) = \mathbf{L}\mathbf{\Omega}+\mathbf{\Omega}\mathbf{L}'+(\alpha-\tilde{\alpha})\mathbf{L}\mathbf{\Omega}\mathbf{L}'.
	\end{equation}
\end{claim}

The proof of this result consists only of standard matrix operations, and is available in the Online Appendix. It follows that the matrix $\mathbf{O}$ is such that its off-diagonal elements depend only on the true structural parameters $\theta^0$ and $\tilde{\alpha}$ (since $\tilde{\mathbf{D}}$ is diagonal). In particular, the off-diagonal elements of $\mathbf{O}$ can be written as follows:
\begin{corollary} \label{coro:1}
	Every off-diagonal element of $\mathbf{O}$ can be written as
	\[
	O_{r,c} = \Omega_{r,c} + (\alpha - \tilde{\alpha}) J_{r,c}, \qquad r \neq c,
	\]
	where $J_{r,c}$ denotes the $(r,c)$-th element of $\mathbf{J}(\tilde{\alpha}, \theta^0)$.
\end{corollary}
This result follows directly from the decomposition of $\mathbf{O}$ in Claim~\ref{claim:1} and the fact that $\tilde{\mathbf{D}}$ is diagonal, so $\tilde{D}_{r,c} = 0$ for $r \neq c$. Therefore, by looking at the off-diagonal elements of $\mathbf{O}$ we can derive conditions that impose restrictions on $\tilde{\alpha}$, which permits us to isolate the identification problem of $\alpha$ from the identification problem of the rest of the parameters. To do so, we use information contained in the \textit{diagonal exclusion minors} of $\mathbf{O}$, which are formally defined as follows:

\begin{definition}
	Let $\mathbf{A}$ be a $T \times T$ matrix. For index sets $R=\{r_1,\dots,r_k\}$ and $C=\{c_1,\dots,c_k\}$ with $|R|=|C|=k$ and $R \cap C=\emptyset$, let $\mathbf{M}_{R,C}^{\mathbf{A},k}$ denote the $k \times k$ submatrix of $\mathbf{A}$ formed by the rows in $R$ and the columns in $C$. We call $\mathbf{M}_{R,C}^{\mathbf{A},k}$ a \textit{diagonal exclusion submatrix}, and its determinant $\det\left(\mathbf{M}_{R,C}^{\mathbf{A},k}\right)$ a \textit{diagonal exclusion minor}.
\end{definition}

Since $|R|=|C|=k$ and $R \cap C=\emptyset$ then $2k \leq T$, or equivalently, $k \leq T/2$. We now derive properties of these minors that are helpful for identification of $\alpha$. The following propositions characterize some properties of diagonal exclusion minors in this problem. Proofs are available in Appendix \ref{Ap:A}.

\begin{proposition}\label{prop:2} Every diagonal exclusion minor of $\mathbf{O}$ associated with a $k\times k$ diagonal exclusion submatrix can be written as the sum of the corresponding diagonal exclusion minor for $\mathbf{\Omega}$
	and a term
	$(\alpha-\tilde{\alpha})\tilde{J}_{R,C}(\tilde{\alpha},\theta^0)$:
	\begin{equation}\label{eq10:theo3}
		\det\left(\mathbf{M}^{\mathbf{O},k}_{R,C}\right)
		=
		\det\left(\mathbf{M}^{\mathbf{\Omega},k}_{R,C}\right)
		+
		(\alpha-\tilde{\alpha})
		\tilde{J}_{R,C}\left(\tilde{\alpha},\theta^0\right)
	\end{equation}
	where $\tilde{J}_{R,C}(\tilde{\alpha},\theta^0)$ is a function of the true
	parameters $\theta^0$ and $\tilde{\alpha}$, defined for $k=1$ by
	\[
	\tilde J_{\{r\},\{c\}}(\tilde\alpha,\theta^0)
	=
	J_{r,c}(\tilde\alpha,\theta^0),
	\qquad r\neq c,
	\]
	and, for $k\geq2$, defined recursively as follows:
	\begin{equation}\label{eq11:Jtilde}
		\begin{split}
			\tilde{J}_{R,C}(\tilde{\alpha}, \theta^0)
			&=
			\sum_{j=1}^{k}(-1)^{k+j}
			\Bigl[
			J_{r_k,c_j}
			\det\left(\mathbf{M}^{\mathbf{\Omega},k-1}_{R-r_k,C-c_j}\right) \\
			&\qquad
			+
			\Omega_{r_k,c_j}
			\tilde{J}_{R-r_k,C-c_j}(\tilde{\alpha},\theta^0)
			+
			(\alpha-\tilde{\alpha})
			J_{r_k,c_j}
			\tilde{J}_{R-r_k,C-c_j}(\tilde{\alpha},\theta^0)
			\Bigr].
		\end{split}
	\end{equation}
\end{proposition}
	In \eqref{eq11:Jtilde}, $\mathbf{M}^{\mathbf{\Omega},k}_{R,C}$ is the diagonal exclusion submatrix of $\mathbf{\Omega}$ associated with the same rows and columns. $J_{r_k,c_j}$ is the $(r_k,c_j)$-th element of the matrix $\mathbf{J}(\tilde{\alpha}, \theta^0)$, and $\Omega_{r_k,c_j}$ is the $(r_k,c_j)$-th element of the matrix $\mathbf{\Omega}$.\footnote{$R-r_{k}$ means removal of $r_{k}$ from the set $R$, and $C-c_j$ means removal of $c_j$ from the set $C$. They correspond to the usual notation $R\setminus \{r_{k}\}$ and $C\setminus \{c_j\}$.} Note that the value of $\tilde{J}_{R,C}\left(\tilde{\alpha},\theta^0\right)$ depends on the selected rows and columns from the original matrix so in general, $\tilde{J}_{R,C}\left(\tilde{\alpha},\theta^0\right) \neq \tilde{J}_{R',C'}\left(\tilde{\alpha},\theta^0\right)$ for any other valid set of rows and columns $R'$ and $C'$.

\begin{proposition} \label{prop:3}
	Let $\mathbf{M}^{\mathbf{\Omega},k}_{R,C}$ be a diagonal exclusion submatrix of $\mathbf{\Omega}$ of dimension $k>\bar{r}$. Then $\det\left(\mathbf{M}^{\mathbf{\Omega},k}_{R,C}\right)=0$.
\end{proposition}

This result follows from the fact that $\mathbf{D}$ is diagonal and that $\mathbf{F\Psi F}'$ is of rank $\bar{r}<k$. It implies that when considering diagonal exclusion minors associated with submatrices of dimension $k>\bar r$ we only need to consider the term $(\alpha-\tilde\alpha)\tilde{J}_{R,C}\left(\tilde{\alpha},\theta^0\right)$. Moreover, we can show that $\tilde{J}_{R,C}\left(\tilde{\alpha},\theta^0\right)$ is a polynomial in $\tilde{\alpha}$. In particular, under the assumptions of the model:

\begin{proposition}\label{prop:4} Consider any diagonal exclusion minor of $\mathbf{O}$ associated with a submatrix of dimension $k$. Then, under Assumptions \ref{as:1}-\ref{as:6}, $\tilde{J}_{R,C}\left(\tilde{\alpha},\theta^0\right)$ is a polynomial in $\tilde{\alpha}$ whose degree is at most $2k-1$.
\end{proposition}

Moreover, when $k=\bar r+1$ we can also show that $\tilde{J}_{R,C}\left(\tilde{\alpha},\theta^0\right)$ is a nonzero polynomial in $\tilde\alpha$ except for a knife-edge case:

\begin{proposition}\label{prop:5} Consider any diagonal exclusion minor of $\mathbf{O}$ associated with a submatrix of dimension $\bar{r}+1$. Then, under Assumptions \ref{as:1}-\ref{as:6}, the following statements hold:
	
	\begin{enumerate}
		\item If either $
		\alpha\neq 0$ or $
		\min_{r\in R,\ c\in C}|r-c|=1$
		then
		$\tilde J_{R,C}(\tilde\alpha,\theta^0)$
		is a nonzero polynomial in $\tilde\alpha$ almost surely.
		
		\item If $
		\alpha=0 \;\text{and}\;
		\min_{r\in R,\ c\in C}|r-c|>1$
		then $\tilde J_{R,C}(\tilde\alpha,\theta^0)\equiv0$
		as a polynomial in $\tilde\alpha$.
	\end{enumerate}
	
\end{proposition}

Combining Propositions \ref{prop:1}-\ref{prop:5} we can write the determinant of every diagonal exclusion submatrix of $\mathbf{O}$ of dimension $\bar{r}+1$ as:
\[\det\left(\mathbf{M}^{\mathbf{O},\bar{r}+1}_{R,C}\right) = (\alpha-\tilde{\alpha})\tilde{J}_{R,C}\left(\tilde{\alpha},\theta^0\right)=0\]
where under Case 1 of Proposition \ref{prop:5} $\tilde{J}_{R,C}\left(\tilde{\alpha},\theta^0\right)$ is a nonzero polynomial in $\tilde{\alpha}$ of degree at most $2\bar{r}+1$. This creates a system of equations (the identity holds for all valid combinations of $R$ and $C$) that any $\tilde{\alpha}$ must satisfy in order to be a candidate to solve \eqref{eq5:objective}. These equations permit us to isolate the identification problem of $\alpha$ from the rest of the parameters, as we are able to derive conditions that only depend on $\tilde{\alpha}$ and $\theta^0$.

\subsection{Proof of Theorem \ref{thm:1}}\label{sbs:33}

In this Subsection, we build on the intermediate results derived in Subsection \ref{sbs:32} to prove Theorem \ref{thm:1}. Specifically, we demonstrate that any alternative $\tilde\theta$ solving \eqref{eq5:objective} must satisfy $\tilde{\alpha}=\alpha$, except for a set of factor realizations that occur with probability zero. We provide a formal proof of this result and then proceed with a detailed argument for the case of a single factor to illustrate the reasoning.

\begin{proof}
Assume $\alpha\neq0$.\footnote{The case of $\alpha=0$ will be discussed at the end of the proof.} Let $\mathcal T=\{1,\ldots,T\}$ denote the set of time indices. Define
\[
\mathcal I
=
\left\{
(R,C)\in 2^{\mathcal T}\times 2^{\mathcal T}
:
R\cap C=\emptyset,\quad
|R|=|C|=\bar r+1
\right\}
\]
as the set of row and column indices $(R,C)$ forming a diagonal exclusion submatrix of dimension $\bar{r}+1$. For the minors associated with these submatrices, Proposition \ref{prop:5} (given that $\alpha\neq0$) implies that $\tilde{J}_{R,C}(\tilde{\alpha},\theta^0)$ is almost surely a nonzero polynomial in $\tilde{\alpha}$ whose coefficients depend only on $\theta^0$. Now, pick any diagonal exclusion minor of $\mathbf{O}$ associated with a $(\bar r+1)\times(\bar r+1)$ submatrix such that $(R,C)\in\mathcal I$. By Proposition \ref{prop:1} it must be the case that:
	\[\det\left(\mathbf{M}^{\mathbf{O},\bar{r}+1}_{R,C}\right)=0\]
	By Proposition \ref{prop:2} we know that
	\[\det\left(\mathbf{M}^{\mathbf{O},\bar{r}+1}_{R,C}\right) =\det\left(\mathbf{M}^{\mathbf{\Omega},\bar{r}+1}_{R,C}\right) +(\alpha-\tilde{\alpha})\tilde{J}_{R,C}\left(\tilde{\alpha},\theta^0\right)=0\]
	Moreover, by Proposition \ref{prop:3} we have that $\det\left(\mathbf{M}^{\mathbf{\Omega},\bar{r}+1}_{R,C}\right)=0$. Thus,	
	\begin{equation}\label{eq12:proof}
		\det\left(\mathbf{M}^{\mathbf{O},\bar{r}+1}_{R,C}\right)= (\alpha-\tilde{\alpha})\tilde{J}_{R,C}\left(\tilde{\alpha},\theta^0\right)=0
	\end{equation}
	Assumptions \ref{as:1} and \ref{as:3} ensure that there are more than two diagonal exclusion
	minors. Thus, the same argument applies to any such minor. Therefore, for every $(R',C')\in\mathcal I$,
	\begin{equation}\label{eq13:alt}
		\det\left(\mathbf{M}^{\mathbf{O},\bar{r}+1}_{R',C'}\right)= (\alpha-\tilde{\alpha})\tilde{J}_{R',C'}\left(\tilde{\alpha},\theta^0\right)=0,\qquad \forall (R',C')\in\mathcal{I}
	\end{equation}
	Equations \eqref{eq12:proof}-\eqref{eq13:alt} define a system of polynomials that every $\tilde{\alpha}$ must satisfy in order for $\tilde{\theta}$ to be a candidate parameterization solving \eqref{eq5:objective}. Thus, either $\tilde{\alpha}=\alpha$ (meaning that $\alpha$ is identified) or $\tilde{J}_{R,C}\left(\tilde{\alpha},\theta^0\right)=\tilde{J}_{R',C'}\left(\tilde{\alpha},\theta^0\right)=0$ for every $(R',C')\in\mathcal{I}$. Global identification of $\alpha$ can fail only in the latter case. Thus, to show almost sure global identification of $\alpha$ it is sufficient to prove that there exist two different combinations of row and column indices $(R_0,C_0),(R_1,C_1) \in \mathcal{I}$ such that $\tilde{J}_{R_0,C_0}(\tilde\alpha,\theta^0)$ and $\tilde{J}_{R_1,C_1}(\tilde\alpha,\theta^0)$ do not share any root almost surely. The proposition below shows that this is the case when $\alpha\neq0$:
\begin{proposition}\label{prop:6}
	Define
	\[
	R_0=\{1,\ldots,\bar r+1\},\quad C_0=\{\bar r+2,\ldots,2\bar r+2\},\quad
	R_1=\{2,\ldots,\bar r+2\},\quad C_1=\{1,\bar r+3,\ldots,2\bar r+2\}.
	\]
	Fix any admissible $(\alpha,\mathbf D,\mathbf\Psi)$ with $\alpha\neq0$. Under Assumptions \ref{as:1}-\ref{as:5} the resultant
	\[
	\mathcal R_{\alpha,\mathbf D,\mathbf\Psi}(\mathbf F_2)
	=
	\operatorname{Res}_{\tilde\alpha}\{\tilde{J}_{R_0,C_0}(\tilde\alpha,\theta^0), \tilde{J}_{R_1,C_1}(\tilde\alpha,\theta^0)\}
	\]
	is a nonzero polynomial in the unrestricted factors $\mathbf F_2$. Then, if Assumption \ref{as:6} holds,
	\[
	\Pr\left(
	\exists\,\tilde\alpha\in\mathbb R:
	\tilde J_{R_0,C_0}(\tilde\alpha,\theta^0)=0
	\quad\text{and}\quad
	\tilde J_{R_1,C_1}(\tilde\alpha,\theta^0)=0
	\right)=0.
	\]
\end{proposition}

A proof of this Proposition is provided in the Online Appendix. This result is sufficient to complete the proof of Theorem \ref{thm:1} for $\alpha\neq 0$. To see this, note that Proposition \ref{prop:6} guarantees that we can construct indices $(R_0,C_0),(R_1,C_1)\in\mathcal I$ such that the resultant between $\tilde{J}_{R_0,C_0}(\tilde\alpha,\theta^0)$ and $\tilde{J}_{R_1,C_1}(\tilde\alpha,\theta^0)$ is a nonzero polynomial in the unrestricted entries of the factors $\mathbf F_2$ for any fixed $(\alpha,\mathbf D,\mathbf\Psi)$ with $\alpha\neq0$.\footnote{A common root implies a resultant equal to zero.} Therefore, the event where $\tilde{J}_{R_0,C_0}(\tilde\alpha,\theta^0)$ and $\tilde{J}_{R_1,C_1}(\tilde\alpha,\theta^0)$ share a root is contained in the zero set of a nonzero polynomial in $\mathbf F_2$ (defined by $\mathcal R_{\alpha,\mathbf D,\mathbf\Psi}(\mathbf F_2)=0$). By Lemma \ref{lemma:3}, this event has Lebesgue measure zero in the space of unrestricted factor coordinates $\mathbf F_2$. Additionally, under Assumption \ref{as:6} Lemma \ref{lemma:2} implies that this event occurs with probability zero. Hence $\alpha$ is almost surely globally identified when $\alpha\neq0$.
	
Showing almost sure global identification when $\alpha=0$ is much simpler and can be done without the use of diagonal exclusion minors. A direct proof is provided in Appendix \ref{Ap:A}.

Combining both cases implies that $\alpha$ is almost surely globally identified.
	\end{proof}
	
\begin{remark}
	Proposition~\ref{prop:6} uses only two of the polynomial restrictions implied by
	\eqref{eq12:proof}--\eqref{eq13:alt}. If an alternative value
	$\tilde\alpha\neq\alpha$ were to solve \eqref{eq5:objective}, then the same
	$\tilde\alpha$ would have to satisfy
	\[
	\tilde J_{R,C}(\tilde\alpha,\theta^0)=0
	\qquad\text{for every }(R,C)\in\mathcal I.
	\]
Thus, the zero set of the resultant in Proposition~\ref{prop:6} provides only
an outer bound on the exceptional set of factor realizations relevant for the
identification argument. The actual exceptional set when $\alpha\neq0$ is contained in a smaller,
more restrictive set: the set of factor realizations for which all polynomial
restrictions associated with $\mathcal I$ vanish at the same
$\tilde\alpha$.
\end{remark}

Having proved Theorem \ref{thm:1} for the general case with $\bar r$ factors,
we now illustrate the procedure in the simpler case $\bar r=1$ to provide
intuition for the argument.
	
	
	\textit{The case of a single factor $\bar{r}=1$:} Consider the diagonal exclusion minor of $\mathbf{O}$ associated with rows $R=(1,2)$ and columns $C=(3,4)$. A direct computation yields:
	\begin{equation}\label{eq13:detr1}
		\det\left(\mathbf{M}_{(1,2),(3,4)}^{\mathbf{O},2}\right) =\left(\alpha-\tilde{\alpha}\right)\Psi(d_2-\alpha d_1 f_2)\left(f_3\tilde{\alpha}-f_4\right)=0
	\end{equation}
	where in this case $\tilde{J}_{(1,2),(3,4)}\left(\tilde{\alpha},\theta^0\right)$ is a polynomial in $\tilde{\alpha}$ of degree at most 1. Note that one of the following four alternatives must be true in order for \eqref{eq13:detr1} to be satisfied:
	
	\begin{case*}\label{case:1}
		$\Psi=0$
	\end{case*}
	
	\begin{case*}\label{case:2}
		$f_2 d_1 \alpha=d_2$
	\end{case*}
	
	\begin{case*}\label{case:3}
		$\tilde{\alpha}=\frac{f_4}{f_3}$
	\end{case*}
	
	\begin{case*}\label{case:4}
		$\tilde{\alpha}=\alpha$
	\end{case*}
	
	Cases \ref{case:1} and \ref{case:2} correspond to the case in which $\tilde{J}_{(1,2),(3,4)}\left(\tilde{\alpha},\theta^0\right)$ is the zero polynomial, so that all the coefficients are trivially equal to zero. Case \ref{case:3} corresponds to the case in which we can find a root $\tilde{\alpha}^{\star}$ that solves $\tilde{J}_{(1,2),(3,4)}(\tilde{\alpha}^{\star},\theta^0)=0$. However, once $\tilde{\alpha}$ is pinned-down in this way, we have no other free parameter, so the same root must satisfy all the zero-determinant conditions (Proposition \ref{prop:1}) associated with the remaining diagonal exclusion minors of $\mathbf{O}$. Case \ref{case:4} corresponds to the case where identification of $\alpha$ is achieved.
	
	We show that Cases \ref{case:1}-\ref{case:3} are only possible in a set of factor realizations that occur with probability zero. Thus, with probability one, every admissible solution to \eqref{eq5:objective} must fall into Case~\ref{case:4}, and therefore must satisfy $\tilde{\alpha}=\alpha$. Hence $\alpha$ is almost surely globally identified.
	
	\textbf{Case \ref{case:1}}: $\Psi=0$
	
	This case is obvious. We assumed $\Psi>0$, so we can immediately rule this out.
	
	\textbf{Case \ref{case:2}}: $f_2 d_1 \alpha=d_2$
	
	If $\alpha=0$, the term $d_2-\alpha d_1 f_2$ equals $d_2>0$, so this case cannot arise. Then, assume $\alpha\neq 0$. By Assumption \ref{as:6}, $f_2$ is drawn from a continuous probability distribution. Thus, we can invoke Lemma \ref{lemma:2} and ensure that
	$P\left(f_2 = \frac{d_2}{d_1 \alpha}\right) = 0$. Then, Case \ref{case:2} can only occur in a probability-zero set.
	
	\textbf{Case \ref{case:3}}: $\tilde{\alpha}=\frac{f_4}{f_3}$
	
Let $\tilde{\alpha}^{\star}=f_4/f_3$ when $f_3\neq0$.\footnote{The event $f_3=0$ is an event of probability zero under Assumption \ref{as:6}. Moreover, when $f_3=0$, then the factor $f_3\tilde\alpha-f_4$ reduces to $-f_4$, which is also nonzero almost surely under Assumption~\ref{as:6}.} This value makes the restriction in \eqref{eq13:detr1} hold, but it does not by itself imply observational equivalence. For Case \ref{case:3} to generate a failure of global identification, the same value $\tilde{\alpha}^{\star}$ must also make every other admissible diagonal-exclusion minor vanish. We therefore evaluate another diagonal-exclusion restriction at $\tilde{\alpha}^{\star}$. Consider the minor associated with rows $R=(2,3)$ and columns $C=(1,4)$. Proceeding as in \eqref{eq12:proof} and evaluating at $\tilde{\alpha}=f_4/f_3$, we obtain
	\begin{align*}\det\left(\mathbf{M}_{(2,3),(1,4)}^{\mathbf{O},2}\right)&=\left(\alpha-\frac{f_4}{f_3}\right)\Bigg(\alpha d_1 d_3 +d_{1}f_{2}f_{3}\Psi\alpha^2-d_{2}\Psi\alpha f_{3} +\alpha d_{3} \Psi + d_{3}\Psi f_{2} +\alpha d_{1}f_{3}^{2} \Psi-2\alpha d_1 f_2 f_4 \Psi   \\
	&-\frac{f_{4}}{f_{3}}\left(d_{1}d_{3}+d_{3}\Psi+\alpha^{2}d_{1}d_{2}+\alpha^2 d_{2} \Psi + \alpha^{2} d_{1} f_{2}^{2}\Psi\right)+\frac{f_{4}^{2}}{f_{3}^{2}}\alpha\left(d_1\Psi f_{2}^{2}+d_{1}d_{2}+d_{2}\Psi\right)\Bigg)=0
	\end{align*}	
If the first factor is zero, then $\tilde{\alpha}^{\star}=\alpha$, so
there is no failure of global identification of $\alpha$. Hence, under the
alternative $\tilde{\alpha}^{\star}\neq\alpha$, Case \ref{case:3} requires
the second factor to be zero.

Since $f_3\neq 0$ in Case \ref{case:3}, we may multiply the second factor by
$f_3^2$. The resulting expression is a polynomial in
$(f_2,f_3,f_4)$. This polynomial is not identically zero. Indeed, the
monomial $f_2 f_3^2$ appears with coefficient $d_3\Psi$, which is
strictly positive because $d_3>0$ and $\Psi>0$. Therefore, the condition
that the second factor equals zero defines the zero set of a nonzero
polynomial in $(f_2,f_3,f_4)$. By Lemma \ref{lemma:3} this event has Lebesgue measure zero in the unrestricted factor space. Then, by Assumption \ref{as:6} and Lemma
\ref{lemma:2}, this event has probability zero. Consequently, Case \ref{case:3} is ruled out except on a probability-zero set.
	
Having ruled out Cases \ref{case:1}-\ref{case:3} outside probability-zero sets, any admissible solution to \eqref{eq5:objective} must fall into Case~\ref{case:4} almost surely. Hence any admissible observationally equivalent parameterization must satisfy $\tilde{\alpha}=\alpha$ almost surely.\hfill $\square$

\subsection{Global identification of $\mathbf{F}$, $\mathbf{\Psi}$, and $\mathbf{D}$}\label{sbs:34}

We have shown that $\alpha$ is almost surely globally identified: any admissible
$\tilde\theta$ satisfying \eqref{eq5:objective} must satisfy
$\tilde\alpha=\alpha$, except on a set of factor realizations with probability zero. We can then conclude that all the parameters in our model are almost surely globally identified. This is stated in Theorem \ref{thm:2}. We next provide a proof.

\begin{proof}
	By Theorem~\ref{thm:1}, $\alpha$ is almost surely globally identified. Hence
	$\tilde{\alpha}=\alpha$ with probability one and
	$\tilde{\mathbf{\Gamma}}^{-1}\mathbf{\Gamma}=\mathbf{I}_{T}$. The observational
	equivalence condition \eqref{eq5:objective} therefore reduces to
	\[
	\mathbf{F}\mathbf{\Psi}\mathbf{F}'+\mathbf{D}
	=
	\tilde{\mathbf F}\tilde{\mathbf \Psi}\tilde{\mathbf F}'
	+\tilde{\mathbf D}.
	\]
	This is the standard factor-analytic covariance decomposition considered in
	\cite{anderson1956statistical}. In their notation, the factor loading matrix is
	$\mathbf{\Lambda}=\mathbf{F}\mathbf{\Psi}^{1/2}$, so that
	$\mathbf{F}\mathbf{\Psi}\mathbf{F}'=\mathbf{\Lambda}\mathbf{\Lambda}'$.
	Since $\mathbf{\Psi}$ is positive definite, multiplication by
	$\mathbf{\Psi}^{1/2}$ preserves the ranks of all row submatrices. Then, it is enough to verify \cite{anderson1956statistical}'s row-deletion
	rank condition for $\mathbf F$ (Theorem 5.1), namely that after deleting
	any row of $\mathbf F$, the remaining matrix contains two disjoint
	$\bar r\times \bar r$ full-rank submatrices.
	
	Under the normalization $\mathbf F=(\mathbf I_{\bar r},\mathbf F_2')'$ and
Assumption~\ref{as:3}, after deleting any row of $\mathbf F$ there remain enough
rows to select two disjoint $\bar r\times \bar r$ submatrices. If no row of the identity block is deleted, one such submatrix is $\mathbf I_{\bar r}$ and the other is an unrestricted
$\bar r\times\bar r$ block from $\mathbf F_2$. If one row of the identity block is deleted, one may combine the remaining identity rows with one unrestricted row and form a second disjoint unrestricted $\bar r\times\bar r$ block using the remaining unrestricted entries of $\mathbf{F}_{2}$. In each case, the relevant determinants are nonzero
polynomials in $\operatorname{vec}(\mathbf F_2)$. Hence \cite{anderson1956statistical}'s row-deletion rank condition fails only on a proper algebraic subset of the factor space, holding
for almost all configurations of the factors, given the remaining parameters. By
Assumption~\ref{as:6}, $\operatorname{vec}(\mathbf F_2)$ is absolutely continuous; therefore, the row-deletion rank condition holds almost surely.

Therefore, by Theorem 5.1 of \cite{anderson1956statistical}, $\mathbf D$ and the common component $\mathbf{F}\mathbf{\Psi}\mathbf{F}'$ are globally identified almost surely. The normalization $\mathbf F=(\mathbf I_{\bar r},\mathbf F_2')'$ fixes the rotation, meaning that $\mathbf{F}$ and $\mathbf{\Psi}$ are separately identified. Thus all covariance parameters are almost surely globally identified.

\end{proof}

\section{Individual fixed-effects}\label{sec4}

Consider a panel data model with time, individual and interactive effects:
\begin{equation}\label{eq16}
	y_{it} = \alpha y_{it-1}+\delta_t + \gamma_i +\lambda_{i}' f_t + \varepsilon_{it}
\end{equation}
where $\gamma_i$ denotes the individual fixed-effect for observation $i$, and the rest of the parameters are defined as in \eqref{eq1:dgp}. It is a well-known result that individual fixed-effects can be absorbed into the factor structure by treating $\gamma_i$ as the factor loadings of a special factor, which takes the value of 1 for all $t$. As a result, the factor matrix does not satisfy the conditions of Assumption \ref{as:6} because the factor associated with the individual fixed-effects follows a degenerate distribution.

Therefore, the almost sure interpretation of the results derived in Theorem \ref{thm:1} does not apply directly in this case and characterizing the subset of the parameter space where identification may fail is more challenging. We can overcome this issue by showing that our proof can be adapted to allow for one factor with a degenerate distribution. In practice, almost sure global identification requires that there is enough variation in some of the factors, but some factors may have degenerate distributions. It is important to note that in this case we do not estimate $\gamma_i$ (doing so would cause incidental parameter problems), but we absorb it into the factor structure and identify the variance of $\gamma_i$ and the covariance with the rest of the factor loadings through $\mathbf{\Psi}$.

We also show that the diagonal-exclusion approach is useful beyond the interactive-effects setting. In particular, we show that the autoregressive coefficient in a dynamic two-way fixed-effects model can be globally identified exploiting diagonal exclusion minors, even in the unit root case. In this case, no genericity argument is needed.

\subsection{Identification with individual and interactive fixed effects}\label{sbs:41}

Start from model \eqref{eq16} and absorb the individual fixed effects into the factor structure:
\begin{equation}\label{eq17}
	y_{it} = \alpha y_{it-1}+\delta_t + \eta_{i}' F_t + \varepsilon_{it}
\end{equation}
where $\eta_i=[\gamma_i,\lambda_i']'$, $F_t=[1,f_t']'$, and the remaining terms are defined as in \eqref{eq1:dgp}. 
Using the same procedure as in Section \ref{sec2.1}, we project the first observation
$ y_{i1} $ onto the vector $[1, \gamma_i, \lambda_i]$, yielding the following expression:
\[ y_{i1} = \delta_1^{*} + \gamma_i f_{\gamma}^{*} + \lambda_i' f_{1}^{*} + \varepsilon_{i1}^{*} \]
where $ (\delta_1^{*}, f_{\gamma}^{*}, f_1^{*}) $ are the projection coefficients and
$ \varepsilon_{i1}^{*} $ is the projection residual. For notational simplicity, we drop the asterisks from the coefficients and from the residual, as these are free parameters, and we do not require the error terms $ \varepsilon_{it} $ to have identical distributions across $ t $.
Therefore, for $ t = 1 $, we can rewrite the equation as $y_{i1} = \delta_1 + \gamma_i f_{\gamma} + \lambda_i' f_1 + \varepsilon_{i1}$.

Combining this with equation \eqref{eq17} and stacking the observations
over $ t $, we obtain:
\[\mathbf{B}y_{i} = \delta + \mathbf{F}\eta_{i} + \varepsilon_{i}\]
where $\mathbf{B}$, $y_i$, $\delta$, and $\varepsilon_{i}$ are all defined earlier, but
\[\mathbf{F}' = \begin{bmatrix}
	f_\gamma & 1 & \cdots & 1 \\
	f_{1}^{'} & f_{2}^{'} & \cdots & f_{T}^{'}
\end{bmatrix}, \quad \text{ and }  \eta_{i} = \begin{bmatrix}
	\gamma_i\\
	\lambda_{i}
\end{bmatrix}\]
The first element in the first column of $ \mathbf{F} $ is $ f_\gamma $, not 1, so it is a free parameter arising from the projection of the initial observation.
Also, the first column of $\mathbf{F}$ contains many fixed
entries equal to $1$.
Consequently, the factor matrix does not satisfy the
conditions of Assumption~\ref{as:6}. The argument below shows that, despite the degenerate fixed-effect factor, $\alpha$ remains almost surely globally identified outside a lower-dimensional exceptional subset of the parameter space.

For simplicity, we provide a proof for the almost sure global identification of $\alpha$ in the presence of individual fixed-effects and a single unrestricted factor ($\dim(f_t)=1$), with $f_t$ drawn from a continuous distribution. In this case, $\bar{r}=2$. However, the proof can be extended to more general cases encompassing multiple factors.

Observe that the presence of individual fixed-effects imposes structure on $\mathbf{F}$, so we require fewer restrictions to eliminate rotational indeterminacy. More specifically, we only need to impose two restrictions on $\mathbf{F}$ to prevent indeterminacy. We shall impose the normalization in the following way, referred to as tail normalization:
\begin{equation}\label{eq18}
	\mathbf{F}' = \begin{bmatrix}
		f_\gamma & 1 & \cdots & 1\ & 1 \; & 1 \\
		f_1 & f_2 & \cdots & f_{T-2}\ & 0 \; & 1
	\end{bmatrix}
\end{equation}
This normalization restricts the last two realizations of the factor $f$ to be 0 and 1, respectively. This is purely for convenience in terms of exposition. Normalizing other entries (except the first row) has the same effect.
Under this normalization, $\mathbf F \mathbf{A}= \mathbf F $, if and only if $\mathbf{A}=\mathbf{I}_2$, thus eliminating rotational indeterminacy.		

Before showing that $\alpha$ can be almost surely identified in this context, we shall explicitly state the identification assumptions:

\begin{assumptionprime}\label{asp:1}
	There is one factor associated with the individual fixed-effects and an unrestricted factor ($\bar{r}=2$).
\end{assumptionprime}

\begin{assumptionprime}\label{asp:2}
	$\mathbf{\Psi}$ is an unrestricted symmetric and PD matrix, and $\mathbf{F}$ is normalized as shown in \eqref{eq18}.
\end{assumptionprime}

\begin{assumptionprime}\label{asp:3}
	$T \geq 6$.	
\end{assumptionprime}

\begin{assumptionprime}\label{asp:4}
	
	For each $t$, the $\varepsilon_{it}$ in \eqref{eq17} are i.i.d.\ over $i$. For each $i$, they satisfy the MDC with respect to the filtration ${\cal F}_{it}$ defined earlier (replacing $\lambda_i$ with $\eta_i)$, and $\mathbb{V}(\varepsilon_{it})=d_t>0$.
	
\end{assumptionprime}

\begin{assumptionprime}\label{asp:5}
	The factor loadings $\eta_{i}$ are either:
	\begin{enumerate}
		\item Random coefficients satisfying $\E[\eta_i]=0$ and $\mathbf{\Psi}_{N} \xrightarrow{p} \mathbf{\Psi}$ as $N\to\infty$.
		\item Deterministic coefficients satisfying $\mathbf{\Psi}_{N} \to \mathbf{\Psi}$ as $N \to \infty$.
	\end{enumerate}
\end{assumptionprime}

\begin{assumptionprime}\label{asp:6} The unrestricted factor coordinates $f_{1},\dots,f_{T-2}$ are drawn as an absolutely continuous random vector with respect to Lebesgue measure on $\mathbb{R}^{(T-2)}$.
\end{assumptionprime}

Assumptions \ref{asp:1}-\ref{asp:6} are an adaptation of the ones outlined in Section \ref{sec3} for the particular case we are studying. As in Section~\ref{sec3}, Assumption~\ref{asp:6} is used only to express a genericity condition in probabilistic language. The factors are still treated as fixed population parameters; the absolutely continuous distribution is a device for stating that the exceptional factor realizations under which almost sure global identification may not hold are contained in a Lebesgue-measure-zero set of the unrestricted factor space. Under these assumptions, we can show that $\alpha$ is almost surely globally identified for all parameter values outside a Lebesgue-null subset of the parameter space:

\begin{theorem}\label{thm:3}
	Let $y_{it}$ be generated from \eqref{eq16}. Under Assumptions
	\ref{asp:1}--\ref{asp:6}, $\alpha$ is almost surely globally identified
	for all parameter values outside a Lebesgue-null subset
	$\Theta^{*}_{(1)}\subset\Theta$.
\end{theorem}

The proof of Theorem~\ref{thm:3} follows a similar strategy to that of Theorem~\ref{thm:1} in the case of a single factor. The full argument and the characterization of $\Theta^{*}_{(1)}$ are provided in the Online Appendix. Theorem \ref{thm:3} does not state that the region $\Theta^{*}_{(1)}\subset\Theta$ necessarily corresponds to an underidentification region; it simply states that our proof strategy cannot guarantee identification in that region, but it might still be possible to establish identification using a different strategy. Moreover, note that Theorem \ref{thm:3} implies generic identification with respect to the full parametric space $\Theta$.

Once identification of $\alpha$ is established, it follows that the remaining parameters are also almost surely globally identified outside the same exceptional set. This result is formalized in the following theorem.

\begin{theorem}\label{thm:4}
	Let $y_{it}$ be generated from \eqref{eq16}. Under Assumptions
	\ref{asp:1}--\ref{asp:6}, the parameters $\{d_t\}_{t=1}^{T}$,
	$\mathbf{\Psi}$, and $\mathbf F$ are almost surely globally
	identified for all parameter values outside a Lebesgue-null subset
	$\Theta^{*}_{(1)}\subset\Theta$.
\end{theorem}

\begin{proof}
	The proof is identical to that of Theorem \ref{thm:2}, replacing
	Theorem \ref{thm:1} by Theorem \ref{thm:3} and using the tail
	normalization \eqref{eq18}, under which \cite{anderson1956statistical}'s row-deletion rank
	condition holds almost surely.
\end{proof}

\subsection{Dynamic panel with additive fixed effects}

Consider a dynamic panel with additive fixed effects:
\begin{equation}\label{eq23}
	y_{it} \;=\; \alpha\, y_{it-1} \;+\; \gamma_i \;+\;\delta_t\;+\; \varepsilon_{it}, \qquad t \ge 2,
\end{equation}
where $\varepsilon_{it}$ are i.i.d.\ across $i$, and satisfy the MDC with respect to the filtration:
\[
\mathcal F_{i0}
=
\sigma(\gamma_i,y_{i0}), \quad
\mathcal F_{it}
=
\sigma(\gamma_i,y_{i0},\varepsilon_{i1},\ldots,\varepsilon_{it}), \text{ for } t\ge 1,
\]
and $\mathbb{V}(\varepsilon_{it})=d_t$. The rest of the variables are defined as in \eqref{eq1:dgp}.

This canonical model nests some of the most influential papers in the dynamic panel data literature. For instance, \cite{arellano1991some} consider a simplified version of this model with no time fixed effects. However, the restriction $|\alpha|<1$ is imposed. Subsequent work argues that $\alpha$ is not identified in the unit-root case under conventional differenced/GMM moment conditions (see \cite{ahn1995efficient}, \cite{blundell1998initial}, \cite{alvarez2022robust}, and \cite{sentana2024finite}). In this subsection, we show that $\alpha$ is globally identified for any value, including $\alpha=1$, given the moment conditions considered in this paper. The difference with respect to the previous literature comes from the moment conditions employed, as working with the data in levels contains additional information that permits identifying the model in this particular case.

As in previous sections, we treat the first observation separately and project $y_{i1}$ onto the linear span of $[1,\gamma_i]$:
\[y_{i1} \;=\; \delta_1^{*} \;+\; f_\gamma\, \gamma_i \;+\; \varepsilon^{*}_{i1}\]
where $(\delta_1^{*},f_\gamma)$ are projection coefficients and $\varepsilon^{*}_{i1}$ is the projection residual. Since we place no restrictions on the distribution of $\varepsilon_{i1}$, we simply relabel $(\delta_1^{*},\varepsilon_{i1}^{*})$ as $(\delta_1,\varepsilon_{i1})$ without loss of generality.

We include time effects $\delta_t$ for generality, but our method does not rely on time effects to identify the model. If the true DGP does not feature time fixed effects (as in \cite{arellano1991some}), one may impose $\delta_t=0$ for $t\ge2$ and retain a free $\delta_1$ arising from the projection at $t=1$. In any case, $\delta_t$ can be concentrated out as in Section \ref{sec2.2}, and our identification arguments are based on second-moment conditions that are invariant to this step.

As we have previously outlined, we can write a dynamic panel model with individual fixed effects as an interactive fixed-effects model with a unique factor given by $\mathbf{F}=[f_\gamma,1,\dots,1]'$ and factor loadings given by $\gamma_i$. Then, we can stack observations over $t$ as before and write:
\[y_{i} = \mathbf{\Gamma}\delta + \mathbf{\Gamma}\mathbf{F}\gamma_i+\mathbf{\Gamma}\varepsilon_{i}\]
where everything is defined as before except that now $\mathbf{F}=[f_\gamma,1,\dots,1]'$ and $\gamma_i$ is a scalar containing the individual fixed effect of observation $i$. Also note that the matrix $\mathbf{\Psi}$ is now a scalar $\Psi>0$ corresponding to the population variance of the individual fixed effects $\mathbb{V}(\gamma_i)$.\footnote{The variance characterization is appropriate when individual fixed effects (loadings) are random. If they are deterministic a similar consideration as in Section \ref{sec2} follows.} The assumptions required for identification are:

\begin{assumptionprimeprime}\label{asppp:1}
	
	The errors	$\varepsilon_{it}$ in \eqref{eq23} are i.i.d.\ over $i$ and satisfy the MDC with respect to the filtration ${\cal F}_{it}$ defined earlier, and $\mathbb{V}(\varepsilon_{it})=d_t>0$.
	
\end{assumptionprimeprime}

\begin{assumptionprimeprime}\label{asppp:2}
	The individual fixed effects $\gamma_{i}$ are either:
	\begin{enumerate}
		\item Random coefficients satisfying $\E[\gamma_i]=0$ and $\Psi_{N} \xrightarrow{p} \Psi>0$ as $N\to\infty$.
		\item Deterministic coefficients satisfying $\Psi_{N} \to \Psi>0$ as $N \to \infty$.
	\end{enumerate}
\end{assumptionprimeprime}

Under these assumptions the autoregressive coefficient $\alpha$ can be globally identified provided $T$ is large enough. The proof strategy changes depending on whether $\alpha=1$ or not, so we separate these cases into two different Theorems. In the main text, we show the proof for the unit-root case as it is the one where previous literature has reported underidentification. The proof does not rely on distributional assumptions on the factor $\mathbf{F}$.

\begin{theorem}\label{thm:5}
	Let $y_{it}$ be generated from \eqref{eq23} with $\alpha=1$. Then, under Assumptions \ref{asppp:1}-\ref{asppp:2}:
	\begin{enumerate}
		\item If $T=4$, then $\alpha$ is globally identified for all parameter values outside a Lebesgue-null subset $\Theta^{*}_{(2)}\subset\Theta$.
		\item If $T\geq 5$, then $\alpha$ is globally identified.
	\end{enumerate}
\end{theorem}
\begin{proof}
	
	In a dynamic panel model with individual effects the limiting expected value of the sample covariance matrix of $y_i$ can be written as in Section \ref{sec2}
	\begin{equation*}
		\mathbf{\Sigma}(\theta^0) = \mathbf{\Gamma}( \mathbf{F}\Psi \mathbf{F}' + \mathbf{D} )\mathbf{\Gamma}'
	\end{equation*}
	where now $\mathbf{F}=[f_{\gamma},1,\dots,1]'$, $\Psi$ is a positive scalar corresponding to the limiting cross-sectional variance of the individual effects, and $\mathbf{F}\Psi \mathbf{F}'$ is a $T\times T$ symmetric (PSD) matrix of rank $1$. Everything else is defined as before.
	
	As in previous sections, to establish global identification of $\alpha$ it is enough to show that any admissible parameterization satisfying $\mathbf{\Sigma}(\theta^0)=\mathbf{\Sigma}(\tilde{\theta})$ must have $\tilde{\alpha}=\alpha$. We show that, when $T=4$, this holds for all parameter values outside a Lebesgue-null subset $\Theta^{*}_{(2)}\subset\Theta$. We then show that, when $T\geq 5$, the implication holds for all admissible parameter values, meaning that $\alpha$ is globally identified.
	
	Note that most of the results derived in Section \ref{sec3} can be applied without modifications to the case in which we have a single factor given by $\mathbf{F}=[f_\gamma,1,\dots,1]$, as they did not require any distributional assumption on $\mathbf{F}$. In particular, we can define $\mathbf{O}$ in the same way as before and Proposition \ref{prop:1} still implies that the determinant of every diagonal exclusion submatrix of $\mathbf{O}$ of dimension $2\times 2$ must be equal to zero (we only have $\bar{r}=1$ factor given by the individual effects). Moreover, Propositions \ref{prop:2}-\ref{prop:3} still apply so we can write the determinant of every $2\times2$ diagonal exclusion submatrix of $\mathbf{O}$ as:
	\begin{equation*}
		\det\left(\mathbf{M}_{R,C}^{\mathbf{O},2}\right) = (\alpha-\tilde{\alpha})\tilde{J}_{R,C}(\tilde{\alpha},\theta^0) =0
	\end{equation*}
	Consider the unit root case. Start with the first part of the Theorem and assume $T=4$. Consider the diagonal exclusion minor of $\mathbf{O}$ associated with rows $R=(1,2)$ and columns $C=(3,4)$. We can show that:
	\begin{equation}\label{eq24}
		\det\left(\mathbf{M}_{(1,2),(3,4)}^{\mathbf{O},2}\right) =\Psi(\tilde{\alpha}-1)^{2}(d_1-d_2 f_\gamma)=0
	\end{equation}
	\noindent Thus, in order for \eqref{eq24} to be satisfied it must be that either:
	
	\begin{caseprime}\label{casepp:1}
		$\Psi=0$
	\end{caseprime}
	
	\begin{caseprime}\label{casepp:2}
		$d_1=f_\gamma d_{2}$
	\end{caseprime}
	
	\begin{caseprime}\label{casepp:3}
		$\tilde{\alpha}=\alpha=1$
	\end{caseprime}
	
	Case \ref{casepp:1} is ruled out by assumption, as $\Psi>0$. The case in which we allow $\Psi=0$ (i.e. $\gamma_i=\gamma \ \forall \ i$) can also be identified, but we omit the discussion here for simplicity.
	
	Case \ref{casepp:2} is the only parameter restriction left by the selected minor when $T=4$.\footnote{Trying other diagonal exclusion minors does not yield promising results in this case.} This restriction is lower-dimensional, since it is contained in the zero set of the nonzero polynomial $d_2f_\gamma-d_1$, and therefore has Lebesgue measure zero in the parameter space. Hence, defining $\Theta^{*}_{(2)}=\{\theta\in\Theta:d_{1} =d_{2} f_\gamma\}$, \eqref{eq24} implies that, when $T=4$ and $\alpha=1$, the autoregressive parameter is globally identified for all parameter values outside $\Theta^{*}_{(2)}$. When $T\geq5$, the additional period provides further diagonal-exclusion restrictions that also rule out this remaining case, as shown below.
	
	Now, consider the case where $T\geq 5$. We can construct the same minor associated with rows $R=(1,2)$ and columns $C=(3,4)$ considered in Equation \eqref{eq24} and repeat the same argument, so we are only left to show that identification also holds when $d_{1}=f_{\gamma}d_{2}$. Thus, from now on we work out the case where $\alpha=1$ and $d_{1}=f_{\gamma}d_{2}$.
	
	By Proposition \ref{prop:1} the determinant of every diagonal exclusion submatrix of dimension $2\times 2$ must be equal to zero. We consider the minor of $\mathbf{O}$ associated with rows $R=(1,3)$ and columns $C=(4,5)$ under the proposed restrictions:	
	\begin{equation}\label{eq25}
		\det\left(\mathbf{M}_{(1,3),(4,5)}^{\mathbf{O},2}\right) =\Psi \frac{d_{1}}{d_{2}} \left(d_{2}-d_{3}\right)(\tilde{\alpha}-1)^{2}=0
	\end{equation}
	Thus, since both $\Psi>0$ and $d_1>0$, unless $d_{1}=f_{\gamma}d_{2}$ and $d_{3}=d_{2}$ simultaneously global identification of $\alpha$ would be guaranteed as $\tilde{\alpha}=\alpha=1$ would be the only way of making this minor equal to zero. Therefore, we consider the case where $d_{1}=f_{\gamma}d_{2}$ and $d_{3}=d_{2}$. Now, pick the minor associated with rows $R=(1,3)$ and columns $C=(2,4)$. We can show that
	\begin{equation}\label{eq26}
		\det\left(\mathbf{M}_{(1,3),(2,4)}^{\mathbf{O},2}\right) =\frac{d_{1}}{d_{2}}(\tilde{\alpha}-1)^{2}\left(d_{1}\Psi + 3d_{2}\Psi-\tilde{\alpha}d_{2}^{2}+d_{2}^{2}-\tilde{\alpha} d_{1} \Psi - 2\tilde{\alpha} d_{2}\Psi\right)=0
	\end{equation}
	As $d_{1}>0$, the only alternative in which $\alpha$ is not globally identified under these restrictions is when we construct an alternative parameterization in which $\tilde{\alpha}$ makes the last term equal to zero. This corresponds to setting $\tilde{\alpha}^{\star} = \frac{d_{2}^{2}+3\Psi d_{2}+d_{1}\Psi}{d_{2}^{2}+2\Psi d_{2}+d_{1}\Psi}$.\footnote{The denominator is strictly greater than zero so such a solution can always be constructed.}
	
	In order for this alternative parameterization with $\tilde{\alpha}^{\star} = \frac{d_{2}^{2}+3\Psi d_{2}+d_{1}\Psi}{d_{2}^{2}+2\Psi d_{2}+d_{1}\Psi}$ to be valid in the case where $\alpha=1$, $d_{1}=f_{\gamma}d_{2}$ and $d_{3}=d_{2}$ it must make the determinant of all $2\times 2$ diagonal exclusion submatrices equal to zero. However, this is not the case. To see this, pick the minor associated with rows $R=(2,3)$ and columns $C=(1,5)$. Under these restrictions (i.e. $\alpha=1$, $d_{1}=f_{\gamma}d_{2}$, $d_3=d_2$, and $\tilde{\alpha}=\frac{d_{2}^{2}+3\Psi d_{2}+d_{1}\Psi}{d_{2}^{2}+2\Psi d_{2}+d_{1}\Psi}$) the minor can be shown to be equal to:
	\begin{equation}\label{eq27}
		\det\left(\mathbf{M}_{(2,3),(1,5)}^{\mathbf{O},2}\right) =-\frac{d_{1} d_{2}^{3} \Psi^{4}}{\left(d_{2}^{2}+2\Psi d_{2}+d_{1}\Psi\right)^{3}}
	\end{equation}
	and note that because $d_{1}>0$, $d_{2}>0$, and $\Psi>0$, then the determinant is strictly negative. This means that in the unit root case when $d_{1}=f_{\gamma}d_{2}$, $d_{3}=d_{2}$, and $\tilde{\alpha}^{\star} = \frac{d_{2}^{2}+3\Psi d_{2}+d_{1}\Psi}{d_{2}^{2}+2\Psi d_{2}+d_{1}\Psi}$ (which was necessary to make the minor associated with rows $R=(1,3)$ and columns $C=(2,4)$ equal to zero) there is a $2\times 2$ minor different from zero. This violates the necessary condition in Proposition \ref{prop:1}. Therefore, no admissible parameterization satisfying \eqref{eq5:objective} can have $\tilde{\alpha}\neq\alpha$ in the remaining case $d_1=f_\gamma d_2$ when $T\geq5$.
	
	Thus, when $T\geq5$ and $\alpha=1$, any admissible parameterization satisfying \eqref{eq5:objective} must have $\tilde{\alpha}=\alpha=1$. Hence $\alpha$ is globally identified in the unit-root case.
\end{proof}

Moreover, it is also the case that $\alpha$ can be identified for any arbitrary value different from $1$ as well, which is formalized below:

\begin{theorem}\label{thm:6} Let $y_{it}$ be generated from \eqref{eq23} with $\alpha \in \mathbb{R}-\{1\}$. Then, under Assumptions \ref{asppp:1}-\ref{asppp:2}:
	\begin{enumerate}
		\item If $T=4$, then $\alpha$ is globally identified for all parameter values outside a Lebesgue-null subset $\Theta^{*}_{(3)}\subset\Theta$.
		\item If $T\geq 5$, then $\alpha$ is globally identified.
	\end{enumerate}
\end{theorem}

The proof is slightly more involved and it is available in the Online Appendix. Moreover, the exceptional region $\Theta^{*}_{(3)}$ reduces to $\Theta^{*}_{(2)}$ when $\alpha=1$, so $\Theta^{*}_{(2)} = \Theta^{*}_{(3)} \big\vert_{\alpha=1}$. Once $\alpha$ has been identified for every possible value, the rest of the parameters are identified as well:

\begin{corollary}\label{coro:2} Let $y_{it}$ be generated from \eqref{eq23}. Then, under Assumptions \ref{asppp:1}-\ref{asppp:2}:
	\begin{enumerate}
	\item If $T=4$, then $\{d_t\}_{t=1}^{T}$, $\Psi$, and $f_\gamma$ are globally identified for all parameter values outside a Lebesgue-null subset $\Theta^{*}_{(3)}\subset\Theta$.
		\item If $T\geq 5$ then $\{d_{t}\}_{t=1}^{T},\Psi$, and $f_{\gamma}$ are globally identified.
	\end{enumerate}
\end{corollary}

\begin{proof}
	By Theorems~\ref{thm:5} and~\ref{thm:6}, $\alpha$ is globally identified
	outside, at most, $\Theta^{*}_{(3)}$ when $T=4$, and globally identified
	for all admissible parameter values when $T\ge5$. The remaining argument is identical to that in the proof of
	Theorem~\ref{thm:2}. The only condition to check is \cite{anderson1956statistical}'s
	row-deletion rank condition. Here $\mathbf F=[f_\gamma,1,\ldots,1]'$, so
	the condition holds for $T\ge4$. Thus $\mathbf D$ and
	$\mathbf F\Psi\mathbf F'$, and hence $\{d_{t}\}_{t=1}^{T}$, $\Psi$, and $f_\gamma$ are
	identified.
\end{proof}

\section{Concluding remarks}\label{sec5}

In this paper, we show that the parameters of a
dynamic panel data model with interactive effects
are almost surely globally identified. Our result extends the
seminal work of \cite{anderson1956statistical} and \cite{bekker1997generic} to dynamic panels with interactive effects, a class of models previously believed to lack global identification.

This finding has important practical implications. Despite the model’s inherent nonlinearity, there
exists (almost surely) a one-to-one mapping between the first two population moments and the model parameters. This guarantees that researchers can uniquely recover the true parameters, avoiding ambiguity and ensuring the validity of estimation, inference, and policy analysis
in empirical applications.

We believe the results derived in this paper can be applied to establish identification in a wide range of dynamic
models featuring time, individual, or interactive
effects. The proof techniques we develop extend
naturally to settings with more general error
structures, providing a flexible framework for
addressing identification in such cases. For instance, in the baseline model \eqref{eq1:dgp}, if the error term
$\varepsilon_{it}$ follows an MA($q$) process, identification of
$\alpha$ can be obtained using minors of $\mathbf O$ whose entries lie
outside the band formed by the main diagonal and the first $q$ super- and
subdiagonals, provided that $T$ is sufficiently large.

More generally, we believe that some modifications of our original proof can be used to show identification in dynamic models where the covariance between the error terms features some level of sparsity (i.e., the matrix $\mathbf{D}$ has ``enough'' zeros), permitting us to construct minors excluding those terms in $\tilde{\mathbf{D}}$.

Our proof shows that almost sure global identification is achievable when $ T \geq 2(\bar{r} + 1) $. An interesting open question is whether this result can be attained under weaker conditions on $T$, such as when $T$ satisfies the Ledermann bound plus one. The requirement for an additional period, beyond the Ledermann bound, arises from the dynamic nature of the model, which includes an extra parameter compared to the static case.

\newpage
\bibliographystyle{apalike}
\bibliography{bibl}

@article{alvarez2022robust,
	title={Robust likelihood estimation of dynamic panel data models},
	author={Alvarez, Javier and Arellano, Manuel},
	journal={Journal of Econometrics},
	volume={226},
	number={1},
	pages={21--61},
	year={2022},
	publisher={Elsevier}
}

@article{bekker1997generic,
	title={Generic global indentification in factor analysis},
	author={Bekker, Paul A and Ten Berge, Jos MF},
	journal={Linear Algebra and its Applications},
	volume={264},
	pages={255--263},
	year={1997},
	publisher={Elsevier}
}

@inproceedings{anderson1956statistical,
	title={Statistical inference in factor analysis},
	author={Anderson, TW and Rubin, Herman},
	booktitle={Proceedings of the Berkeley Symposium on Mathematical Statistics and Probability},
	pages={111--150},
	year={1956},
	organization={University of California Press}
}

@article{arellano1991some,
	title={Some tests of specification for panel data: Monte {C}arlo evidence and an application to employment equations},
	author={Arellano, Manuel and Bond, Stephen},
	journal={The Review of Economic Studies},
	volume={58},
	number={2},
	pages={277--297},
	year={1991},
	publisher={Wiley-Blackwell}
}

@article{bai2009panel,
	title={Panel data models with interactive fixed effects},
	author={Bai, Jushan},
	journal={Econometrica},
	volume={77},
	number={4},
	pages={1229--1279},
	year={2009},
	publisher={Wiley Online Library}
}

@article{bai2024likelihood,
	title = {Likelihood approach to dynamic panel models with interactive effects},
	journal = {Journal of Econometrics},
	volume = {240},
	number = {1},
	pages = {105636},
	year = {2024},
	author = {Jushan Bai}
	}

@article{hayakawa2023short,
	title={Short {T} dynamic panel data models with individual, time and interactive effects},
	author={Hayakawa, Kazuhiko and Pesaran, M Hashem and Smith, L Vanessa},
	journal={Journal of Applied Econometrics},
	volume={38},
	number={6},
	pages={940--967},
	year={2023},
	publisher={Wiley Online Library}
}

@article{kociecki2023solution,
	title   = {A solution to the global identification problem in {DSGE} models},
	author  = {Kociecki, Andrzej and Ko{\l}asa, Marcin},
	journal = {Journal of Econometrics},
	volume  = {236},
	number  = {2},
	pages   = {105477},
	year    = {2023},
	publisher = {Elsevier}
}

@article{komunjer2011dynamic,
	title={Dynamic {I}dentification of {D}ynamic {S}tochastic {G}eneral {E}quilibrium {M}odels},
	author={Komunjer, Ivana and Ng, Serena},
	journal={Econometrica},
	volume={79},
	number={6},
	pages={1995--2032},
	year={2011},
	publisher={Wiley Online Library}
}

@article{komunjer2012global,
	title={Global identification in nonlinear models with moment restrictions},
	author={Komunjer, Ivana},
	journal={Econometric Theory},
	volume={28},
	number={4},
	pages={719--729},
	year={2012},
	publisher={Cambridge University Press}
}

@article{ledermann1937rank,
	title={On the rank of the reduced correlational matrix in multiple-factor analysis},
	author={Ledermann, Walter},
	journal={Psychometrika},
	volume={2},
	number={2},
	pages={85--93},
	year={1937},
	publisher={Springer}
}

@article{lewbel2019identification,
	title={The {I}dentification {Z}oo: {M}eanings of {I}dentification in {E}conometrics},
	author={Lewbel, Arthur},
	journal={Journal of Economic Literature},
	volume={57},
	number={4},
	pages={835--903},
	year={2019},
	publisher={American Economic Association 2014 Broadway, Suite 305, Nashville, TN 37203-2425}
}

@article{moon2017dynamic,
	title={Dynamic linear panel regression models with interactive fixed effects},
	author={Moon, Hyungsik Roger and Weidner, Martin},
	journal={Econometric Theory},
	volume={33},
	number={1},
	pages={158--195},
	year={2017},
	publisher={Cambridge University Press}
}

@article{pesaran2006estimation,
	title={Estimation and {I}nference in {L}arge {H}eterogeneous {P}anels with a {M}ultifactor {E}rror {S}tructure},
	author={Pesaran, M Hashem},
	journal={Econometrica},
	volume={74},
	number={4},
	pages={967--1012},
	year={2006},
	publisher={Wiley Online Library}
}

@article{rothenberg1971identification,
	author = {Thomas J. Rothenberg},
	title = {Identification in Parametric Models},
	journal = {Econometrica},
	volume = {39},
	number = {3},
	pages = {577-591},
	year = {1971}
}

@article{shapiro1985identifiability,
	title={Identifiability of factor analysis: Some results and open problems},
	author={Shapiro, Alexander},
	journal={Linear Algebra and its Applications},
	volume={70},
	pages={1--7},
	year={1985},
	publisher={Elsevier}
}

@book{arellano2003panel,
	title={Panel data econometrics},
	author={Arellano, Manuel},
	year={2003},
	publisher={OUP Oxford}
}

@book{baltagi2008econometric,
	title={Econometric analysis of panel data},
	author={Baltagi, Badi Hani},
	volume={4},
	year={2008},
	publisher={Springer}
}

@book{hsiao2022analysis,
	title={Analysis of panel data},
	author={Hsiao, Cheng},
	year={2022},
	publisher={Cambridge University Press}
}

@book{wooldridge2010econometric,
	title={Econometric analysis of cross section and panel data},
	author={Wooldridge, Jeffrey M},
	year={2010},
	publisher={MIT press}
}

@article{holtz1988estimating,
	author = {Douglas Holtz-Eakin AND Harvey S. Rosen AND Whitney Newey},
	title = {Estimating Vector Autoregressions with Panel Data},
	journal = {Econometrica},
	volume = {56},
	number = {6},
	pages = {1371-1395},
	year = {1988}
}

@article{ahn2013panel,
	title={Panel data models with multiple time-varying individual effects},
	author={Ahn, Seung C and Lee, Young H and Schmidt, Peter},
	journal={Journal of Econometrics},
	volume={174},
	number={1},
	pages={1--14},
	year={2013},
	publisher={Elsevier}
}

@article{ahn2001gmm,
	title={G{M}{M} estimation of linear panel data models with time-varying individual effects},
	author={Ahn, Seung Chan and Lee, Young Hoon and Schmidt, Peter},
	journal={Journal of Econometrics},
	volume={101},
	number={2},
	pages={219--255},
	year={2001},
	publisher={Elsevier}
}

@article{sentana2024finite,
	title={Finite underidentification},
	author={Sentana, Enrique},
	journal={Journal of Econometrics},
	volume={240},
	number={1},
	pages={105692},
	year={2024},
	publisher={Elsevier}
}

@article{blundell1998initial,
	title={Initial conditions and moment restrictions in dynamic panel data models},
	author={Blundell, Richard and Bond, Stephen},
	journal={Journal of Econometrics},
	volume={87},
	number={1},
	pages={115--143},
	year={1998},
	publisher={Elsevier}
}

@book{ersel2011linear,
	title={Linear algebra for economists},
	author={Ersel, Hasan and Piontkovski, Dmitri},
	year={2011},
	publisher={Springer Science \& Business Media}
}

@book{athreya2006measure,
	title={Measure theory and probability theory},
	author={Athreya, Krishna B and Lahiri, Soumendra N},
	year={2006},
	publisher={Springer}
}

@article{caron2005zero,
	title={The zero set of a polynomial},
	author={Caron, Richard and Traynor, Tim},
	journal={WSMR report},
	pages={05--02},
	year={2005}
}

@article{gorgens2019moment,
	title={Moment restrictions and identification in linear dynamic panel data models},
	author={G{\o}rgens, Tue and Han, Chirok and Xue, Sen},
	journal={Annals of Economics and Statistics},
	number={134},
	pages={149--176},
	year={2019},
	publisher={JSTOR}
}

@book{boyd2004convex,
	title={Convex optimization},
	author={Boyd, Stephen and Vandenberghe, Lieven},
	year={2004},
	publisher={Cambridge university press}
}

@article{ahn1995efficient,
	title={Efficient estimation of models for dynamic panel data},
	author={Ahn, Seung C and Schmidt, Peter},
	journal={Journal of econometrics},
	volume={68},
	number={1},
	pages={5--27},
	year={1995},
	publisher={Elsevier}
}

\clearpage

\appendix

\onehalfspacing
\begin{appendix}
	
	\section{Proofs of main results}\label{Ap:A}
		
	\subsection*{Proof of Proposition \ref{prop:2}:}
	
	\begin{proof} The proof is done by induction. Recall that the sets $R$ and $C$ are assumed to be ordered so that $r_1<\dots<r_k$ and $c_1<\dots<c_k$ (see Remark \ref{rem:2}). We sometimes denote $\tilde{J}_{R,C}$ instead of $\tilde{J}_{R,C}(\tilde{\alpha},\theta^0)$ to keep the notation as simple as possible.
	
	\textit{1. Base case $k=1$}: Let $R=\{r\}$ and $C=\{c\}$, with $r\neq c$. Then, a diagonal exclusion submatrix of $\mathbf{O}$ of dimension 1, say $\mathbf{M}^{\mathbf{O},1}_{r,c}$, is simply the term $O_{r,c}$ (for $r\neq c$) and its determinant is trivially given by
	\[\det\left(\mathbf{M}^{\mathbf{O},1}_{r,c}\right) = O_{r,c}\]
	From Corollary \ref{coro:1}, we know that
	\[O_{r,c} = \Omega_{r,c}+(\alpha-\tilde{\alpha})J_{r,c} \quad r\neq c\]
	Since $\Omega_{r,c}$ is scalar we have that $\Omega_{r,c}=\det\left(\mathbf{M}^{\mathbf{\Omega},1}_{r,c}\right)$. We can rewrite the preceding equation as:
	\begin{equation}\label{eq36}
		\det\left(\mathbf{M}^{\mathbf{O},1}_{r,c}\right) = \det\left(\mathbf{M}^{\mathbf{\Omega},1}_{r,c}\right) + (\alpha-\tilde{\alpha})\tilde{J}_{r,c}
	\end{equation}
	where $\tilde J_{r,c}= J_{r,c}$ when $k=1$. Therefore, the base step in the induction holds.
	
	\textit{2. Inductive hypothesis:} Assume our proposition holds for diagonal exclusion minors associated with submatrices of dimension $(k-1)\times(k-1)$. This implies that for any set of valid $R$ and $C$ we have
	\begin{equation}\label{eq37}
		\det\left(\mathbf{M}^{\mathbf{O},k-1}_{R,C}\right) = \det\left(\mathbf{M}^{\mathbf{\Omega},k-1}_{R,C}\right) + (\alpha-\tilde{\alpha})\tilde{J}_{R,C}\left(\tilde{\alpha},\theta^0\right)
	\end{equation}
	
	\textit{3. Inductive step for $k$}: We prove the desired result using both Corollary \ref{coro:1}, and invoking the Laplace expansion for determinants.
	
	Let $\mathbf{M}^{\mathbf{O},k}_{R,C}$ be a diagonal exclusion submatrix of $\mathbf{O}$ of dimension $k$ with rows $R=\{r_1,r_2,\dots,r_k\}$ and columns $C=\{c_1,c_2,\dots,c_k\}$ satisfying the usual conditions. Before proceeding, we need to establish a correspondence between the indices of rows $R$ and columns $C$ in the original matrices (either $\mathbf{O}$ or $\mathbf{\Omega}$) and their new indices (i.e. $\{1,2,\dots,k\}$) in the $k\times k$ submatrices $\mathbf{M}^{\mathbf{O},k}_{R,C}$ and $\mathbf{M}^{\mathbf{\Omega},k}_{R,C}$. Without loss of generality, let $(i,j)$ denote the position of the $i$-th row in $R$ and the $j$-th column in $C$ within the submatrices $\mathbf{M}^{\mathbf{O},k}_{R,C}$ and $\mathbf{M}^{\mathbf{\Omega},k}_{R,C}$. This establishes that the $(i,j)$-th elements of the submatrices $\mathbf{M}^{\mathbf{O},k}_{R,C}$ and $\mathbf{M}^{\mathbf{\Omega},k}_{R,C}$ are respectively given by $O_{r_i,c_j}$ and $\Omega_{r_i,c_j}$.
	
	With this established, we can use the Laplace expansion (fixing the row $r_k\in R$ in the original matrix) to write the determinant of $\mathbf{M}^{\mathbf{O},k}_{R,C}$ as
	\begin{equation}\label{eq38}
		\det\left(\mathbf{M}^{\mathbf{O},k}_{R,C}\right) = \sum_{j=1}^{k} (-1)^{k+j}O_{r_k,c_j}\det\left(\mathbf{M}^{\mathbf{O},k-1}_{R-r_k,C-c_j}\right)
	\end{equation}
	where $k$ and $j$ represent the indices in the submatrix $\mathbf{M}^{\mathbf{O},k}_{R,C}$ of row $r_k$ and columns $c_j$ of the matrix $\mathbf{O}$. The scalar $O_{r_k,c_j}$ is the $(k,j)$-th element of the submatrix $\mathbf{M}^{\mathbf{O},k}_{R,C}$; the matrix $\mathbf{M}^{\mathbf{O},k-1}_{R-r_k,C-c_j}$ is the submatrix of matrix $\mathbf{O}$ with dimension $k-1$ that includes rows $R-r_k$ and columns $C-c_j$.
	
	Because $r_k\neq c_j$ for all $j$ (recall $R \cap C =\emptyset$), we can use Corollary \ref{coro:1} to write
	\[O_{r_k,c_j} = \Omega_{r_k,c_j}+(\alpha-\tilde{\alpha})J_{r_k,c_j}\]
	Furthermore, by the inductive hypothesis
	\[\det\left(\mathbf{M}^{\mathbf{O},k-1}_{R-r_k,C-c_j}\right) = \det\left(\mathbf{M}^{\mathbf{\Omega},k-1}_{R-r_k,C-c_j}\right) + (\alpha-\tilde{\alpha})\tilde{J}_{R-r_k,C-c_j}\]
	Plugging these results into \eqref{eq38} we have that
	\begin{small}
		\[\det\left(\mathbf{M}^{\mathbf{O},k}_{R,C}\right) = \sum_{j=1}^{k} (-1)^{k+j}\left(\Omega_{r_k,c_j}+(\alpha-\tilde{\alpha})J_{r_k,c_j}\right)\left[\det\left(\mathbf{M}^{\mathbf{\Omega},k-1}_{R-r_k,C-c_j}\right) + (\alpha-\tilde{\alpha})\tilde{J}_{R-r_k,C-c_j}\right]\]
	\end{small}
	Distributing terms we obtain
	\begin{small}
		\begin{equation*}
			\begin{split}
				&\det\left(\mathbf{M}^{\mathbf{O},k}_{R,C}\right) = \sum_{j=1}^{k} (-1)^{k + j} \Omega_{r_k,c_j} \det\left( \mathbf{M}^{\mathbf{\Omega},k-1}_{R-r_k,C-c_j} \right)+ \\
				& (\alpha - \tilde{\alpha})\sum_{j=1}^{k} (-1)^{k + j}  \left[ J_{r_k,c_j} \det\left( \mathbf{M}^{\mathbf{\Omega},k-1}_{R-r_k,C-c_j} \right) + \Omega_{r_k,c_j} \tilde{J}_{R-r_k,C-c_j} + (\alpha - \tilde{\alpha}) J_{r_k,c_j} \tilde{J}_{R-r_k,C-c_j} \right]
			\end{split}
		\end{equation*}
	\end{small}
	Then, using the Laplace expansion we can rewrite the first term of the expression above as
	\begin{equation*}
		\sum_{j=1}^{k} (-1)^{k+j}\Omega_{r_k,c_j}\det\left(\mathbf{M}^{\mathbf{\Omega},k-1}_{R-r_k,C-c_j}\right)= \det\left(\mathbf{M}^{\mathbf{\Omega},k}_{R,C}\right).
	\end{equation*}
	Next, define $\tilde{J}_{R,C}(\tilde{\alpha}, \theta^0)$ as
	\begin{equation*}
		\begin{split}
			\tilde{J}_{R,C}(\tilde{\alpha}, \theta^0) = \sum_{j=1}^{k}(-1)^{k+j}\Bigl[J_{r_k,c_j}\det\left(\mathbf{M}^{\mathbf{\Omega},k-1}_{R-r_k,C-c_j}\right)+\\
			\Omega_{r_k,c_j}\tilde{J}_{R-r_k,C-c_j}(\tilde{\alpha},\theta^0) +(\alpha-\tilde{\alpha})J_{r_k,c_j}\tilde{J}_{R-r_k,C-c_j}(\tilde{\alpha},\theta^0)\Bigr].
		\end{split}
	\end{equation*}
	This permits us to conclude that
	\[\det\left(\mathbf{M}^{\mathbf{O},k}_{R,C}\right) =\det\left(\mathbf{M}^{\mathbf{\Omega},k}_{R,C}\right) +(\alpha-\tilde{\alpha})\tilde{J}_{R,C}(\tilde{\alpha}, \theta^0).\]
	The two preceding equations correspond to (\ref{eq11:Jtilde}) and (\ref{eq10:theo3}), respectively.
	\end{proof}
	
	\subsection*{Proof of Proposition \ref{prop:3}:}
	\begin{proof} Recall that we defined $\mathbf{\Omega}=\mathbf{F\Psi F}' + \mathbf{D}$. Then, as $\mathbf{D}$ is diagonal, one can see that all the off-diagonal entries of $\mathbf{\Omega}$ are equal to the off-diagonal terms of $\mathbf{F\Psi F}'$. Thus, as all the terms are equal it follows that for every diagonal exclusion $k \times k$ submatrix we have that
	\[\mathbf{M}^{\mathbf{\Omega},k}_{R,C}=\mathbf{M}^{\mathbf{F\Psi F}',k}_{R,C}\]
	Then, the determinant of every diagonal exclusion submatrix of $\mathbf{\Omega}$ is equal to the determinant of the corresponding diagonal exclusion submatrix of $\mathbf{F\Psi F}'$. Now, let $k>\bar{r}$. Then, as $\mathbf{F\Psi F}'$ is of rank $\bar{r}$, Lemma \ref{lemma:1} implies that $\det\left(\mathbf{M}^{\mathbf{\Omega},k}_{R,C}\right)=0$, proving the proposition. \end{proof}
	
	\subsection*{Proof of Proposition \ref{prop:4}:}
	
	\begin{proof} First, note from \eqref{eq11:Jtilde} that $\tilde{J}_{R,C}(\tilde{\alpha}, \theta^0)$ is recursively defined, so once we know $\tilde{J}_{R,C}(\tilde{\alpha}, \theta^0)$ for $|R|=|C|=1$ we can compute $\tilde{J}_{R,C}(\tilde{\alpha}, \theta^0)$ for any arbitrary set of rows and columns. When $k=1$, we have shown that $\tilde{J}_{R,C} = J_{r,c}$. Now, recall the definition of $\mathbf{J}(\tilde{\alpha}, \theta^0)$ in Claim \ref{claim:1}. $\mathbf{\Omega}$ and $\mathbf{L}$ do not depend on $\tilde{\alpha}$, so $J_{r,c}$ is a polynomial of degree at most 1 in $\tilde{\alpha}$. Given the recursive structure of $\tilde{J}_{R,C}(\tilde{\alpha}, \theta^0)$ we can conclude that for any value of $k$, $\tilde{J}_{R,C}\left(\tilde{\alpha},\theta^0\right)$ is a polynomial in $\tilde{\alpha}$.
	
	Conditional on the realized value of $\mathbf{F}_2$, the coefficients of $\tilde J_{R,C}(\tilde\alpha,\theta^0)$ are fixed functions of $\theta^0=(\alpha,\mathbf{F},\mathbf{\Psi},\mathbf{D})$. Now, we show that the proposed upper bound holds. The proof is done by induction. We show that $\tilde{J}_{R,C}\left(\tilde{\alpha},\theta^0\right)$ is a polynomial in $\tilde\alpha$ of degree at most $2k-1$ when $\vert R \vert = \vert C \vert = k$.
	
	\textit{Base case $k=1$}: We have shown that when $k=1$ $\tilde{J}_{R,C}\left(\tilde{\alpha},\theta^0\right)$ is a polynomial in $\tilde{\alpha}$ of at most degree $1=2k-1$, so the base step holds.
	
	\textit{Inductive hypothesis:} Assume our proposition holds for diagonal exclusion minors associated with submatrices of dimension $(k-1)\times(k-1)$. This implies that for any set of valid $R$ and $C$ we have that $\tilde{J}_{R,C}\left(\tilde{\alpha},\theta^0\right)$ is a polynomial in $\tilde{\alpha}$ of at most degree $2(k-1)-1$.
	
	\textit{Inductive step:} Recall the recursive definition of $\tilde{J}_{R,C}(\tilde{\alpha}, \theta^0)$ derived in \eqref{eq11:Jtilde}. By the inductive hypothesis $\tilde{J}_{R-r_k,C-c_j}$ is a polynomial in $\tilde{\alpha}$ of at most degree $2(k-1)-1$, the term $(\alpha-\tilde{\alpha})$ is a polynomial of degree 1 in $\tilde{\alpha}$, and $J_{r_k,c_j}$ is also a polynomial in $\tilde{\alpha}$ of at most degree one (see base step). Therefore, the term $(\alpha-\tilde{\alpha})J_{r_k,c_j}\tilde{J}_{R-r_k,C-c_j}$ is a polynomial in $\tilde{\alpha}$ of at most degree $2k-1$. Since the first two terms are of lower degree ($\mathbf{\Omega}$ does not depend on $\tilde{\alpha}$), we conclude that $\tilde{J}_{R,C}\left(\tilde{\alpha},\theta^0\right)$ is a polynomial in $\tilde{\alpha}$ of at most degree $2k-1$. \end{proof}

\subsection*{Proof of Proposition \ref{prop:5}:}

\begin{proof} Let $k=\bar{r}+1$ and recall that the sets $R$ and $C$ are ordered so that $r_1<\dots<r_{\bar{r}+1}$ and $c_1<\dots<c_{\bar{r}+1}$. We first prove the first claim.
It is enough to show that $
\tilde J_{R,C}(\alpha,\theta^0)\neq0$
almost surely, because then
$\tilde J_{R,C}(\tilde\alpha,\theta^0)$ cannot be the zero polynomial in
$\tilde\alpha$. By Lemma~\ref{lem:cofactor-representation-alpha} in Appendix \ref{Ap:B},
\begin{equation}\label{eq:cof}
	\tilde J_{R,C}(\alpha,\theta^0)
	=
	\sum_{i=1}^{k}\sum_{j=1}^{k}
	(-1)^{i+j}
	J^\alpha_{r_i,c_j}
	\det \left(\mathbf{M}^{\mathbf{\Omega},k-1}_{R-r_i,C-c_j}\right),
\end{equation}
where $
\mathbf{J}^\alpha\equiv \mathbf{J}(\alpha,\theta^0)=\mathbf{L}\mathbf{\Omega}+\mathbf{\Omega}\mathbf{L}'$.	

Because $R\cap C=\emptyset$, the diagonal matrix $\mathbf{D}$ does not enter the
block $\mathbf{M}^{\mathbf{\Omega},k}_{R,C}$. Hence
\[	\mathbf{M}^{\mathbf{\Omega},k}_{R,C}
=
\mathbf{\Omega}_{R,C}
=
\mathbf{F}_R\mathbf{\Psi}\mathbf{F}_C'. \]
Now decompose
\[
\mathbf{J}^\alpha
=
\mathbf{L}\mathbf{\Omega}+\mathbf{\Omega}\mathbf{L}'
=
\mathbf{L}(\mathbf{F}\mathbf{\Psi}\mathbf{F}')+(\mathbf{F}\mathbf{\Psi}\mathbf{F}')\mathbf{L}'
+
\mathbf{L}\mathbf{D}+\mathbf{D}\mathbf{L}'.
\]
We first show that the factor part contributes zero to the preceding cofactor
sum \eqref{eq:cof}. Consider
\[
\mathbf{M}(\varepsilon)
=
\bigl[(\mathbf{I}+\varepsilon \mathbf{L})\mathbf{F}\bigr]\mathbf{\Psi}
\bigl[(\mathbf{I}+\varepsilon \mathbf{L})\mathbf{F}\bigr]'.
\]
For every $\varepsilon$, the matrix $\mathbf{M}(\varepsilon)$ has rank at most
$\bar{r}$. Since $k=\bar{r}+1$,
\[
\det\bigl(\mathbf{M}(\varepsilon)_{R,C}\bigr)\equiv0.
\]
Moreover,
\[
\mathbf{M}(\varepsilon)_{R,C}
=
(\mathbf{F}\mathbf{\Psi}\mathbf{F}')_{R,C}
+
\varepsilon
\bigl[\mathbf{L}(\mathbf{F}\mathbf{\Psi}\mathbf{F}')+(\mathbf{F}\mathbf{\Psi}\mathbf{F}')\mathbf{L}'\bigr]_{R,C}
+
O(\varepsilon^2).
\]
Taking the coefficient of $\varepsilon$ in
$\det\left(\mathbf{M}(\varepsilon)_{R,C}\right)\equiv0$, by multilinearity of the determinant,
gives
\[
\sum_{i=1}^{k}\sum_{j=1}^{k}
(-1)^{i+j}
\bigl[\mathbf{L}(\mathbf{F}\mathbf{\Psi}\mathbf{F}')+(\mathbf{F}\mathbf{\Psi}\mathbf{F}')\mathbf{L}'\bigr]_{r_i,c_j}
\det \left( (\mathbf{F}\mathbf{\Psi}\mathbf{F}')_{R-r_i,C-c_j} \right)
=
0.
\]
Since $R\cap C=\emptyset$, also
\[
(R-r_i)\cap(C-c_j)=\emptyset
\]
for all $i,j$. Therefore the diagonal matrix $\mathbf{D}$ does not enter these
smaller blocks, and
\[
(\mathbf{F}\mathbf{\Psi}\mathbf{F}')_{R-r_i,C-c_j}
=
\mathbf{M}^{\mathbf{\Omega},k-1}_{R-r_i,C-c_j}.
\]
Hence
\[
\sum_{i=1}^{k}\sum_{j=1}^{k}
(-1)^{i+j}
\bigl[\mathbf{L}(\mathbf{F}\mathbf{\Psi}\mathbf{F}')+(\mathbf{F}\mathbf{\Psi}\mathbf{F}')\mathbf{L}'\bigr]_{r_i,c_j}
\det \left( \mathbf{M}^{\mathbf{\Omega},k-1}_{R-r_i,C-c_j} \right)
=
0.
\]
Thus only the diagonal-error part $\mathbf{L}\mathbf{D}+\mathbf{D}\mathbf{L}'$ remains:
\[
\tilde J_{R,C}(\alpha,\theta^0)
=
\sum_{i=1}^{k}\sum_{j=1}^{k}
(-1)^{i+j}
(\mathbf{L}\mathbf{D}+\mathbf{D}\mathbf{L}')_{r_i,c_j}
\det \left( \mathbf{M}^{\mathbf{\Omega},k-1}_{R-r_i,C-c_j} \right).
\]
Again using $(R-r_i)\cap(C-c_j)=\emptyset$, we have
\[
\mathbf{M}^{\mathbf{\Omega},k-1}_{R-r_i,C-c_j}
=
\mathbf{F}_{R-r_i}\mathbf{\Psi}\mathbf{F}_{C-c_j}'.
\]
Since $k-1=\bar{r}$, the matrices $\mathbf{F}_{R-r_i}$ and $\mathbf{F}_{C-c_j}$ are
$\bar{r}\times \bar{r}$. Therefore
\[
\det \left(\mathbf{M}^{\mathbf{\Omega},k-1}_{R-r_i,C-c_j}\right)
=
\det(\mathbf{\Psi})
\det(\mathbf{F}_{R-r_i})
\det(\mathbf{F}_{C-c_j}).
\]
Thus
\[
\tilde J_{R,C}(\alpha,\theta^0)
=
\det(\mathbf{\Psi})
\sum_{i=1}^{k}\sum_{j=1}^{k}
(-1)^{i+j}
a_{ij}
\det(\mathbf{F}_{R-r_i})
\det(\mathbf{F}_{C-c_j}),
\]
where
\[
a_{ij}=(\mathbf{L}\mathbf{D}+\mathbf{D}\mathbf{L}')_{r_i,c_j}
=
\begin{cases}
	\alpha^{r_i-c_j-1}d_{c_j}, & r_i>c_j,\\[0.4em]
	\alpha^{c_j-r_i-1}d_{r_i}, & c_j>r_i.
\end{cases}
\]
Since $R\cap C=\emptyset$, these are the only cases. Let
\[
\mathbf{A}=(a_{ij})_{i,j=1}^{k}.
\]
Because $d_t>0$ for every $t$, the matrix $\mathbf{A}$ is nonzero whenever
$\alpha\neq0$. If $\alpha=0$ but
\[
\min_{r\in R,\ c\in C}|r-c|=1,
\]
then there exist $r_\ell\in R$ and $c_h\in C$ with
$r_\ell=c_h+1$ or $c_h=r_\ell+1$. In the first case,
\[
a_{\ell h}=d_{c_h}>0,
\]
and in the second case,
\[
a_{\ell h}=d_{r_\ell}>0.
\]
Hence $\mathbf{A}\neq0$ under the conditions of the first claim. Define
\[
u_R(\mathbf{F})
=
\left(
(-1)^i\det(\mathbf{F}_{R-r_i})
\right)_{i=1}^{k},
\qquad
u_C(\mathbf{F})
=
\left(
(-1)^j\det(\mathbf{F}_{C-c_j})
\right)_{j=1}^{k}.
\]
Then
\[
\tilde J_{R,C}(\alpha,\theta^0)
=
\det(\mathbf{\Psi})\,u_R(\mathbf{F})'\mathbf{A}\,u_C(\mathbf{F}).
\]
Since $\mathbf{A}\neq0$, it follows by Lemma~\ref{lem:bilinear-nonzero-polynomial} (see Appendix \ref{Ap:B}) that $
u_R(\mathbf{F})'\mathbf{A}\,u_C(\mathbf{F})$ is a nonzero polynomial in the unrestricted entries of $\mathbf{F}_2$. Since
$\det(\mathbf{\Psi})\neq0$,
\[
\tilde J_{R,C}(\alpha,\theta^0)
=
\det(\mathbf{\Psi})\,u_R(\mathbf{F})'\mathbf{A}\,u_C(\mathbf{F})
\]
is also a nonzero polynomial in $\operatorname{vec}(\mathbf{F}_2)$. Equivalently, for
fixed $\alpha$, $\mathbf{D}$, and $\mathbf{\Psi}$, the set
\[
\mathcal N_{R,C}
=
\left\{
\operatorname{vec}(\mathbf{F}_2)\in\mathbb R^{(T-\bar r)\bar r}:
\tilde J_{R,C}(\alpha,\theta^0)=0
\right\}
\]
is the zero set of a nonzero polynomial in $\operatorname{vec}(\mathbf{F}_2)$. Hence,
by Lemma~\ref{lemma:3}, $\mathcal N_{R,C}$ has Lebesgue measure zero
in the unrestricted-factor parameter space.

We now invoke the probabilistic assumption. By Assumption \ref{as:6},
$\operatorname{vec}(\mathbf{F}_2)$ is absolutely continuous with respect to Lebesgue
measure. Since $\mathcal N_{R,C}$ has Lebesgue measure zero, it follows that
\[
\Pr\left(\operatorname{vec}(\mathbf{F}_2)\in\mathcal N_{R,C}\right)=0.
\]
Therefore
\[
\tilde J_{R,C}(\alpha,\theta^0)\neq0
\]
almost surely, and so
$\tilde J_{R,C}(\tilde\alpha,\theta^0)$ cannot be the zero polynomial in
$\tilde\alpha$ almost surely. This proves the first claim.\vspace{0.4cm}

Now suppose
\[
\alpha=0
\qquad\text{and}\qquad
\min_{r\in R,\ c\in C}|r-c|>1.
\]
When $\alpha=0$, the matrix $\mathbf{L}$ is the one-lag shift matrix and
\[
\mathbf{I}+(\alpha-\tilde\alpha)\mathbf{L}
=
\mathbf{I}-\tilde\alpha \mathbf{L}.
\]
On the off-diagonal block $R,C$, the diagonal matrix $\tilde{\mathbf{D}}$ does not
contribute because $R\cap C=\emptyset$. Hence
\[
\mathbf{M}^{\mathbf{O},k}_{R,C}
=
\left[
(\mathbf{I}-\tilde\alpha \mathbf{L})\mathbf{\Omega}(\mathbf{I}-\tilde\alpha \mathbf{L})'
\right]_{R,C}.
\]
Writing $\mathbf{\Omega}=\mathbf{F}\mathbf{\Psi}\mathbf{F}'+\mathbf{D}$, we can decompose:
\[
\mathbf{M}^{\mathbf{O},k}_{R,C}
=
\left[
(\mathbf{I}-\tilde\alpha \mathbf{L})\mathbf{F}\mathbf{\Psi}\mathbf{F}'(\mathbf{I}-\tilde\alpha \mathbf{L})'
\right]_{R,C}+\left[
(\mathbf{I}-\tilde\alpha \mathbf{L})\mathbf{D}(\mathbf{I}-\tilde\alpha \mathbf{L})'
\right]_{R,C}.
\]
The transformed factor part is
\[
(\mathbf{I}-\tilde\alpha \mathbf{L})\mathbf{F}\mathbf{\Psi}\mathbf{F}'(\mathbf{I}-\tilde\alpha \mathbf{L})'
=
\mathbf{F}^\star(\tilde\alpha)\mathbf{\Psi}\mathbf{F}^\star(\tilde\alpha)',
\]
where
\[
\mathbf{F}^\star(\tilde\alpha)=(\mathbf{I}-\tilde\alpha \mathbf{L})\mathbf{F}.
\]
The matrix $\mathbf{F}^\star(\tilde\alpha)\mathbf{\Psi}\mathbf{F}^\star(\tilde\alpha)'$ has rank at most $\bar{r}$. The transformed diagonal-error part is
\[
(\mathbf{I}-\tilde\alpha \mathbf{L})\mathbf{D}(\mathbf{I}-\tilde\alpha \mathbf{L})'.
\]
Since $\mathbf{L}$ is the one-lag shift matrix when $\alpha=0$, this matrix is tridiagonal. Therefore
its $R,C$ block is zero whenever
\[
\min_{r\in R,\ c\in C}|r-c|>1.
\]
Thus
\[
\mathbf{M}^{\mathbf{O},k}_{R,C}
=
\mathbf{F}^\star_R(\tilde\alpha)\mathbf{\Psi}\mathbf{F}^\star_C(\tilde\alpha)',
\]
which has rank at most $\bar{r}$. Since $k=\bar{r}+1$, $
\det \left(\mathbf{M}^{\mathbf{O},\bar r+1}_{R,C}\right)=0$ for every $\tilde\alpha$. 	Also, $
\det \left(\mathbf{M}^{\mathbf{\Omega},k}_{R,C}\right)=0$
because $
\mathbf{M}^{\mathbf{\Omega},k}_{R,C}=\mathbf{F}_R\mathbf{\Psi}\mathbf{F}_C'$
has rank at most $\bar{r}$. Therefore, using
\[
\det \left( \mathbf{M}^{\mathbf{O},k}_{R,C} \right)
=
\det \left(\mathbf{M}^{\mathbf{\Omega},k}_{R,C}\right)
+
(\alpha-\tilde\alpha)
\tilde J_{R,C}(\tilde\alpha,\theta^0),
\]
and setting $\alpha=0$, we get
\[
0
=
-\tilde\alpha\,
\tilde J_{R,C}(\tilde\alpha,\theta^0)
\]
for every $\tilde\alpha$. Thus
$\tilde J_{R,C}(\tilde\alpha,\theta^0)=0$ for every
$\tilde\alpha\neq0$. Since a univariate polynomial of finite degree that vanishes on infinitely many points is identically zero,
\[
\tilde J_{R,C}(\tilde\alpha,\theta^0)\equiv0.
\]
This proves the second claim. \end{proof}

\subsection*{Proof of Theorem \ref{thm:1} for $\alpha=0$:}

\begin{proof}
	Since $\alpha=0$, we have $\mathbf{\Gamma}=\mathbf{I}_T$ and
	\[
	\mathbf{\Sigma}(\theta^0)=\mathbf{F}\mathbf{\Psi}\mathbf{F}'+\mathbf{D}.
	\]
	Suppose, toward a contradiction, that there exists an observationally equivalent
	$\tilde{\theta}$ with $\tilde{\alpha}\neq 0$. Then
	\[
	\mathbf{F}\mathbf{\Psi}\mathbf{F}'+\mathbf{D}
	=
	\tilde{\mathbf{\Gamma}}
	\left(
	\tilde{\mathbf F}\tilde{\mathbf\Psi}\tilde{\mathbf F}'
	+\tilde{\mathbf D}
	\right)
	\tilde{\mathbf{\Gamma}}'.
	\]
	Premultiplying and postmultiplying by $\tilde{\mathbf{\Gamma}}^{-1}$ and rearranging gives
	\[
	\tilde{\mathbf D}
	-
	\tilde{\mathbf{\Gamma}}^{-1}\mathbf{D}\tilde{\mathbf{\Gamma}}^{-1'}
	=
	\tilde{\mathbf{\Gamma}}^{-1}\mathbf{F}\mathbf{\Psi}\mathbf{F}'
	\tilde{\mathbf{\Gamma}}^{-1'}
	-
	\tilde{\mathbf F}\tilde{\mathbf\Psi}\tilde{\mathbf F}'.
	\]
	The right-hand side is the difference of two matrices of rank at most $\bar r$.
	Therefore,
	\[
	\operatorname{rank}
	\left(
	\tilde{\mathbf D}
	-
	\tilde{\mathbf{\Gamma}}^{-1}\mathbf{D}\tilde{\mathbf{\Gamma}}^{-1'}
	\right)
	\leq 2\bar r.
	\]
	However, when $\alpha=0$ the matrix $\mathbf{L}$ reduces to the subdiagonal
	shift: $L_{t,t-1}=1$ for $t=2,\ldots,T$ and all other below-diagonal entries
	are zero. Hence $\tilde{\mathbf{\Gamma}}^{-1}=\mathbf{I}_T-\tilde{\alpha}\mathbf{L}$
	is lower bidiagonal, and
	$\tilde{\mathbf{\Gamma}}^{-1}\mathbf{D}\tilde{\mathbf{\Gamma}}^{-1'}=
	(\mathbf{I}_T-\tilde{\alpha}\mathbf{L})\mathbf{D}(\mathbf{I}_T-\tilde{\alpha}\mathbf{L}')$
	is the product of a lower bidiagonal matrix, a diagonal matrix, and an upper
	bidiagonal matrix, and is therefore tridiagonal. Since $\tilde{\mathbf{D}}$ is
	diagonal, the matrix
	\[
	\tilde{\mathbf D}
	-
	\tilde{\mathbf{\Gamma}}^{-1}\mathbf{D}\tilde{\mathbf{\Gamma}}^{-1'}
	\]
	is also tridiagonal. Its first off-diagonal elements are
	\[
	\left(
	\tilde{\mathbf D}
	-
	\tilde{\mathbf{\Gamma}}^{-1}\mathbf{D}\tilde{\mathbf{\Gamma}}^{-1'}
	\right)_{t,t+1}
	=
	\left(
	\tilde{\mathbf D}
	-
	\tilde{\mathbf{\Gamma}}^{-1}\mathbf{D}\tilde{\mathbf{\Gamma}}^{-1'}
	\right)_{t+1,t}
	=
	\tilde{\alpha}D_{t,t},
	\qquad t=1,\dots,T-1.
	\]
	Since $\tilde{\alpha}\neq0$ and both $\mathbf{D}$ and $\tilde{\mathbf{D}}$ are diagonal PD,
	\[
	\left(
	\tilde{\mathbf D}
	-
	\tilde{\mathbf{\Gamma}}^{-1}\mathbf{D}\tilde{\mathbf{\Gamma}}^{-1'}
	\right)_{t,t+1}
	=
	\tilde{\alpha}D_{t,t}\neq0,
	\qquad t=1,\dots,T-1.
	\]
	Consider the submatrix obtained by taking rows $1,\dots,T-1$ and columns $2,\dots,T$.
	Because the matrix is tridiagonal, this submatrix is lower-triangular, with diagonal
	entries given by $
	\tilde{\alpha}D_{1,1},\dots,\tilde{\alpha}D_{T-1,T-1}$. Therefore its determinant is
	\[
	\prod_{t=1}^{T-1}\tilde{\alpha}D_{t,t}\neq0.
	\]
	Thus the full matrix has a nonzero $(T-1)\times(T-1)$ minor, and hence
	\[
	\operatorname{rank}
	\left(
	\tilde{\mathbf D}
	-
	\tilde{\mathbf{\Gamma}}^{-1}\mathbf{D}\tilde{\mathbf{\Gamma}}^{-1'}
	\right)
	\geq T-1.
	\]
	Under the baseline condition $T\geq 2(\bar r+1)$,
	\[
	T-1\geq 2\bar r+1>2\bar r,
	\]
	contradicting the earlier rank bound. Therefore, no observationally equivalent
	representation with $\tilde{\alpha}\neq0$ exists, so $\tilde{\alpha}=0$.
\end{proof}

\newpage
	
	\section{Additional results}\label{Ap:B}
	
	This section collects some intermediate results that are instrumental in deriving the proof of the main identification strategy. The proofs of these results are available in the Online Appendix.
	
	\begin{lemma}
		\label{lem:cofactor-representation-alpha}
		Let $R=\{r_1,\cdots,r_k\}$ and $C=\{c_1,\cdots,c_k\}$, with
		$r_1<\cdots<r_k$, $c_1<\cdots<c_k$, and $R\cap C=\emptyset$. Define
		\[
		\mathbf{J}^\alpha=\mathbf{J}(\alpha,\theta^0)=\mathbf{L\Omega}+\mathbf{\Omega L}'.
		\]
		Then, under Assumptions \ref{as:1}-\ref{as:5},
		\[
		\tilde J_{R,C}(\alpha,\theta^0)
		=
		\sum_{i=1}^{k}\sum_{j=1}^{k}
		(-1)^{i+j}
		J^\alpha_{r_i,c_j}
		\det \left( \mathbf{M}^{\mathbf{\Omega},k-1}_{R-r_i,C-c_j} \right).
		\]
	\end{lemma}
	
	\begin{lemma}
		\label{lem:span-signed-minors}
		Fix any ordered index set
		\[
		S=\{s_1,\cdots,s_{\bar{r}+1}\}.
		\]
		Define
		\[
		u_S(\mathbf{F})
		=
		\left(
		(-1)^i\det(\mathbf{F}_{S-s_i})
		\right)_{i=1}^{\bar{r}+1}.
		\]
		Then, under the normalization in Assumption \ref{as:2},
		\[
		\operatorname{span}
		\left\{
		u_S(\mathbf{F}): \operatorname{vec}(\mathbf{F}_2)\in\mathbb R^{(T-\bar r)\bar r}
		\right\}
		=
		\mathbb R^{\bar r+1}.
		\]
	\end{lemma}
	
	\begin{lemma}
		\label{lem:bilinear-nonzero-polynomial}
		Let $R=\{r_1,\cdots,r_k\}$ and $C=\{c_1,\cdots,c_k\}$, with
		$r_1<\cdots<r_k$, $c_1<\cdots<c_k$, $R\cap C=\emptyset$, and
		$k=\bar r+1$. Suppose $\mathbf{A}\in\mathbb R^{k\times k}$
		is nonzero. Define
		\[
		u_R(\mathbf{F})
		=
		\left(
		(-1)^i\det(\mathbf{F}_{R-r_i})
		\right)_{i=1}^{k},
		\qquad
		u_C(\mathbf{F})
		=
		\left(
		(-1)^j\det(\mathbf{F}_{C-c_j})
		\right)_{j=1}^{k}.
		\]
		Then, under Assumption \ref{as:2},
		\[
		u_R(\mathbf{F})'\mathbf{A}\,u_C(\mathbf{F})
		\]
		is a nonzero polynomial in the unrestricted entries of $\mathbf{F}_2$.
	\end{lemma}
	
	\newpage

\end{appendix}

\clearpage
\setcounter{page}{1}
\onehalfspacing
\begin{center}
	\begin{LARGE}
\textbf{Online Appendix:  Technical Details}
	\end{LARGE}
	\end{center}
\addcontentsline{toc}{section}{Online Appendix}

\section{Proof of remaining results}

\subsection*{Proof of Claim \ref{claim:1}:}
\begin{proof}
	By definition $\mathbf{\Omega} = \mathbf{F}\mathbf{\Psi}\mathbf{F}'+\mathbf{D}$. Using the fact that $\tilde{\mathbf{\Gamma}}^{-1}\mathbf{\Gamma}=\mathbf{I}_T+(\alpha-\tilde{\alpha})\mathbf{L}$ and distributing terms:
	\begin{align*}
		\mathbf{O}&=\left(\mathbf{I}_T+(\alpha-\tilde{\alpha})\mathbf{L}\right)\mathbf{\Omega}\left(\mathbf{I}_T+(\alpha-\tilde{\alpha})\mathbf{L}\right)'-\tilde{\mathbf{D}}= \mathbf{\Omega} - \tilde{\mathbf{D}} + (\alpha - \tilde{\alpha}) \mathbf{J}(\tilde{\alpha}, \theta^0)
	\end{align*}
	\noindent where $\mathbf{J}(\tilde{\alpha}, \theta^0)=
	\mathbf{L}\mathbf{\Omega}+\mathbf{\Omega}\mathbf{L}'+(\alpha-\tilde{\alpha})\mathbf{L}\mathbf{\Omega}\mathbf{L}'$.
\end{proof}

\subsection*{Proof of Proposition \ref{prop:6}}

\begin{proof}
	
	Relabel the polynomials to simplify notation as follows:
	\[
	p(\tilde\alpha)=\tilde J_{R_0,C_0}(\tilde\alpha,\theta^0),
	\qquad
	q(\tilde\alpha)=\tilde J_{R_1,C_1}(\tilde\alpha,\theta^0).
	\]
	
The proof proceeds according to the following strategy. First, we explain why the resultant is a polynomial in $\mathbf F_2$. Second, for every fixed value of $(\alpha,\mathbf{D},\mathbf{\Psi})$ we construct one admissible value of $\mathbf F_2$ for which the resultant is nonzero. The first two steps automatically imply that the resultant is a nonzero polynomial in $\mathbf F_2$. Third, we use Assumption \ref{as:6} to turn this algebraic statement into a probability-zero statement.
	
	\medskip
	\noindent\textit{Step 1: the resultant is a polynomial in the unrestricted factor rows.}
	Throughout this proof, resultants and common roots are understood over $\mathbb C$. Define
	$\mathbf Q=\tilde{\mathbf{\Gamma}}^{-1}\mathbf{\Gamma}$. The entries of this matrix are:
	\[
	Q_{ij}=
	\begin{cases}
		1, & i=j,\\
		(\alpha-\tilde\alpha)\alpha^{i-j-1}, & i>j,\\
		0, & i<j.
	\end{cases}
	\]
	Fix admissible $(\alpha,\mathbf D,\mathbf\Psi)$ with $\alpha\neq0$. Since
	$\mathbf F=(\mathbf I_{\bar r},\mathbf F_2')'$, every entry of
	$\mathbf F\mathbf\Psi\mathbf F'$ is a polynomial in the unrestricted entries of
	$\mathbf F_2$. Also, every entry of $\mathbf Q$ is a polynomial in
	$\tilde\alpha$, with coefficients depending only on $\alpha$. Therefore every entry of $
	\mathbf Q(\mathbf F\mathbf\Psi\mathbf F'+\mathbf D)
	\mathbf Q'$ is a polynomial jointly in $\tilde\alpha$ and the unrestricted entries of $\mathbf F_2$.
	Taking determinants preserves polynomiality, so the minor
	$\det\!\left([\mathbf Q(\mathbf F\mathbf\Psi\mathbf F'+\mathbf D)\mathbf Q']_{R_0,C_0}\right)$
	is a polynomial jointly in $\tilde\alpha$ and the unrestricted entries of
	$\mathbf F_2$. Since $R_0\cap C_0=\emptyset$, the term $\tilde{\mathbf D}_{R_0,C_0}$
	is zero, so this minor equals $\det(\mathbf O_{R_0,C_0})$.
	
	By Propositions~\ref{prop:2}
	and~\ref{prop:3}, $\det(\mathbf O_{R_0,C_0})=(\alpha-\tilde\alpha)\tilde
	J_{R_0,C_0}(\tilde\alpha,\theta^0)$. Hence $p(\tilde\alpha)=\tilde
	J_{R_0,C_0}(\tilde\alpha,\theta^0)$ is obtained by exact cancellation of
	$(\alpha-\tilde\alpha)$ from a polynomial in $\tilde\alpha$, and therefore
	remains a polynomial in $\tilde\alpha$ whose coefficients are polynomial
	functions of the unrestricted entries of $\mathbf F_2$. The identical argument
	applies to $q(\tilde\alpha)=\tilde J_{R_1,C_1}(\tilde\alpha,\theta^0)$. Choose fixed
	formal degree bounds $d_p,d_q\geq 0$ and write
	\[
	p(\tilde\alpha)
	=
	a_0(\mathbf F_2)\tilde\alpha^{d_p}
	+a_1(\mathbf F_2)\tilde\alpha^{d_p-1}
	+\cdots
	+a_{d_p}(\mathbf F_2),\qquad
	q(\tilde\alpha)
	=
	b_0(\mathbf F_2)\tilde\alpha^{d_q}
	+b_1(\mathbf F_2)\tilde\alpha^{d_q-1}
	+\cdots
	+b_{d_q}(\mathbf F_2),
	\]
	where each
	$a_j(\mathbf F_2)$ and each $b_j(\mathbf F_2)$ is a polynomial function of the
	unrestricted entries of $\mathbf F_2$ whose coefficients depend on
	$(\alpha,\mathbf D,\mathbf\Psi)$. For notational simplicity, we suppress the dependence of the coefficients on
	$\mathbf F_2$ and write
	\[
	p(\tilde\alpha)
	=
	a_0\tilde\alpha^{d_p}
	+a_1\tilde\alpha^{d_p-1}
	+\cdots
	+a_{d_p},
	\qquad
	q(\tilde\alpha)
	=
	b_0\tilde\alpha^{d_q}
	+b_1\tilde\alpha^{d_q-1}
	+\cdots
	+b_{d_q}.
	\]
	The Sylvester resultant is the determinant of the Sylvester matrix formed by
	stacking shifted copies of the coefficient rows
	$(a_0,\ldots,a_{d_p})$ and $(b_0,\ldots,b_{d_q})$. Therefore $
	\mathcal R_{\alpha,\mathbf D,\mathbf\Psi}(\mathbf F_2)$ is itself a polynomial in the coefficients $a_j,b_j$, and hence a polynomial in the unrestricted entries of $\mathbf F_2$. It remains to show that this polynomial is not identically zero. We will not compute the
	Sylvester determinant explicitly. Instead, we use the following standard property of the
	resultant. If a specialization\footnote{We use the term specialization to refer to a chosen value of $\mathbf F_2$ within the parametric space.} of $\mathbf F_2$ is chosen so that the specialized polynomial
	$p$ retains its formal degree $d_p$, then
	\begin{equation}\label{eq:resultant-product-didactic}
		\mathcal R_{\alpha,\mathbf D,\mathbf\Psi}(\mathbf F_2)
		=
		\operatorname{lc}(p)^{d_q}
		\prod_{\lambda:p(\lambda)=0} q(\lambda),
	\end{equation}
	where the complex roots of $p$ are counted with multiplicity and $\operatorname{lc}(p)$ is the
	leading coefficient of the specialized $p$ (which will be different from zero by construction). Thus, for such a specialization, the resultant is
	nonzero if $q$ is nonzero at every root of $p$. This is the reason for the construction below:
	we choose $\mathbf F_2$ so that the first minor gives a simple polynomial $p$, and then we
	show that the second minor $q$ does not vanish at any root of that polynomial.
	
	\medskip
	\noindent\textit{Step 2.1: derive the determinant factorization for the first selected minor.}

Recall that $
R_0=\{1,\ldots,\bar r+1\}$, and let $S\subset\{\bar r+2,\ldots,T\}$ have cardinality $\bar r+1$. Order the elements of $S$ increasingly. We first compute the selected block $\mathbf O_{R_0,S}$. Since
$R_0\cap S=\emptyset$, the block $(R_0,S)$ of $\tilde{\mathbf D}$ is zero. Hence
\[
\mathbf O_{R_0,S}
=
\left[
\mathbf Q(\mathbf F\mathbf\Psi\mathbf F'+\mathbf D)\mathbf Q'
\right]_{R_0,S}.
\]
Because $\mathbf Q_{R_0,R_0}$ is a $(\bar{r}+1)\times(\bar{r}+1)$ lower-triangular matrix with unit diagonal, we have $
\det(\mathbf Q_{R_0,R_0})=1$. Therefore premultiplication by $\mathbf Q_{R_0,R_0}^{-1}$ does not change the
determinant: $
\det(\mathbf O_{R_0,S})
=
\det\left(
\mathbf Q_{R_0,R_0}^{-1}\mathbf O_{R_0,S}
\right)$.

We now compute the matrix inside this determinant. First consider the term in $\mathbf O_{R_0,S}$ generated by
$\mathbf F\mathbf\Psi\mathbf F'$:
\[
\left[
\mathbf Q\mathbf F\mathbf\Psi\mathbf F'\mathbf Q'
\right]_{R_0,S}=
(\mathbf Q\mathbf F)_{R_0}
\mathbf\Psi
(\mathbf Q\mathbf F)_S'.
\]
We now rewrite the rows of $\mathbf Q\mathbf F$. Define
\begin{equation}\label{eq:h}
	h_t=\sum_{\ell=1}^t\alpha^{t-\ell}f_\ell,
	\qquad h_0=0.
\end{equation}
Then, the $t$-th row of $\mathbf Q\mathbf F$ can be written as
\[
(\mathbf Q\mathbf F)_{t.}^{'} = f_t
+
(\alpha-\tilde\alpha)
\sum_{\ell=1}^{t-1}\alpha^{t-\ell-1}f_\ell =
\sum_{\ell=1}^t\alpha^{t-\ell}f_\ell
-
\tilde\alpha
\sum_{\ell=1}^{t-1}\alpha^{t-\ell-1}f_\ell
=
h_t-\tilde\alpha h_{t-1}.
\]
where $f_{t} \in \mathbb{R}^{\bar r}$ denotes the $t$-th factor row written as a column vector. Hence, for any set of row indices $R$ we have
\[
(\mathbf Q\mathbf F)_R
=
\begin{bmatrix}
	\left(h_r-\tilde\alpha h_{r-1}\right)^{'}
\end{bmatrix}_{r\in R}.
\]

Apply this identity to the row set $R_0$. Since $R_0$ is the initial block
$\{1,\ldots,\bar r+1\}$ and $\mathbf Q$ is lower-triangular, the rows $R_0$
of $\mathbf Q\mathbf F$ depend only on the rows $R_0$ of $\mathbf F$. Therefore, for $R_0$ this expression simplifies to $
(\mathbf Q\mathbf F)_{R_0}
=
\mathbf Q_{R_0,R_0}\mathbf F_{R_0}$. It follows that
\[
\left[
\mathbf Q\mathbf F\mathbf\Psi\mathbf F'\mathbf Q'
\right]_{R_0,S}
=
\mathbf Q_{R_0,R_0}\mathbf F_{R_0}\mathbf\Psi
\begin{bmatrix}
	\left(h_s-\tilde\alpha h_{s-1}\right)'
\end{bmatrix}_{s\in S}'.
\]
After premultiplication by $\mathbf Q_{R_0,R_0}^{-1}$, this term becomes $
\mathbf F_{R_0}\mathbf\Psi
\begin{bmatrix}
	\left(h_s-\tilde\alpha h_{s-1}\right)'
\end{bmatrix}_{s\in S}'$.

Next consider the term in $\mathbf O_{R_0,S}$ generated by the diagonal matrix
$\mathbf D$:
\[
\left[
\mathbf Q\mathbf D\mathbf Q'
\right]_{R_0,S}.
\]
Since $R_0=\{1,\ldots,\bar r+1\}$ is an initial block and $\mathbf Q$ is lower
triangular, we have that $
\mathbf Q_{R_0,\{1,\ldots,T\}}
=
\begin{bmatrix}
	\mathbf Q_{R_0,R_0} & 0
\end{bmatrix}$. Therefore $
\left[
\mathbf Q\mathbf D\mathbf Q'
\right]_{R_0,S}
=
\mathbf Q_{R_0,R_0}
\operatorname{diag}(d_1,\ldots,d_{\bar r+1})
\mathbf Q_{S,R_0}'$.

Premultiplying by $\mathbf Q_{R_0,R_0}^{-1}$ gives
\[
\mathbf Q_{R_0,R_0}^{-1}
\left[
\mathbf Q\mathbf D\mathbf Q'
\right]_{R_0,S}
=
\operatorname{diag}(d_1,\ldots,d_{\bar r+1})
\mathbf Q_{S,R_0}'.
\]
We now study the structure of this matrix. Its $(r,s)$-entry,
with $r\in R_0$ and $s\in S$, is $
d_r Q_{s,r}
=
d_r(\alpha-\tilde\alpha)\alpha^{s-r-1}$. Since $
s-r-1=(\bar r+1-r)+(s-\bar r-2)$, this entry can be written as $
(\alpha-\tilde\alpha)
\left(\alpha^{\bar r+1-r}d_r\right)
\alpha^{s-\bar r-2}$. Thus the dependence on $r$ and $s$ separates multiplicatively. Hence
\[
\operatorname{diag}(d_1,\ldots,d_{\bar r+1})
\mathbf Q_{S,R_0}'
=
(\alpha-\tilde\alpha)a b_S',
\]
where
\[
a=
\left(
\alpha^{\bar r}d_1,
\alpha^{\bar r-1}d_2,
\ldots,
\alpha d_{\bar r},
d_{\bar r+1}
\right)',\qquad
b_S=
\left(
\alpha^{s-\bar r-2}
\right)_{s\in S}.
\]
Combining this with the factor component gives
\begin{equation}\label{eq:premultiplied-block-R0}
	\mathbf Q_{R_0,R_0}^{-1}\mathbf O_{R_0,S}
	=
	\mathbf F_{R_0}\mathbf\Psi
	\begin{bmatrix}
		\left(h_s-\tilde\alpha h_{s-1}\right)'
	\end{bmatrix}_{s\in S}'
	+
	(\alpha-\tilde\alpha)a b_S'.
\end{equation}
The right-hand side can be written as a product of three square matrices (of dimension $\bar{r}+1$):
\[
\mathbf Q_{R_0,R_0}^{-1}\mathbf O_{R_0,S}
=
\begin{bmatrix}
	\mathbf F_{R_0} & a
\end{bmatrix}
\begin{bmatrix}
	\mathbf\Psi & 0\\
	0 & \alpha-\tilde\alpha
\end{bmatrix}
\begin{bmatrix}
	\left[\left(h_s-\tilde\alpha h_{s-1}\right)'\right]_{s\in S}'\\
	b_S'
\end{bmatrix}.
\]
Indeed, multiplying the three matrices gives exactly the two terms in
\eqref{eq:premultiplied-block-R0}. Taking determinants of the preceding display gives
\[
\det\!\left(\mathbf Q_{R_0,R_0}^{-1}\mathbf O_{R_0,S}\right)
=
\det
\begin{bmatrix}
	\mathbf F_{R_0} & a
\end{bmatrix}
\det
\begin{bmatrix}
	\mathbf\Psi & 0\\
	0 & \alpha-\tilde\alpha
\end{bmatrix}
\det
\begin{bmatrix}
	\left[\left(h_s-\tilde\alpha h_{s-1}\right)'\right]_{s\in S}'\\
	b_S'
\end{bmatrix}.
\]
Since $\mathbf Q_{R_0,R_0}$ is lower-triangular with unit diagonal,
\[
\det(\mathbf O_{R_0,S})
=
(\alpha-\tilde\alpha)
\det(\mathbf\Psi)
\det
\begin{bmatrix}
	\mathbf F_{R_0} & a
\end{bmatrix}
\det
\begin{bmatrix}
	\left[\left(h_s-\tilde\alpha h_{s-1}\right)'\right]_{s\in S}'\\
	b_S'
\end{bmatrix}.
\]
The last determinant is the determinant of the transpose of $
\left[
\begin{array}{cc}
	\left(h_s-\tilde\alpha h_{s-1}\right)' & \alpha^{s-\bar r-2}
\end{array}
\right]_{s\in S}$. Then
\[
\det(\mathbf O_{R_0,S})
=
(\alpha-\tilde\alpha)
\det(\mathbf\Psi)
\det
\begin{bmatrix}
	\mathbf F_{R_0} & a
\end{bmatrix}
\det
\left[
\begin{array}{cc}
	\left(h_s-\tilde\alpha h_{s-1}\right)' & \alpha^{s-\bar r-2}
\end{array}
\right]_{s\in S}.
\]
By Propositions \ref{prop:2} and \ref{prop:3} given that $R_0\cap S=\emptyset$ and $|R_0|=|S|=\bar r+1>\bar r$,
\[
\det(\mathbf O_{R_0,S})
=
(\alpha-\tilde\alpha)
\tilde J_{R_0,S}(\tilde\alpha,\theta^0).
\]
The preceding two displays are polynomial identities in $\tilde\alpha$. Since
$\alpha-\tilde\alpha$ is not the zero polynomial, it can be canceled from both sides.
Thus
\begin{equation}\label{eq:factorization-R0}
	\tilde J_{R_0,S}(\tilde\alpha,\theta^0)
	=
	\det(\mathbf\Psi)
	\det
	\begin{bmatrix}
		\mathbf F_{R_0} & a
	\end{bmatrix}
	\det
	\left[
	\begin{array}{cc}
		\left(h_s-\tilde\alpha h_{s-1}\right)' & \alpha^{s-\bar r-2}
	\end{array}
	\right]_{s\in S}.
\end{equation}
where
\[
a=
\left(
\alpha^{\bar r}d_1,
\alpha^{\bar r-1}d_2,
\ldots,
\alpha d_{\bar r},
d_{\bar r+1}
\right)'
\]

\medskip
\noindent\textit{Step 2.2: choose one admissible factor realization that makes $p$ simple.} We want to construct an arbitrary value of $\mathbf{F}$ that is convenient for the proof. Keep $f_1=e_1,\ldots,f_{\bar r}=e_{\bar r}$ (where $e_j$ is the
$j$-th canonical column vector in $\mathbb R^{\bar r}$), as required by the normalization
$\mathbf F=(\mathbf I_{\bar r},\mathbf F_2')'$ (see Assumption \ref{as:2}). These fixed rows determine
$h_1,\ldots,h_{\bar r}$ (see Equation \eqref{eq:h}). Since all rows $f_t$ with $t>\bar r$ are unrestricted, and since
$\alpha\neq0$, we may choose the later filtered rows and then recover the corresponding
factor rows from $
f_t=h_t-\alpha h_{t-1}$.

Choose a real number $\tau$ and impose
\begin{equation}\label{eq:h-specialization}
	h_{\bar r+1}=0,
	\qquad
	h_{\bar r+1+i}=\tau e_i,
	\quad i=1,\ldots,\bar r,
	\qquad
	h_{2\bar r+2}=0.
\end{equation}
Rows after $2\bar r+2$ do not enter the two selected minors and can be fixed arbitrarily. This
constructs an admissible one-parameter family of values of $\mathbf F_2$, indexed by
$\tau$. Under \eqref{eq:h-specialization}, we have imposed $h_{\bar r+1}=0$. Then, because by \eqref{eq:h} $
h_{\bar r+1}
=
\alpha^{\bar r}f_1+\alpha^{\bar r-1}f_2+\cdots+\alpha f_{\bar r}
+f_{\bar r+1}$,
and the normalization gives $f_j=e_j$ for $j=1,\ldots,\bar r$, we have $
f_{\bar r+1}
=
-(\alpha^{\bar r},\alpha^{\bar r-1},\ldots,\alpha)'$.
Therefore
\[
\begin{bmatrix}
	\mathbf F_{R_0} & a
\end{bmatrix}
=
\begin{bmatrix}
	I_{\bar r} & u\\
	f_{\bar r+1}' & d_{\bar r+1}
\end{bmatrix},
\]
where $
u=
(\alpha^{\bar r}d_1,\alpha^{\bar r-1}d_2,\ldots,\alpha d_{\bar r})'$, and
$
f_{\bar r+1}=(-\alpha^{\bar r},-\alpha^{\bar r-1},\ldots,-\alpha)'$.
By the Schur-complement formula,
\[
\det
\begin{bmatrix}
	\mathbf F_{R_0} & a
\end{bmatrix}
=
d_{\bar r+1}-f_{\bar r+1}'u.
\]
Since $
f_{\bar r+1}'u
=
-\sum_{j=1}^{\bar r}
\alpha^{2(\bar r+1-j)}d_j$, it follows that
\begin{equation}\label{eq:detF}
	\det
	\begin{bmatrix}
		\mathbf F_{R_0} & a
	\end{bmatrix}
	=
	d_{\bar r+1}
	+
	\sum_{j=1}^{\bar r}
	\alpha^{2(\bar r+1-j)}d_j
	=
	\sum_{j=1}^{\bar r+1}
	\alpha^{2(\bar r+1-j)}d_j.
\end{equation}
For $S=C_0=\{\bar r+2,\ldots,2\bar r+2\}$, the specialization
\eqref{eq:h-specialization} gives
\[
h_{\bar r+2}-\tilde\alpha h_{\bar r+1}
=
\tau e_1, \quad h_{2\bar r+2}-\tilde\alpha h_{2\bar r+1}
=
-\tau\tilde\alpha e_{\bar r}
\]
\[
h_{\bar r+1+i}-\tilde\alpha h_{\bar r+i}
=
\tau(e_i-\tilde\alpha e_{i-1}),
\qquad i=2,\ldots,\bar r.
\]
The last column in \eqref{eq:factorization-R0} is $
\left(\alpha^{s-\bar r-2}\right)_{s\in C_0}
=
(1,\alpha,\ldots,\alpha^{\bar r})'$. Then, the determinant in \eqref{eq:factorization-R0} is
\[
\det
\begin{bmatrix}
	\tau e_1' & 1\\
	\tau(e_2-\tilde\alpha e_1)' & \alpha\\
	\vdots & \vdots\\
	\tau(e_{\bar r}-\tilde\alpha e_{\bar r-1})' & \alpha^{\bar r-1}\\
	-\tau\tilde\alpha e_{\bar r}' & \alpha^{\bar r}
\end{bmatrix}=
\tau^{\bar r}
\det
\begin{bmatrix}
	e_1' & 1\\
	(e_2-\tilde\alpha e_1)' & \alpha\\
	\vdots & \vdots\\
	(e_{\bar r}-\tilde\alpha e_{\bar r-1})' & \alpha^{\bar r-1}\\
	-\tilde\alpha e_{\bar r}' & \alpha^{\bar r}
\end{bmatrix}\equiv
\tau^{\bar r} P_{\bar r}(\tilde\alpha)
\]
where $P_{\bar r}(\tilde\alpha)$ denotes the determinant of the matrix in the second expression. This matrix is such that the first $\bar r$ columns have ones on the diagonal and $-\tilde\alpha$ on the subdiagonal. Expanding along the last row gives

\[
P_{\bar r}(\tilde\alpha)
=
\alpha^{\bar r}
+
\tilde\alpha P_{\bar r-1}(\tilde\alpha),
\]
where $P_{\bar r-1}(\tilde\alpha)$ is the analogous determinant of one lower
dimension. Iterating this recursion yields
\[
P_{\bar r}(\tilde\alpha)
=
\alpha^{\bar r}
+\alpha^{\bar r-1}\tilde\alpha
+\cdots
+\alpha\tilde\alpha^{\bar r-1}
+\tilde\alpha^{\bar r}.
\]
Hence
\begin{equation}\label{eq:H-C0}
	\det\left[
	\begin{array}{cc}
		\left(h_s-\tilde\alpha h_{s-1}\right)' & \alpha^{s-\bar r-2}
	\end{array}
	\right]_{s\in C_0}
	=
	\tau^{\bar r}P_{\bar r}(\tilde\alpha),
	\qquad
	P_{\bar r}(\tilde\alpha)
	=
	\sum_{\ell=0}^{\bar r}\alpha^{\bar r-\ell}\tilde\alpha^\ell.
\end{equation}
The sign depends only on the ordering convention and is irrelevant. Combining
\eqref{eq:factorization-R0}, \eqref{eq:detF}, and \eqref{eq:H-C0}, the first selected
minor becomes
\begin{equation}\label{eq:p-tau-formula}
	p_\tau(\tilde\alpha)
	=
	\pm\tau^{\bar r}\det(\mathbf\Psi)
	\left(\sum_{j=1}^{\bar r+1}\alpha^{2(\bar r+1-j)}d_j\right)P_{\bar r}(\tilde\alpha).
\end{equation}
where $p_\tau(\tilde\alpha)$ denotes the specialization of
$p(\tilde\alpha)=\tilde J_{R_0,C_0}(\tilde\alpha,\theta^0)$
under the above choice of the rows of $\mathbf F_2$ (see \eqref{eq:h-specialization}). The scalar multiplying $P_{\bar r}$ is nonzero whenever $\tau\neq0$\footnote{$\tau$ is an arbitrary constant so it can always be chosen to be different from zero.}, because
$\det(\mathbf\Psi)>0$, $d_j>0$ for all $j$ and $\alpha^{2(\bar r+1-j)} > 0$ as well. Since $\alpha\neq0$, $P_{\bar r}$ has degree $\bar r$
and leading coefficient one. Moreover, \eqref{eq:factorization-R0} shows that the formal
degree of $p$ is at most $\bar r$. Since the specialization \eqref{eq:p-tau-formula} has degree
$\bar r$ for every $\tau\neq0$, the formal degree of $p$ is exactly $\bar r$, and $p_\tau$ retains
that formal degree.

The roots of $p_\tau$ are therefore exactly the roots of $P_{\bar r}$, since
the scalar multiplying $P_{\bar r}$ in \eqref{eq:p-tau-formula} is nonzero.
We now characterize those roots. Since
\[
P_{\bar r}(\tilde\alpha)
=
\alpha^{\bar r}
+\alpha^{\bar r-1}\tilde\alpha
+\cdots
+\alpha\tilde\alpha^{\bar r-1}
+\tilde\alpha^{\bar r},
\]
multiplying by $\tilde\alpha-\alpha$ gives
\[
\begin{aligned}
	(\tilde\alpha-\alpha)P_{\bar r}(\tilde\alpha)
	&=
	\tilde\alpha
	\left(
	\alpha^{\bar r}
	+\alpha^{\bar r-1}\tilde\alpha
	+\cdots
	+\alpha\tilde\alpha^{\bar r-1}
	+\tilde\alpha^{\bar r}
	\right)
	-\alpha
	\left(
	\alpha^{\bar r}
	+\alpha^{\bar r-1}\tilde\alpha
	+\cdots
	+\alpha\tilde\alpha^{\bar r-1}
	+\tilde\alpha^{\bar r}
	\right) =
	\tilde\alpha^{\bar r+1}-\alpha^{\bar r+1},
\end{aligned}
\]
because all intermediate terms cancel. Equivalently, for
$\tilde\alpha\neq\alpha$,
\[
P_{\bar r}(\tilde\alpha)
=
\frac{\tilde\alpha^{\bar r+1}-\alpha^{\bar r+1}}
{\tilde\alpha-\alpha}.
\]
Moreover,
\[
P_{\bar r}(\alpha)
=
(\bar r+1)\alpha^{\bar r}\neq0,
\]
because $\alpha\neq0$. Hence $\alpha$ itself is not a root of $P_{\bar r}$. Let $\lambda$ be any root of $P_{\bar r}$. Since $\lambda\neq\alpha$, the identity
above implies $
\lambda^{\bar r+1}-\alpha^{\bar r+1}=0$,
or equivalently, $
\lambda^{\bar r+1}=\alpha^{\bar r+1}$.

Also, $\lambda\neq0$, because otherwise the last display would imply
$\alpha^{\bar r+1}=0$, contradicting $\alpha\neq0$. Thus every root $\lambda$ of
$P_{\bar r}$ satisfies
\begin{equation}\label{eq:Pm-root}
	\lambda^{\bar r+1}=\alpha^{\bar r+1},
	\qquad \lambda\neq\alpha,
	\qquad \lambda\neq0.
\end{equation}
To complete the construction, it remains to choose $\tau$ so that
$q_\tau(\lambda)\neq0$ for every root $\lambda$ of $P_{\bar r}$.
Here $q_\tau(\tilde\alpha)$ denotes the specialization of
$q(\tilde\alpha)=\tilde J_{R_1,C_1}(\tilde\alpha,\theta^0)$
under the same one-parameter choice of the rows of $\mathbf F_2$.

\medskip
\noindent\textit{Step 2.3: show that $q_\tau$ is nonzero at the roots of $p_\tau$.}
Given the derivations in step 2.1, we can write the submatrix of $\mathbf{O}$ associated with rows and columns $R_1$ and $C_1$ as:
\[
\mathbf O_{R_1,C_1}
=
\begin{bmatrix}\left(h_r-\tilde\alpha h_{r-1}\right)'\end{bmatrix}_{r\in R_1}
\mathbf\Psi
\begin{bmatrix}\left(h_s-\tilde\alpha h_{s-1}\right)'\end{bmatrix}_{s\in C_1}'
+
\mathbf E_{R_1,C_1}(\tilde\alpha),
\qquad
\mathbf E(\tilde\alpha)
=
\mathbf Q\mathbf D\mathbf Q'.
\]
Given the specialization of $\mathbf{F}$ described in \eqref{eq:h-specialization}, $q_\tau(\lambda)$ is a polynomial in $\tau$. We compute only the coefficient of the highest possible power of $\tau$ in
$q_\tau(\tilde\alpha)$. This is enough for the argument: if this coefficient is
nonzero after evaluating at a root $\lambda$ of $p_\tau(\lambda)$, then
$q_\tau(\lambda)$ is not the zero polynomial in $\tau$ (and therefore in $\mathbf{F}_{2}$). This would mean that there are only finitely many values of $\tau$ under which the roots of the polynomials $p_\tau(\lambda)$ and $q_\tau(\lambda)$ coincide, so as $\tau$ is a free parameter, we can choose any other value to get the desired conclusion. This shows that there exists a parameterization of $\mathbf{F}$ (see
\eqref{eq:h-specialization}) such that the resultant $\mathcal
R_{\alpha,\mathbf D,\mathbf\Psi}(\mathbf F_2)
=
\mathcal R_{\alpha,\mathbf D,\mathbf\Psi}(\mathbf F_2)$ is nonzero. This implies that
$\mathcal R_{\alpha,\mathbf D,\mathbf\Psi}(\mathbf F_2)
=
\mathcal R_{\alpha,\mathbf D,\mathbf\Psi}(\mathbf F_2)$ is a nonzero polynomial in
$\mathbf{F}_2$, so the set of values of $\mathbf{F}_{2}$ under which the roots
of $\tilde J_{R_0,C_0}(\tilde\alpha,\theta^0)$ and $\tilde
J_{R_1,C_1}(\tilde\alpha,\theta^0)$ coincide is contained in the zero set of
a nonzero polynomial in $\mathbf{F}_{2}$.

Under \eqref{eq:h-specialization}, the rows of the transformed factor
matrix $\mathbf Q\mathbf F$ indexed by $R_1=\{2,\ldots,\bar r+2\}$ are as
follows. For $2\le t\le \bar r$,
\[
h_t-\tilde\alpha h_{t-1}
=
\sum_{\ell=1}^{t-1}
\alpha^{t-\ell-1}(\alpha-\tilde\alpha)e_\ell
+
e_t.
\]
The two boundary rows are
\[
h_{\bar r+1}-\tilde\alpha h_{\bar r}
=
-\tilde\alpha
\sum_{\ell=1}^{\bar r}
\alpha^{\bar r-\ell}e_\ell,\qquad
h_{\bar r+2}-\tilde\alpha h_{\bar r+1}
=
\tau e_1.
\]
Thus, among the rows of $\mathbf Q\mathbf F$ indexed by $R_1$, exactly one
row depends on $\tau$. For the factor rows associated with the column indices
$C_1=\{1,\bar r+3,\ldots,2\bar r+2\}$, the specialization
\eqref{eq:h-specialization} gives
\[
\begin{aligned}
	h_1-\tilde\alpha h_0
	&=
	e_1,\\
	h_{\bar r+2+i}-\tilde\alpha h_{\bar r+1+i}
	&=
	\tau(e_{i+1}-\tilde\alpha e_i),
	\qquad 1\le i\le \bar r-1,\\
	h_{2\bar r+2}-\tilde\alpha h_{2\bar r+1}
	&=
	-\tilde\alpha \tau e_{\bar r}.
\end{aligned}
\]
Thus, among the rows of $\mathbf Q\mathbf F$ associated with $C_1$,
the first row is $e_1$, while the remaining $\bar r$ rows are of order
$\tau$. Equivalently, these rows are
\[
e_1,
\quad
\tau(e_2-\tilde\alpha e_1),
\quad
\tau(e_3-\tilde\alpha e_2),
\quad
\ldots,
\quad
\tau(e_{\bar r}-\tilde\alpha e_{\bar r-1}),
\quad
-\tau\tilde\alpha e_{\bar r},
\]
with the middle terms omitted when $\bar r=1$.
Write
\[
A_R(\tilde\alpha)
=
\begin{bmatrix}
	\left(h_r-\tilde\alpha h_{r-1}\right)'
\end{bmatrix}_{r\in R_1},
\qquad
A_C(\tilde\alpha)
=
\begin{bmatrix}
	\left(h_s-\tilde\alpha h_{s-1}\right)'
\end{bmatrix}_{s\in C_1}.
\]
Then
\[
\mathbf O_{R_1,C_1}
=
A_R(\tilde\alpha)\mathbf\Psi A_C(\tilde\alpha)'
+
\mathbf E_{R_1,C_1}(\tilde\alpha),
\]
The matrix $A_R(\tilde\alpha)\mathbf\Psi A_C(\tilde\alpha)'$ has rank at most
$\bar r$. Since it is an $(\bar r+1)\times(\bar r+1)$ matrix, its
determinant is zero. We use the minor expansion for the determinant of a sum. If
$M$ and $E$ are $n\times n$ matrices, then
\[
\det(M+E)
=
\sum_{k=0}^n
\sum_{\substack{I,J\subseteq\{1,\ldots,n\}\\ |I|=|J|=k}}
(-1)^{\sum_{i\in I}i+\sum_{j\in J}j}
\det(E_{I,J})\det(M_{I^c,J^c}),
\]
where $E_{I,J}$ denotes the submatrix of $E$ with rows $I$ and columns
$J$, and $M_{I^c,J^c}$ denotes the complementary submatrix of $M$. Apply this identity with
\[
M
=
A_R(\tilde\alpha)\mathbf\Psi A_C(\tilde\alpha)',
\qquad
E
=
\mathbf E_{R_1,C_1}(\tilde\alpha),
\qquad
n=\bar r+1.
\]
The $k=0$ term is $
\det\left(A_R(\tilde\alpha)\mathbf\Psi A_C(\tilde\alpha)'\right)$, which is zero because
$A_R(\tilde\alpha)\mathbf\Psi A_C(\tilde\alpha)'$ has rank at most
$\bar r$, while it is an $(\bar r+1)\times(\bar r+1)$ matrix. Hence the
contribution using only entries from the factor component vanishes.

The terms using exactly one entry from
$\mathbf E_{R_1,C_1}(\tilde\alpha)$ are the $k=1$ terms in the same
expansion. They are obtained by choosing one position $(i,j)$, using
$E_{r_i,c_j}(\tilde\alpha)$ at that position, and using the factor
component for all remaining rows and columns. Therefore these terms are
\begin{equation}\label{eq:CB-top}
	\sum_{i,j=1}^{\bar r+1}
	(-1)^{i+j}
	E_{r_i,c_j}(\tilde\alpha)
	\det\left(
	A_{R_1\setminus\{r_i\}}(\tilde\alpha)
	\mathbf\Psi
	A_{C_1\setminus\{c_j\}}(\tilde\alpha)'
	\right),
\end{equation}
where $r_i$ is the $i$-th element of $R_1$ and $c_j$ is the
$j$-th element of $C_1$, and
\[
A_{R_1\setminus\{r_i\}}(\tilde\alpha)
=
\begin{bmatrix}
	\left(h_r-\tilde\alpha h_{r-1}\right)'
\end{bmatrix}_{r\in R_1,\ r\neq r_i},
\qquad
A_{C_1\setminus\{c_j\}}(\tilde\alpha)
=
\begin{bmatrix}
	\left(h_s-\tilde\alpha h_{s-1}\right)'
\end{bmatrix}_{s\in C_1,\ s\neq c_j}.
\]
Since these two matrices are $\bar r\times \bar r$, each determinant in
\eqref{eq:CB-top} factors as
\[
\det\left(
A_{R_1\setminus\{r_i\}}(\tilde\alpha)
\mathbf\Psi
A_{C_1\setminus\{c_j\}}(\tilde\alpha)'
\right)
=
\det\left(A_{R_1\setminus\{r_i\}}(\tilde\alpha)\right)
\det(\mathbf\Psi)
\det\left(A_{C_1\setminus\{c_j\}}(\tilde\alpha)\right).
\]

For any $k\geq2$, the complementary minor $\det(M_{I^c,J^c})$ involves at most
$\bar r+1-k\leq\bar r-1$ rows from $A_R$ and at most $\bar r+1-k\leq\bar r-1$
rows from $A_C$. Since $A_R$ contains exactly one row that has degree one in $\tau$ and
$A_C$ contains $\bar r$ rows that have degree one in $\tau$, the $\tau$-degree of
$\det(M_{I^c,J^c})$ is at most $1+(\bar r+1-k)=\bar r+2-k\leq\bar r$.
Because the entries of $\mathbf E(\tilde\alpha)=\mathbf Q\mathbf D\mathbf Q'$ do
not depend on $\tau$, every $k\geq2$ term in the expansion has total
$\tau$-degree at most $\bar r$. Hence the coefficient of $\tau^{\bar r+1}$ can
only come from the one-$\mathbf E$ terms in \eqref{eq:CB-top}.

We now identify which summands in \eqref{eq:CB-top} actually contribute to $\tau^{\bar r+1}$. First, the determinant involving the $C_1$ rows must retain
all $\bar r$ rows that have degree one in $\tau$. This happens only when the deleted
$C_1$ row is the first one, corresponding to the original column index $1$.
Thus we must have $j=1$. After deleting that row, the remaining $C_1$ rows are
\[
\tau(e_2-\tilde\alpha e_1)',
\quad
\tau(e_3-\tilde\alpha e_2)',
\quad
\ldots,
\quad
\tau(e_{\bar r}-\tilde\alpha e_{\bar r-1})',
\quad
-\tau\tilde\alpha e_{\bar r}'.
\]
Thus, $A_{C_1\setminus\{c_1\}}(\tilde\alpha)$ is upper-triangular with diagonal terms equal to $-\tau\tilde\alpha$. Then
\[
\det\left(A_{C_1\setminus\{c_1\}}(\tilde\alpha)\right) = \tau^{\bar r}(-\tilde\alpha)^{\bar r}.
\]
Second, the determinant involving the $R_1$ rows must retain the unique row of
degree one in $\tau$, namely $\tau e_1'$. Therefore the deleted row cannot be the last
row in $R_1$. Hence $i=1,\ldots,\bar r$.

For such $i$, we have row indices $r_i=i+1$ and column index $c_1=1$. Since $
\mathbf E(\tilde\alpha)
=
\mathbf Q\mathbf D\mathbf Q'$ its $(a,b)$-entry is
\[
E_{a,b}(\tilde\alpha)
=
\sum_{m=1}^T
Q_{a,m}(\tilde\alpha)d_mQ_{b,m}(\tilde\alpha).
\]
Taking $b=1$, we obtain
\[
E_{a,1}(\tilde\alpha)
=
\sum_{m=1}^T
Q_{a,m}(\tilde\alpha)d_mQ_{1,m}(\tilde\alpha).
\]
The first row of $\mathbf Q$ has only one nonzero entry,
namely $Q_{1,1}(\tilde\alpha)=1$. Hence all terms in the preceding sum
vanish except the term with $m=1$, and therefore $
E_{a,1}(\tilde\alpha)
=
Q_{a,1}(\tilde\alpha)d_1$.

Since $r_i=i+1$, this gives $
E_{r_i,1}(\tilde\alpha)
=
E_{i+1,1}(\tilde\alpha)
=
Q_{i+1,1}(\tilde\alpha)d_1$. Recall that $
Q_{i+1,1}(\tilde\alpha)
=
(\alpha-\tilde\alpha)\alpha^{i-1}$. Therefore
\[
E_{r_i,1}(\tilde\alpha)
=
(\alpha-\tilde\alpha)d_1\alpha^{i-1}.
\]
It remains to compute the $R_1$-side determinant in the summands with
$j=1$ and $i=1,\ldots,\bar r$. Let $H$ be the $\bar r\times\bar r$
matrix whose rows are $h_1',\ldots,h_{\bar r}'$. Since $H$ is lower
triangular with unit diagonal, $\det(H)=1$. Hence right multiplication by
$H^{-1}$ does not change the determinant. Moreover, because $HH^{-1}=I$,
we have $h_t'H^{-1}=e_t'$ for $t=1,\ldots,\bar r$. Therefore
\[
\left(h_{t+1}-\tilde\alpha h_t\right)'H^{-1}
=
e_{t+1}'-\tilde\alpha e_t',\quad t=1,\ldots,\bar r-1,
\]
and
\[
(-\tilde\alpha h_{\bar r})'H^{-1}
=
-\tilde\alpha e_{\bar r}',
\qquad
(\tau h_1)'H^{-1}
=
\tau e_1'.
\]

Thus, after deleting the $i$-th row among the first $\bar r$ rows, the
relevant determinant is $\det\!\begin{bmatrix} e_2-\tilde\alpha e_1 & \cdots & e_i-\tilde\alpha e_{i-1} & e_{i+2}-\tilde\alpha e_{i+1} & \cdots & -\tilde\alpha e_{\bar r} & \tau e_1 \end{bmatrix}'$. Expanding along the last row, whose only nonzero entry is $\tau$ in the
first column, gives
\[
\det\left(A_{R_1\setminus\{r_i\}}(\tilde\alpha)\right)
=
\tau(-1)^{\bar r+1}\det(B_i),
\]
where $B_i$ is the $(\bar r-1)\times(\bar r-1)$ matrix obtained by
deleting the last row and the first column. The sign
$(-1)^{\bar r+1}$ is the cofactor sign of the entry in position $(\bar r,1)$. The matrix $B_i$ has block triangular form
\[
B_i=
\begin{bmatrix}
	L_{i-1} & 0\\
	0 & U_{\bar r-i}
\end{bmatrix}.
\]
Here $L_{i-1}$ comes from the rows above the deleted $i$-th row, namely $
(e_2-\tilde\alpha e_1)',\,
(e_3-\tilde\alpha e_2)',\,
\ldots,\,
(e_i-\tilde\alpha e_{i-1})'$ after the first column has been deleted. Thus $L_{i-1}$ is lower
triangular with diagonal entries all equal to $1$, so $
\det(L_{i-1})=1$.

Similarly, $U_{\bar r-i}$ comes from the rows below the deleted $i$-th
row, namely $
(e_{i+2}-\tilde\alpha e_{i+1})',\,
(e_{i+3}-\tilde\alpha e_{i+2})',\,
\ldots,\,
(e_{\bar r}-\tilde\alpha e_{\bar r-1})',\,
-\tilde\alpha e_{\bar r}'$, again after the first column has been deleted. Hence $U_{\bar r-i}$ is
upper triangular with diagonal entries all equal to $-\tilde\alpha$, so $
\det(U_{\bar r-i})=(-\tilde\alpha)^{\bar r-i}$.
Therefore
\[
\det(B_i)
=
\det(L_{i-1})\det(U_{\bar r-i})
=
(-\tilde\alpha)^{\bar r-i}.
\]
Substituting this into the cofactor expansion yields $
\det\left(A_{R_1\setminus\{r_i\}}(\tilde\alpha)\right)
=
(-1)^{\bar r+1}\tau(-\tilde\alpha)^{\bar r-i}$. Finally, $
(-1)^{\bar r+1}(-\tilde\alpha)^{\bar r-i}
=
(-1)^{\bar r+1}(-1)^{\bar r-i}\tilde\alpha^{\bar r-i}
=
(-1)^{2\bar r+1-i}\tilde\alpha^{\bar r-i}
=
(-1)^{i-1}\tilde\alpha^{\bar r-i}$.
Hence
\[
\det\left(A_{R_1\setminus\{r_i\}}(\tilde\alpha)\right)
=
\tau(-1)^{i-1}\tilde\alpha^{\bar r-i},
\qquad i=1,\ldots,\bar r.
\]
We can now combine the pieces. For $i=1,\ldots,\bar r$ and $j=1$, the
corresponding summand in \eqref{eq:CB-top} contributes $
(-1)^{i+1}
E_{r_i,1}(\tilde\alpha)
\det\left(A_{R_1\setminus\{r_i\}}(\tilde\alpha)\right)
\det(\mathbf\Psi)
\det\left(A_{C_1\setminus\{c_1\}}(\tilde\alpha)\right)$.

Using $
E_{r_i,1}(\tilde\alpha)
=
(\alpha-\tilde\alpha)d_1\alpha^{i-1},\
\det\left(A_{R_1\setminus\{r_i\}}(\tilde\alpha)\right)
=
\tau(-1)^{i-1}\tilde\alpha^{\bar r-i}$, and
$\det\left(A_{C_1\setminus\{c_1\}}(\tilde\alpha)\right)
=
\tau^{\bar r}(-\tilde\alpha)^{\bar r}$,
the contribution of the $i$-th summand to the coefficient of
$\tau^{\bar r+1}$ in $\det(\mathbf O_{R_1,C_1})$ is
\[
(\alpha-\tilde\alpha)
d_1\det(\mathbf\Psi)
(-\tilde\alpha)^{\bar r}
\alpha^{i-1}\tilde\alpha^{\bar r-i},
\]
as $
(-1)^{i+1}(-1)^{i-1}=1$. Summing over $i=1,\ldots,\bar r$, the coefficient of
$\tau^{\bar r+1}$ in $\det(\mathbf O_{R_1,C_1})$ is
\[
(\alpha-\tilde\alpha)
d_1\det(\mathbf\Psi)
(-\tilde\alpha)^{\bar r}
\sum_{i=1}^{\bar r}
\alpha^{i-1}\tilde\alpha^{\bar r-i}.
\]
Since $R_1\cap C_1=\emptyset$ and $|R_1|=|C_1|=\bar r+1>\bar r$,
Propositions~\ref{prop:2} and~\ref{prop:3} give
$\det(\mathbf O_{R_1,C_1})=(\alpha-\tilde\alpha)\tilde
J_{R_1,C_1}(\tilde\alpha,\theta^0)=(\alpha-\tilde\alpha)q_\tau(\tilde\alpha)$.
Therefore the coefficient of $\tau^{\bar r+1}$ in $q_\tau(\tilde\alpha)$ is
\begin{equation}\label{eq:q-leading}
	d_1\det(\mathbf\Psi)(-\tilde\alpha)^{\bar r}
	\sum_{i=1}^{\bar r}
	\alpha^{i-1}\tilde\alpha^{\bar r-i}.
\end{equation}
Let
\[
S_{\bar r}(\tilde\alpha)
=
\sum_{i=1}^{\bar r}
\alpha^{i-1}\tilde\alpha^{\bar r-i}
=
\tilde\alpha^{\bar r-1}
+\alpha\tilde\alpha^{\bar r-2}
+\cdots
+\alpha^{\bar r-1}.
\]
Then, for every $\tilde\alpha$,
\[
(\tilde\alpha-\alpha)S_{\bar r}(\tilde\alpha)
=
\tilde\alpha^{\bar r}-\alpha^{\bar r},
\]
because all intermediate terms cancel in the difference
$\tilde\alpha S_{\bar r}(\tilde\alpha)-\alpha S_{\bar r}(\tilde\alpha)$. Now let $\lambda$ be a root of $P_{\bar r}$. By \eqref{eq:Pm-root},
\[
\lambda^{\bar r+1}=\alpha^{\bar r+1},
\qquad
\lambda\neq\alpha,
\qquad
\lambda\neq0.
\]
Evaluate $S_{\bar r}(\tilde\alpha)$ at $\tilde{\alpha}=\lambda$. Since $\lambda\neq\alpha$, we may divide the preceding identity by
$\lambda-\alpha$, obtaining\footnote{Another way of seeing this is that the system $\Big\{\begin{matrix}
		\lambda^{\bar r+1} = \alpha^{\bar r+1}\\
		\lambda^{\bar r} = \alpha^{\bar r}
	\end{matrix}$ does not have any solution in $\mathbb{C}$ when $\lambda\neq\alpha$.}
\[
S_{\bar r}(\lambda)
=
\frac{\lambda^{\bar r}-\alpha^{\bar r}}{\lambda-\alpha}.
\]
Using $\lambda^{\bar r+1}=\alpha^{\bar r+1}$ and $\lambda\neq0$, we have $
\lambda^{\bar r}
=
\frac{\alpha^{\bar r+1}}{\lambda}$. Therefore
\[
S_{\bar r}(\lambda)
=
\frac{\alpha^{\bar r+1}/\lambda-\alpha^{\bar r}}{\lambda-\alpha}
=
\frac{\alpha^{\bar r}(\alpha-\lambda)}{\lambda(\lambda-\alpha)}
=
-\frac{\alpha^{\bar r}}{\lambda} \neq 0.
\]

Now evaluate $q_\tau$ at this fixed root $\lambda$. Since the coefficient of
$\tau^{\bar r+1}$ in $q_\tau(\tilde\alpha)$ is given by
\eqref{eq:q-leading}, we can write $q_\tau(\lambda)$ as a polynomial in
$\tau$ of the form
\[
q_\tau(\lambda)
=
\left[
\pm d_1\det(\mathbf\Psi)(-\lambda)^{\bar r}
S_{\bar r}(\lambda)
\right]\tau^{\bar r+1}
+
\sum_{k=0}^{\bar r}c_k(\lambda)\tau^k,
\]
where the coefficients $c_k(\lambda)$ collect all determinant terms of degree
at most $\bar r$ in $\tau$. Substituting $
S_{\bar r}(\lambda)
=
-\frac{\alpha^{\bar r}}{\lambda}$ gives
\[
q_\tau(\lambda)
=
\left[
\pm \alpha^{\bar r}d_1\det(\mathbf\Psi)\lambda^{\bar r-1}
\right]\tau^{\bar r+1}
+
\sum_{k=0}^{\bar r}c_k(\lambda)\tau^k.
\]
The coefficient multiplying $\tau^{\bar r+1}$ is nonzero because
$\alpha\neq0$, $d_1>0$, $\det(\mathbf\Psi)>0$, and $\lambda\neq0$. Hence
$q_\tau(\lambda)$ is not the zero polynomial in $\tau$. There are only finitely many roots $\lambda$ of $P_{\bar r}$. For each such
root, $q_\tau(\lambda)$ is a nonzero polynomial in $\tau$, and therefore has
only finitely many zeros.

Hence the condition $q_\tau(\lambda)=0$ excludes only finitely many values of $\tau$ for each root $\lambda$ of
$P_{\bar r}$. Choose a real $\tau\neq0$ outside the finite union of these
exceptional sets. Then $
q_\tau(\lambda)\neq0$ for every root $\lambda$ of $P_{\bar r}$. At this admissible choice of $\mathbf F_2$, $p_\tau$ retains its formal degree,
and every root of $p_\tau$ is a root of $P_{\bar r}$. Since $q_\tau$ is nonzero
at every root of $P_{\bar r}$, the polynomials $p_\tau$ and $q_\tau$ have no
common complex root. By \eqref{eq:resultant-product-didactic}, $
\operatorname{Res}_{\tilde\alpha}(p_\tau,q_\tau)\neq0$.

Thus $\mathcal R_{\alpha,\mathbf D,\mathbf\Psi}(\mathbf F_2)$ is not identically zero as a
polynomial in the unrestricted entries of $\mathbf F_2$.

	\medskip
	\noindent\textit{Step 3: convert nonzeroness of the resultant into an almost-sure statement.}
	If $\tilde J_{R_0,C_0}(\tilde\alpha,\theta^0)$ and $\tilde J_{R_1,C_1}(\tilde\alpha,\theta^0)$ share a root, then their resultant must be zero. Hence the common-root event
	is contained in $
	\left\{\mathcal R_{\alpha,\mathbf D,\mathbf\Psi}(\mathbf F_2)=0\right\}$.
	
	By Lemma \ref{lemma:3}, the zero set of this nonzero polynomial has Lebesgue
	measure zero. With $(\alpha,\mathbf D,\mathbf\Psi)$ held fixed, Assumption
	\ref{as:6} and Lemma \ref{lemma:2} imply
	\[
	\Pr\left(\mathcal R_{\alpha,\mathbf D,\mathbf\Psi}(\mathbf F_2)=0\right)=0.
	\]
	Therefore the two selected polynomials cannot share a root almost surely, with
	the fixed admissible $(\alpha,\mathbf D,\mathbf\Psi)$ satisfying $\alpha\neq0$.
\end{proof}

	\subsection*{Proof of Theorem \ref{thm:3}:}

\begin{proof} The definition of $\mathbf{O}$ is the same as in Section 3 (for the case of $\bar{r}=2$) and Proposition \ref{prop:1} still holds, so the determinant of any $3\times3$ diagonal exclusion submatrix of $\mathbf{O}$ must be equal to zero. Moreover, Propositions \ref{prop:2} and \ref{prop:3} hold, so we can still write diagonal exclusion minors as $\det\left(\mathbf{M}^{\mathbf{O},3}_{R,C}\right)=(\alpha-\tilde{\alpha})\tilde{J}_{R,C}\left(\tilde{\alpha},\theta^0\right) =0$.
	
We show that $\alpha$ is almost surely identified when $T=6$. A larger $T$ can only help identification. Let $ T = 6 $ and $ R = (1, 2, 3) $, $ C = (4, 5, 6) $. The diagonal
exclusion minor associated with these rows and columns is given by:
\begin{equation}\label{eq333}
	\begin{aligned}
		\det\left(\mathbf{M}^{\mathbf{O},3}_{(1,2,3),(4,5,6)}\right) &= (\alpha - \tilde{\alpha})\left( \Psi_{1,1}\Psi_{2,2} - \Psi_{1,2}^2 \right)(\tilde{\alpha} - 1)(\tilde{\alpha} f_4 + 1) \\
		&\quad \times \left( d_3 f_1 - \alpha d_2 f_1 - d_3 f_2 f_{\gamma} + \alpha^2 d_1 f_2 - \alpha^2 d_1 f_3 + \alpha d_2 f_3 f_{\gamma} \right)= 0
	\end{aligned}
\end{equation}
where, in this case, $ \tilde{J}_{(1,2,3),(4,5,6)}\left(\tilde{\alpha},\theta^0\right) $ is a polynomial in $ \tilde{\alpha} $ of degree 2.
The proof is similar to the case of $ \bar{r} = 1 $ in Section \ref{sbs:33}.
Note that for \eqref{eq333} to hold, one of the following five alternatives must be true:

\begin{caseprimeprime}\label{casep:1}
	$\Psi_{1,2}^2 = \Psi_{1,1}\Psi_{2,2}$
\end{caseprimeprime}

\begin{caseprimeprime}\label{casep:2}
	$(d_3 f_1 - \alpha d_2 f_1 - d_3 f_2 f_{\gamma} + \alpha^2 d_1 f_2 - \alpha^2 d_1 f_3 + \alpha d_2 f_3 f_{\gamma})=0$
\end{caseprimeprime}

\begin{caseprimeprime}\label{casep:3}
	$\tilde{\alpha}=1$
\end{caseprimeprime}

\begin{caseprimeprime}\label{casep:4}
	$\tilde{\alpha}=-\frac{1}{f_4}$
\end{caseprimeprime}

\begin{caseprimeprime}\label{casep:5}
	$\tilde{\alpha}=\alpha$
\end{caseprimeprime}

Cases \ref{casep:1} and \ref{casep:2} correspond to the case in which $\tilde{J}_{(1,2,3),(4,5,6)}$ is the zero polynomial, so that all the coefficients are trivially equal to zero. Cases \ref{casep:3} and \ref{casep:4} correspond to the case in which we can find a root $\tilde{\alpha}^{\star}$ that solves $\tilde{J}_{(1,2,3),(4,5,6)}\left(\tilde{\alpha}^{\star},\theta^0\right)=0$. However, once $\tilde{\alpha}$ is pinned down in this way, we have no other free parameter, so the same root must satisfy all the zero-determinant conditions (Proposition \ref{prop:1}). Case \ref{casep:5} corresponds to the case where identification of $\alpha$ is achieved.

We show that Cases \ref{casep:1}-\ref{casep:4} can occur only on probability-zero sets of factor realizations, apart from a Lebesgue-null subset of the parameter space. Hence, outside this exceptional parameter set, any admissible observationally equivalent parameterization must satisfy Case~\ref{casep:5}, namely $\tilde{\alpha}=\alpha$, with probability one over the unrestricted factor realizations.

\textbf{Case \ref{casep:1}}: $\Psi_{1,2}^2 = \Psi_{1,1}\Psi_{2,2}$

This implies that $ \det(\mathbf{\Psi}) = 0 $, which is analogous to the case where $ \Psi = 0 $ in subsection \ref{sbs:33}.
This possibility is excluded because we assume that $ \mathbf{\Psi} $ is PD.

\textbf{Case \ref{casep:2}}: $d_3 f_1 - \alpha d_2 f_1 - d_3 f_2 f_{\gamma} + \alpha^2 d_1 f_2 - \alpha^2 d_1 f_3 + \alpha d_2 f_3 f_{\gamma}=0$

We can rewrite this condition as $
(d_{3}-\alpha d_{2}) f_{1} + (\alpha^{2}d_{1}-d_{3}f_{\gamma})f_{2} + (\alpha d_2 f_{\gamma}-\alpha^2 d_{1})f_{3} = 0$.

This expression characterizes a region of the parameter space where our proof cannot guarantee identification, which corresponds to a restriction of Lebesgue-zero measure. In particular, this region corresponds to the joint restrictions $d_{2}f_{\gamma} = \alpha d_{1}$ and $d_{3}=\alpha d_{2}$. Thus, unless these restrictions hold jointly, the expression defines a plane in which the three factors $f_1,f_2,f_3$ must lie. Therefore, by Lemmas \ref{lemma:2} and \ref{lemma:3} this term can only be equal to zero in a set of factor realizations that occur with probability zero.

\textbf{Case \ref{casep:3}}: $\tilde{\alpha}=1$

We next show that the candidate value $ \tilde{\alpha}=1 $ cannot satisfy all diagonal-exclusion restrictions with positive probability over the factor realizations, except possibly on a Lebesgue-null subset of the parameter space. If we set $ \tilde{\alpha} = 1 $ in order to satisfy \eqref{eq333},
this same $ \tilde{\alpha} $
must satisfy all the zero-determinant conditions (Proposition \ref{prop:1}).
However, this cannot occur with positive probability except in a low-dimensional subset of the parameter space. To see this, consider the diagonal exclusion minor
associated with rows $ R = (1, 5, 6) $ and columns $ C = (2, 3, 4) $.
By Propositions \ref{prop:1},\ref{prop:2}, and \ref{prop:4}
we know it can be written as:
\[\det\left(\mathbf{M}^{\mathbf{O},3}_{(1,5,6),(2,3,4)}\right) =(\alpha - \tilde{\alpha})\tilde{J}_{(1,5,6),(2,3,4)}\left(\tilde{\alpha},\theta^0\right)\]
where $ \tilde{J}_{(1,5,6),(2,3,4)}\left(\tilde{\alpha},\theta^0\right) $ is a polynomial in $ \tilde{\alpha} $ of degree 4.
Case \ref{casep:3} is plausible if both $ \tilde{J}_{(1,2,3),(4,5,6)}\left(\tilde{\alpha},\theta^0\right) $ and $ \tilde{J}_{(1,5,6),(2,3,4)}\left(\tilde{\alpha},\theta^0\right) $ share the root $ \tilde{\alpha} = 1 $.
This is equivalent to stating that $ \tilde{J}_{(1,5,6),(2,3,4)}\left(1,\theta^0\right) = 0 $.
Evaluating $ \tilde{J}_{(1,5,6),(2,3,4)}\left(\tilde{\alpha},\theta^0\right) $ at $ \tilde{\alpha} = 1 $ yields:
\begin{align*}
	\tilde{J}_{(1,5,6),(2,3,4)}\left(1,\theta^0\right)&= \left( d_4 f_2 - d_4 f_3 - \alpha d_3 f_2 + \alpha d_3 f_4 + \alpha^2 d_2 f_3 - \alpha^2 d_2 f_4 \right) \\
	&\quad \times \left( f_{\gamma}^2 \Psi_{1,2}^2 - f_{\gamma} \Psi_{1,2}^2 - d_1 \Psi_{2,2} - f_{\gamma}^2 \Psi_{1,1} \Psi_{2,2} \right. \\
	&\quad \left. - \alpha f_{\gamma}^2 \Psi_{1,2}^2 + \alpha d_1 \Psi_{2,2} + f_{\gamma} \Psi_{1,1} \Psi_{2,2} + \alpha f_{\gamma}^2 \Psi_{1,1} \Psi_{2,2} \right)
\end{align*}

The expression above is a product of two terms. Case~\ref{casep:3} would be plausible if either term were equal to zero. For each term, this can occur with positive probability only under low-dimensional restrictions on the parameter space. Consider the first factor. Setting it equal to zero gives $
(d_4 - \alpha d_3)f_2
-
(d_4 - \alpha^2 d_2)f_3
-
(\alpha^2 d_2 - \alpha d_3)f_4
=0$.

If this expression is not identically zero as a polynomial in
$(f_2,f_3,f_4)$, its zero set has Lebesgue measure zero. Hence, by
Assumption \ref{asp:6} and Lemma \ref{lemma:3}, this equality can hold
only on a probability-zero set of factor realizations. It remains to characterize when the expression is identically zero. This
requires all three coefficients to vanish:
\[
d_4-\alpha d_3=0,\qquad
d_4-\alpha^2d_2=0,\qquad
\alpha^2d_2-\alpha d_3=0.
\]
If $\alpha=0$, the second condition reduces to $d_4=0$, which is impossible
because $d_4>0$. Thus the expression cannot be identically zero when
$\alpha=0$. If $\alpha\neq0$, the three coefficient restrictions are
equivalent to $
d_{3}=\alpha d_{2} \ \text{and}\
d_{4}=\alpha d_{3}$. Therefore, when $\alpha\neq1$ and $\alpha\neq0$, the simultaneous restrictions
\[
d_{3}=\alpha d_{2}
\qquad\text{and}\qquad
d_{4}=\alpha d_{3}
\]
define an additional lower-dimensional subset of the parameter space on which
this part of the proof does not rule out observational equivalence. When
$\alpha=1$ and $\tilde\alpha=1$, we have $\tilde\alpha=\alpha$, so this is
not a non-identification case.

Next, consider the second term. It can be expressed as a second-degree polynomial in $ f_{\gamma} $:
\[
\mathcal{P}(f_{\gamma}) = a f_{\gamma}^{2} + b f_{\gamma} + c
\]
where $a= (\alpha-1)\left(\Psi_{1,1}\Psi_{2,2} - \Psi_{1,2}^2\right)$, $b= \Psi_{1,1}\Psi_{2,2} - \Psi_{1,2}^2$, $c = (\alpha-1) d_1 \Psi_{2,2}$.

Setting $\mathcal{P}(f_{\gamma})$ equal to zero is equivalent to finding the roots of this polynomial. First, note that $ \mathcal{P} $
can never be the zero polynomial since $ b > 0 $, which holds because $ \mathbf{\Psi} $ is PD by assumption. Then, by the fundamental theorem of algebra, there are at most two real values of $ f_{\gamma} $ that satisfy $ \mathcal{P}(f_{\gamma}) = 0 $. Thus, the two exact values corresponding to the roots of $\mathcal{P}(f_{\gamma})$ correspond to exceptional regions in the parameter space where our proof does not show identification, which correspond to a Lebesgue measure zero region. Thus, Case \ref{casep:3} can occur only on a Lebesgue-null subset of the parameter space.

\textbf{Case \ref{casep:4}}: $\tilde{\alpha}=-\frac{1}{f_4}$

Consider the case where $\tilde{\alpha}=-1/f_4$ and $f_4\neq0$ (if $f_4=0$, then the factor $\tilde\alpha f_4+1$ equals $1$, not zero so the problem disappears). As in the case $\tilde{\alpha}=1$, this value of $\tilde{\alpha}$ must also satisfy all other admissible diagonal-exclusion restrictions. We therefore evaluate the minor associated with rows $R=(1,2,4)$ and columns $C=(3,5,6)$ at $\tilde{\alpha}=-1/f_4$, which yields
\[
\tilde{J}_{(1,2,4),(3,5,6)}\left(-\frac{1}{f_4},\theta^0\right)=\frac{\mathcal{R}(\theta^0)}{f_4^3},
\]
where $\mathcal{R}(\theta^0)$ is a finite-degree polynomial in $f_1,\dots,f_4$. The coefficient associated with the monomial $f_1f_3f_4^3$ is $-d_4\det(\mathbf{\Psi})<0$, since $d_4>0$ and $\mathbf{\Psi}$ is positive definite. Hence $\mathcal{R}(\theta^0)$ is a nonzero polynomial in $f_1,\dots,f_4$. After multiplying by $f_4^3$, the restriction $\tilde{J}_{(1,2,4),(3,5,6)}(-1/f_4,\theta^0)=0$ therefore defines the zero set of a nonzero polynomial in $f_1,\dots,f_4$, apart from the event $f_4=0$, which also has probability zero under Assumption~\ref{asp:6}. By Lemmas~\ref{lemma:2} and~\ref{lemma:3}, this zero set has probability zero. Thus, Case~\ref{casep:4} can occur only on a probability-zero set of factor realizations.

It follows that, outside an exceptional set of lower dimension, Case \ref{casep:5} is the only remaining case under which $
\det\left(\mathbf{M}^{\mathbf{O},3}_{(1,2,3),(4,5,6)}\right)=0$ while all other minors are zero as well.

Therefore, for almost every parameterization, $\tilde{\alpha}=\alpha$ is a necessary condition for \eqref{eq5:objective} to hold in the presence of individual fixed effects. Hence, $\alpha$ is almost surely globally identified outside a low-dimensional exceptional subset of the parameter space. To characterize this subset, let $
Q=\Psi_{1,1}\Psi_{2,2}-\Psi_{1,2}^{2}$. The Lebesgue-null exceptional set used in the preceding argument is given by\footnote{Our proof strategy does not establish identification in this region of the parameter space, but it does not imply that this region is necessarily underidentified. It remains an open question.}
\begin{small}
\begin{align*}
	\Theta^{*}_{(1)}
	&=
	\left\{
	\theta\in\Theta: \alpha\neq 0,\;
	d_{2}f_{\gamma}=\alpha d_{1},\;
	d_{3}=\alpha d_{2}
	\right\}\cup
	\left\{
	\theta\in\Theta:
	\alpha\neq 1,\; \alpha\neq 0,\;
	d_{4}=\alpha d_{3},\;
	d_{3}=\alpha d_{2}
	\right\}
	\\
	&\quad\cup
	\left\{
	\begin{aligned}
		\theta\in\Theta:\;&
		\alpha\neq 1,\;
		Q>0,\;
		4(\alpha-1)^2 d_1\Psi_{2,2}\le Q,\\
		& f_{\gamma}
		=
		\frac{
			-1+
			\sqrt{
				1-
				\dfrac{4(\alpha-1)^2d_1\Psi_{2,2}}{Q}
			}
		}
		{2(\alpha-1)}
	\end{aligned}
	\right\}\cup
	\left\{
	\begin{aligned}
		\theta\in\Theta:\;&
		\alpha\neq 1,\;
		Q>0,\;
		4(\alpha-1)^2 d_1\Psi_{2,2}\le Q,\\
		& f_{\gamma}
		=
		\frac{
			-1-
			\sqrt{
				1-
				\dfrac{4(\alpha-1)^2d_1\Psi_{2,2}}{Q}
			}
		}
		{2(\alpha-1)}
	\end{aligned}
	\right\}.
\end{align*}
\end{small}	
Each component of $\Theta^{*}_{(1)}$ is contained in the zero set of a nonzero polynomial in the finite-dimensional parameter vector. Therefore, $\Theta^{*}_{(1)}$ has Lebesgue measure zero. For every $\theta^0\in\Theta\setminus\Theta^{*}_{(1)}$, Cases~\ref{casep:1}--\ref{casep:4} can occur only on sets of factor realizations with probability zero. Hence, with probability one over the factor realizations, the only remaining possibility is Case~\ref{casep:5}, namely $\tilde\alpha=\alpha$.\end{proof}
\begin{remark}
	We have proven almost-sure identification for the case when $ T = 6 $. Larger $ T $ only strengthens the identification. One way to see this is by applying the same normalization used for the case of $T=6$ even when $ T > 6 $. This corresponds to a non-tail normalization. That is, the transpose of $ \mathbf{F} $ is written as:
	\[
	\mathbf{F}' = \begin{bmatrix}
		f_\gamma & 1 & 1 & 1\ & 1\ & 1\ & 1 & \cdots & 1 \\
		f_1 & f_2 & f_3 & f_4\ & 0\ & 1\ & f_7 & \cdots & f_T
	\end{bmatrix}
	\]
	Our entire proof carries over without any modification.
\end{remark}

	\subsection*{Proof of Theorem \ref{thm:6}:}

\begin{proof} Consider the case where $\alpha\neq 1$ and let $T=4$. The minor associated with rows $R=(1,2)$ and columns $C=(3,4)$ can be written as:
\begin{equation}\label{eq28}
	\det\left(\mathbf{M}_{(1,2),(3,4)}^{\mathbf{O},2}\right) =\Psi(\alpha-\tilde{\alpha})(1-\tilde{\alpha})(\alpha d_1-d_2 f_\gamma)=0
\end{equation}
As $\Psi>0$, we need to rule out the cases where there exists an alternative parameterization with $\tilde{\alpha}=1$ or the particular case in which the true parameters satisfy $f_{\gamma} = \frac{\alpha d_{1}}{d_{2}}$.

\textbf{Regime 1: $\alpha\neq 1$ and $f_{\gamma} = \frac{\alpha d_{1}}{d_{2}}$.}

We first start by the case where the parameters satisfy $f_{\gamma} = \frac{\alpha d_{1}}{d_{2}}$, a restriction on the parameter space with Lebesgue measure zero. For this case, we require $T\geq 5$ to show identification. We follow a similar strategy as in the case when $f_{\gamma} = \frac{d_1}{d_2}$ and $\alpha=1$ in Theorem \ref{thm:5}. Evaluate the minor associated with rows $R=(1,3)$ and columns $C=(4,5)$ in the case when $f_{\gamma} = \frac{\alpha d_{1}}{d_{2}}$:
\begin{equation}\label{eq29}
	\det\left(\mathbf{M}_{(1,3),(4,5)}^{\mathbf{O},2}\right) = \alpha \frac{d_{1}}{d_{2}} \Psi \left(\alpha-\tilde{\alpha}\right)\left(\tilde{\alpha}-1\right)\left(d_{3}-\alpha d_{2}\right)=0
\end{equation}
This minor can be equal to zero in four cases: (i) $\tilde{\alpha}=\alpha$ (which means that $\alpha$ is globally identified), (ii) $\alpha=0$, (iii) $d_{3}=\alpha d_2$, or (iv) $\tilde{\alpha}=1$. The case where $\tilde{\alpha}=1$ will be addressed later in the proof, so we postpone the discussion of this case. Start with the case in which $\alpha=0$:

\textbf{Regime 1.1: $\alpha\neq 1$, $f_{\gamma} = \frac{\alpha d_{1}}{d_{2}}$, and $\alpha=0$.}

Consider case (ii): $\alpha=0$ and $f_{\gamma} = \frac{\alpha d_{1}}{d_{2}}$. Because $d_{2}>0$ this means that when $\alpha=0$ it must be the case that $f_{\gamma}=0$. Then, simply consider the minor associated with rows $R=(2,3)$ and columns $C=(1,5)$ evaluated at these restrictions:
\begin{equation}\label{eq30}
	\det\left(\mathbf{M}_{(2,3),(1,5)}^{\mathbf{O},2}\right)=	-\tilde{\alpha} d_{1}\Psi\left(\tilde{\alpha}-1\right)^{2}
\end{equation}
Since $d_1>0$, $\Psi>0$, and $\alpha=0$, any remaining admissible solution with $\tilde{\alpha}\neq\alpha$ must have $\tilde{\alpha}=1$.\footnote{Setting $\tilde{\alpha}=0$ means that $\alpha$ is indeed identified as we are working in the case where $\alpha=0$} As before, this alternative parameterization will be discussed later in the proof. We can conclude that in the case where $f_{\gamma} = \frac{\alpha d_{1}}{d_{2}}$, and $\alpha=0$ an alternative solution can only (potentially) exist if $\tilde{\alpha}=1$.

\textbf{Regime 1.2: $\alpha\neq 1$ and $f_{\gamma} = \frac{\alpha d_{1}}{d_{2}}$, $d_{3}=\alpha d_2$.}

Consider case (iii): $d_{3}=\alpha d_2$ and $f_{\gamma} = \frac{\alpha d_{1}}{d_{2}}$. Now, consider the minor associated with rows $R=(1,3)$ and columns $C=(2,4)$ under this parameterization:
	\begin{align*}
		\det\left(\mathbf{M}_{(1,3),(2,4)}^{\mathbf{O},2}\right)&=	\frac{d_1}{d_2}\left(1-\tilde{\alpha}\right)\left(\alpha-\tilde{\alpha}\right)\times\\
		&\times \left(d_{2}\Psi + \alpha^{2} d_{2}^{2}+\alpha d_{2} \Psi - \tilde{\alpha} d_{2} \Psi - \alpha\tilde{\alpha}d_{2}^{2}+\alpha^{2}d_{2}\Psi+\alpha^{4}d_{1}\Psi-\alpha^{3}\tilde{\alpha}d_{1}\Psi-\alpha\tilde{\alpha}d_{2}\Psi\right)
	\end{align*}
which means that either $\tilde{\alpha}=\alpha$ (achieving global identification), $\tilde{\alpha}=1$ (to be discussed later) or we need to construct an alternative parameterization where $\tilde{\alpha}^{\star} = \frac{d_{1}\Psi \alpha^4+\alpha^2d_{2}^{2}+\Psi \alpha^2 d_{2}+\Psi \alpha d_{2} + \Psi d_{2}}{d_{1}\Psi \alpha^3+\alpha d_{2}^{2}+\Psi \alpha d_{2} + \Psi d_{2}}$. If the denominator in the expression for $\tilde\alpha^\star$ is zero, then the
corresponding numerator is nonzero, since the numerator minus $\alpha$ times
the denominator equals $\Psi d_2>0$. Hence, when this denominator is zero, the
last factor cannot vanish for any value of $\tilde\alpha$. In that case, the
only remaining possibilities are the previous two cases,
$\tilde\alpha=\alpha$ and $\tilde\alpha=1$.

Now, we have no other degrees of freedom and this same value of $\tilde{\alpha}^{\star}$ should make all minors equal to zero. Consider the minor associated with rows $R=(2,3)$ and columns $C=(1,5)$ evaluated under this alternative parameterization and the restrictions imposed:
\begin{equation}\label{eq32}
	\det\left(\mathbf{M}_{(2,3),(1,5)}^{\mathbf{O},2}\right)=	-\frac{\alpha^{2} d_{1}d_{2} \Psi^{2}\left(d_{1}\Psi \alpha^{3}-d_{1}\Psi \alpha^{2} + \alpha d_{2}^{2}+ \Psi \alpha d_{2}-d_{2}^{2}\right)^{2}}{\left(d_{1}\Psi \alpha^{3}+\alpha d_{2}^{2}+\Psi \alpha d_2 + \Psi d_2\right)^{3}}=0
\end{equation}

Then, as $d_{1}>0, \ d_{2}>0, \ \Psi>0$, and $\alpha \neq 0$ (because we are working under the restriction where $d_{3}=\alpha d_{2}$, having $\alpha=0$ implies $d_{3}=0$ which contradicts the assumptions of the model) the only way this minor can be equal to zero is if the last term is equal to zero. This can only correspond (under some restrictions on the possible values of $d_1$, $d_2$, and $\alpha$) to the variance of the individual fixed effects being exactly equal to $\Psi = \frac{d_{2}^{2}(1-\alpha)}{\alpha\left(d_{2}-\alpha\left(1-\alpha\right)d_{1}\right)}$. Thus, we need to rule out the case where $d_{3}=\alpha d_2$, $f_{\gamma} = \frac{\alpha d_{1}}{d_{2}}$, $\Psi = \frac{d_{2}^{2}(1-\alpha)}{\alpha\left(d_{2}-\alpha\left(1-\alpha\right)d_{1}\right)}>0$, and $\tilde{\alpha}^{\star} = \frac{d_{1}\Psi \alpha^4+\alpha^2d_{2}^{2}+\Psi \alpha^2 d_{2}+\Psi \alpha d_{2} + \Psi d_{2}}{d_{1}\Psi \alpha^3+\alpha d_{2}^{2}+\Psi \alpha d_{2} + \Psi d_{2}}$. However, note that when $\Psi = \frac{d_{2}^{2}(1-\alpha)}{\alpha\left(d_{2}-\alpha\left(1-\alpha\right)d_{1}\right)}$ then $\tilde{\alpha}^{\star}$ simplifies to $\tilde{\alpha}^{\star}=1$, which is the separate case we need to study.

Therefore, when $f_{\gamma}=\alpha d_1/d_2$ and $T\geq5$, any remaining admissible alternative with $\tilde{\alpha}\neq\alpha$ must have $\tilde{\alpha}=1$. Thus, to show global identification of $\alpha$ it remains to show that when $\alpha\neq 1$ it cannot be the case that there is an alternative representation where $\tilde{\alpha}=1$. We discuss this case below:

\textbf{Regime 2: $\alpha\neq 1$ and $\tilde{\alpha}=1$.}

In this case showing global identification requires demonstrating that there is not an alternative parameterization of the model satisfying all the assumptions in which $\tilde{\alpha}=1$. We should pay special attention to the problematic regimes evaluated before (1.1 and 1.2) in which we did not rule out an alternative parameterization with $\tilde{\alpha}=1$. All subsequent parts of the proof use only row and column indices smaller or equal to 4, meaning that the proof below is valid for every $T\geq 4$, so this will enable us to rule out the case $\tilde{\alpha}=1$ both in the case where $T=4$ and $T\geq 5$.

For this problem, diagonal exclusion minors built using other rows and columns do not provide useful information, so we need to resort to other methods. Our goal is to show that $\nexists \ \tilde{\theta}$ with $\tilde{\alpha}\neq\alpha$ such that $\mathbf{\Sigma}(\theta^0)=\mathbf{\Sigma}(\tilde{\theta})$. Condition \eqref{eq28} combined with the subsequent analysis of regimes 1.1 and 1.2 permits us to conclude that if such $\tilde{\theta}$ exists, then it must be the case that $\tilde{\alpha}=1$. Moreover, such $\tilde{\theta}$ must satisfy that for each position of the matrix $\mathbf{\Sigma}$ it must be the case that $\Sigma(\theta^0)_{rc}=\Sigma(\tilde{\theta})_{rc}$, or identically: $\Sigma(\theta^0)_{rc}-\Sigma(\tilde{\theta})_{rc}=0$. However, we show that an alternative parameterization satisfying all assumptions of the model does not exist when $\tilde{\alpha}=1$. To see this, first consider the entries $(1,1), (2,1), (3,1)$ of $\mathbf{\Sigma}(\theta^0)-\mathbf{\Sigma}(\tilde{\theta})$ evaluated at $\tilde{\alpha}=1$:
\begin{equation}\label{eq33}
	\begin{cases}
		d_1+\Psi f_{\gamma}^{2}-(\tilde{d}_1+\tilde{\Psi}\tilde{f}_{\gamma}^{2})=0\\
		\alpha(\Psi f_{\gamma}^{2}+d_{1})+\Psi f_\gamma-(\tilde{\Psi}\tilde{f}_{\gamma}^{2}+\tilde{d}_{1}+\tilde{\Psi}\tilde{f}_{\gamma})=0\\
		\alpha^2(\Psi f_{\gamma}^{2}+d_{1})+\alpha\Psi f_\gamma+\Psi f_\gamma-(\tilde{\Psi}\tilde{f}_{\gamma}^{2}+\tilde{d}_{1}+2\tilde{\Psi}\tilde{f}_\gamma)=0
	\end{cases}
\end{equation}
This is a nonlinear system of three equations in three free parameters $(\tilde{d}_{1},\tilde{\Psi},\tilde{f}_{\gamma})$. However, this system only admits a solution if $\theta^0$ lies in a subset of the parameter space that has Lebesgue measure zero. To see this, note that we can rearrange the first equation into $\tilde{d}_{1} = d_1+\Psi f_{\gamma}^{2}-\tilde{\Psi}\tilde{f}_{\gamma}^{2}$. Plugging into the second equation gives:
\[\alpha(\Psi f_{\gamma}^{2}+d_{1})+\Psi f_\gamma-\tilde{\Psi}\tilde{f}_{\gamma}^{2}-(d_1+\Psi f_{\gamma}^{2}-\tilde{\Psi}\tilde{f}_{\gamma}^{2})-\tilde{\Psi}\tilde{f}_{\gamma}=0\]
Rearranging in terms of $\tilde{\Psi}\tilde{f}_{\gamma}$ yields $\tilde{\Psi}\tilde{f}_{\gamma}=\Psi f_\gamma-(1-\alpha)(d_1+\Psi f_{\gamma}^{2})$.
Finally, substituting $\tilde{d}_{1}$ and $\tilde{\Psi}\tilde{f}_{\gamma}$ into the third equation and rearranging yields:
\begin{align*}
	(\alpha-1)\left((\alpha-1)(d_1+\Psi f_{\gamma}^{2}) + \Psi f_{\gamma}\right)=0
\end{align*}
Note that the final condition depends only on the true structural parameters $(\alpha,d_1,\Psi,f_{\gamma})$. It defines a manifold in which the true value of these parameters must lie in order for an alternative solution with $\tilde{\alpha}=1$ to exist when $\alpha\neq 1$. In particular, solving for $\Psi$ we can conclude that such an alternative parameterization can only exist if $\Psi=\frac{(1-\alpha)d_{1}}{f_{\gamma}(1-(1-\alpha)f_{\gamma})}$\footnote{This imposes some restrictions on the feasible values of $\alpha$, $d_1$ and $f_{\gamma}$.}. Therefore, unless the variance of the individual fixed effects is exactly equal to this expression, the candidate case $\tilde{\alpha}=1$ is ruled out by these covariance restrictions. Also, note that this restriction does not have a very clear economic interpretation. This restriction on $\Psi$ automatically rules out the regime 1.1 because when $\alpha=0$ we have $f_{\gamma}=0$, but in this case $\Psi$ cannot be equal to $\frac{(1-\alpha)d_{1}}{f_{\gamma}(1-(1-\alpha)f_{\gamma})}$ so an alternative solution with $\tilde{\alpha}=1$ cannot exist when $\alpha=0$.

However, we show that even when $\Psi=\frac{(1-\alpha)d_{1}}{f_{\gamma}(1-(1-\alpha)f_{\gamma})}$ the autoregressive coefficient $\alpha$ is globally identified. To see this, consider the solution to the system \eqref{eq33} evaluated when $\Psi=\frac{(1-\alpha)d_{1}}{f_{\gamma}(1-(1-\alpha)f_{\gamma})}$. In particular, the expression derived before reduces to:
\[\tilde{\Psi}\tilde{f}_{\gamma}=\frac{(1-\alpha)d_{1}}{f_{\gamma}(1-(1-\alpha)f_{\gamma})} f_\gamma-(1-\alpha)\left(d_1+\frac{(1-\alpha)d_{1}}{f_{\gamma}(1-(1-\alpha)f_{\gamma})} f_{\gamma}^{2}\right) = 0 \]
This means that the product $\tilde{\Psi}\tilde{f}_{\gamma}=0$. Under the assumptions of the model, any alternative parameterization must satisfy $\tilde{\Psi}>0$, which implies that $\tilde{f}_{\gamma}=0$. Then, substituting into position (1,1) under the same restriction yields $
\tilde{d}_{1} = \frac{d_{1}}{1-(1-\alpha) f_{\gamma}}$.

Then, any alternative solution with $\tilde{\alpha}=1$, must also satisfy $\tilde{d}_{1} = \frac{d_{1}}{1-(1-\alpha) f_{\gamma}}$ and $\tilde{f}_{\gamma}=0$, which is valid only when the true variance of the individual fixed effects is equal to $\Psi=\frac{(1-\alpha)d_{1}}{f_{\gamma}(1-(1-\alpha)f_{\gamma})}$. Now, evaluate the difference between the entries $(2,2), (2,3)$ of $\mathbf{\Sigma}(\theta^0)-\mathbf{\Sigma}(\tilde{\theta})$ in this case and set them equal to zero:
\begin{equation}\label{eq34}
	\begin{cases}
		f_{\gamma}^{-1}\left(d_{1}(1-\alpha)+d_{2}f_{\gamma}-f_{\gamma}\tilde{\Psi}-f_{\gamma}\tilde{d}_{2}\right) = 0\\
		f_{\gamma}^{-1}\left(d_{1}(1-\alpha^2)+d_{2}f_{\gamma}\alpha-2f_{\gamma}\tilde{\Psi}-f_{\gamma}\tilde{d}_{2}\right) = 0
	\end{cases}
\end{equation}
where $f_{\gamma}\neq 0$ because we are working in the case where $\Psi=\frac{(1-\alpha)d_{1}}{f_{\gamma}(1-(1-\alpha)f_{\gamma})}$. This is a system of two linear equations in the free parameters $\tilde{\Psi},\tilde{d}_{2}$ with a unique solution given by: $\tilde{\Psi} = \frac{\alpha d_{1}-d_{2}f_{\gamma}-\alpha^{2}d_{1}+\alpha d_{2} f_{\gamma}}{f_{\gamma}}$ and $\tilde{d}_{2} = \frac{d_{1}-2\alpha d_{1}+2 d_{2} f_{\gamma} + \alpha^{2} d_{1}-\alpha d_{2} f_{\gamma}}{f_{\gamma}}$. However, this alternative parameterization with all the restrictions imposed so far must make all positions of $\mathbf{\Sigma}(\theta^0)-\mathbf{\Sigma}(\tilde{\theta})$ equal to zero. Analyzing the position (2,4) of this difference under these restrictions yields:
\begin{equation}\label{eq35}
	-f_{\gamma}^{-1}\left(\alpha-1\right)^{2}\left(\alpha d_1 -d_2 f_{\gamma}\right)=0
\end{equation}
which imposes an additional restriction on the true value of the parameters beyond the one already imposed in $\Psi$. In particular, in order for an alternative parameterization to exist when $\alpha\neq 1$ it must be the case that $ f_{\gamma}=\frac{\alpha d_{1}}{d_{2}}$. Thus, under the restrictions imposed so far on the alternative parameters, the existence of an admissible alternative with $\tilde{\alpha}=1$ requires $\alpha\neq1$, $\Psi=\frac{(1-\alpha)d_1}{f_\gamma(1-(1-\alpha)f_\gamma)}$, and $f_\gamma=\alpha d_1/d_2$.

Nevertheless, even when simultaneously $\alpha\neq 1$, $\Psi=\frac{(1-\alpha)d_{1}}{f_{\gamma}(1-(1-\alpha)f_{\gamma})}$, and $ f_{\gamma}=\frac{\alpha d_{1}}{d_{2}}$ the restrictions imposed on $\tilde{\theta}$ prevent an alternative solution with $\tilde{\alpha}=1$ to exist. To see this, recall that the alternative parameterization must satisfy
$\tilde{\Psi} = \frac{\alpha d_{1}-d_{2}f_{\gamma}-\alpha^{2}d_{1}+\alpha d_{2} f_{\gamma}}{f_{\gamma}}>0$, but when $ f_{\gamma}=\frac{\alpha d_{1}}{d_{2}}$ it is easy to verify that $\tilde{\Psi} = 0$, which violates the restrictions imposed on the parameters. Then, an alternative solution with $\tilde{\alpha}=1$ cannot exist when $\alpha\neq 1$.


All in all, when $T=4$ and $\alpha\neq1$, the autoregressive coefficient is globally identified for all parameter values outside the Lebesgue-null set $\Theta^{*}_{(3)}=\{\theta\in\Theta:f_\gamma=\alpha d_1/d_2\}$.

Similarly, the preceding restrictions show that, when $T\geq5$, any admissible parameterization satisfying the observational equivalence condition must have $\tilde{\alpha}=\alpha$. Hence $\alpha$ is globally identified. Thus, provided $T\geq 5$, Theorems \ref{thm:5} and \ref{thm:6} imply that global identification of $\alpha$ is guaranteed for every value of $\alpha$. \end{proof}

\section{Proof of intermediate results:}

\subsection*{Proof of Lemma \ref{lem:cofactor-representation-alpha}:}

\begin{proof}
	We prove the result by induction on $k$, starting from $k=2$. The
	one-by-one case used in the base step is $
	\tilde J_{\{r\},\{c\}}(\alpha,\theta^0)=J^\alpha_{r,c}$, which follows directly from the defining recursion evaluated at
	$\tilde\alpha=\alpha$. First consider $k=2$. Let $
	R=\{r_1,r_2\}, \
	C=\{c_1,c_2\}$
	with $r_1<r_2$ and $c_1<c_2$. Setting $\tilde\alpha=\alpha$ in the recursion defining
	$\tilde J_{R,C}(\tilde\alpha,\theta^0)$, the term multiplied by
	$(\alpha-\tilde\alpha)$ vanishes, and we get
	\[
	\begin{aligned}
		\tilde J_{R,C}(\alpha,\theta^0)
		&=
		\sum_{j=1}^{2}(-1)^{2+j}
		\Big[
		J^\alpha_{r_2,c_j}
		\det\left(\mathbf{M}^{\mathbf{\Omega},1}_{R-r_2,C-c_j}\right)
		+
		\Omega_{r_2,c_j}
		\tilde J_{R-r_2,C-c_j}(\alpha,\theta^0)
		\Big] \\
		&=
		\sum_{j=1}^{2}(-1)^{2+j}
		J^\alpha_{r_2,c_j}
		\det\left(\mathbf{M}^{\mathbf{\Omega},1}_{R-r_2,C-c_j}\right) 	-\Omega_{r_2,c_1}J^\alpha_{r_1,c_2}
		+
		\Omega_{r_2,c_2}J^\alpha_{r_1,c_1}.
	\end{aligned}
	\]
	The first sum is exactly the contribution with $i=2$. The last two terms
	are the contribution with $i=1$, because $
	\det\left(\mathbf{M}^{\mathbf{\Omega},1}_{R-r_1,C-c_1}\right)
	=
	\Omega_{r_2,c_2}$, and
	$\det\left(\mathbf{M}^{\mathbf{\Omega},1}_{R-r_1,C-c_2}\right)
	=
	\Omega_{r_2,c_1}$. Hence
	\[
	\tilde J_{R,C}(\alpha,\theta^0)
	=
	\sum_{i=1}^{2}\sum_{j=1}^{2}
	(-1)^{i+j}
	J^\alpha_{r_i,c_j}
	\det\left(\mathbf{M}^{\mathbf{\Omega},1}_{R-r_i,C-c_j}\right).
	\]
	Thus the claim holds for $k=2$.
	
	Now let $k\ge 3$, and suppose the claim holds for all pairs of ordered row
	and column sets of size $k-1$. We prove it for $
	R=\{r_1,\cdots,r_k\},
	\
	C=\{c_1,\cdots,c_k\}$
	where the indices satisfy $r_1<\cdots<r_{k}$, and $c_1<\cdots<c_{k}$ without loss of generality (see Remark 2 in the paper). Setting $\tilde\alpha=\alpha$ in the recursion defining
	$\tilde J_{R,C}(\tilde\alpha,\theta^0)$, the final term vanishes, and we
	obtain
	\[
	\begin{aligned}
		\tilde J_{R,C}(\alpha,\theta^0)
		&=
		\sum_{j=1}^{k}(-1)^{k+j}
		\Big[
		J^\alpha_{r_k,c_j}
		\det\left(\mathbf{M}^{\mathbf{\Omega},k-1}_{R-r_k,C-c_j}\right)
		+
		\Omega_{r_k,c_j}
		\tilde J_{R-r_k,C-c_j}(\alpha,\theta^0)
		\Big] \\
		&=
		\sum_{j=1}^{k}(-1)^{k+j}
		J^\alpha_{r_k,c_j}
		\det\left(\mathbf{M}^{\mathbf{\Omega},k-1}_{R-r_k,C-c_j}\right) +
		\sum_{j=1}^{k}(-1)^{k+j}
		\Omega_{r_k,c_j}
		\tilde J_{R-r_k,C-c_j}(\alpha,\theta^0).
	\end{aligned}
	\]
	The first sum is already the contribution with $i=k$. It remains to rewrite
	the second sum. Fix $j$. The sets $R-r_k$ and $C-c_j$ both have cardinality $k-1$.
	By the induction hypothesis applied to these smaller sets,
	\[
	\tilde J_{R-r_k,C-c_j}(\alpha,\theta^0)
	=
	\sum_{i=1}^{k-1}\sum_{\ell\neq j}
	(-1)^{i+p_j(\ell)}
	J^\alpha_{r_i,c_\ell}
	\det\left(
	\mathbf{M}^{\mathbf{\Omega},k-2}_{R-r_k-r_i,\ C-c_j-c_\ell}
	\right),
	\]
	where $p_j(\ell)$ is the position of $c_\ell$ in the ordered set
	$C-c_j$:
	\[
	p_j(\ell)
	=
	\begin{cases}
		\ell, & \ell<j,\\
		\ell-1, & \ell>j.
	\end{cases}
	\]
	Substituting this expression into the second sum gives
	\[
	\begin{aligned}
		&\sum_{j=1}^{k}(-1)^{k+j}
		\Omega_{r_k,c_j}
		\tilde J_{R-r_k,C-c_j}(\alpha,\theta^0) =
		\sum_{i=1}^{k-1}\sum_{\ell=1}^{k}
		J^\alpha_{r_i,c_\ell}
		\sum_{j\neq \ell}
		(-1)^{k+j+i+p_j(\ell)}
		\Omega_{r_k,c_j}
		\det\left(
		\mathbf{M}^{\mathbf{\Omega},k-2}_{R-r_k-r_i,\ C-c_j-c_\ell}
		\right).
	\end{aligned}
	\]
	
	Now fix $i<k$ and $\ell$. In the ordered row set $R-r_i$, the row
	$r_k$ has position $k-1$. Let $p_\ell(j)$ denote the position of
	$c_j$ in the ordered set $C-c_\ell$. Thus
	\[
	p_\ell(j)
	=
	\begin{cases}
		j, & j<\ell,\\
		j-1, & j>\ell.
	\end{cases}
	\]
	Expanding $\det\left(\mathbf{M}^{\mathbf{\Omega},k-1}_{R-r_i,C-c_\ell}\right)$ along the row $r_k$, we get
	\[
	\det\left(\mathbf{M}^{\mathbf{\Omega},k-1}_{R-r_i,C-c_\ell}\right)
	=
	\sum_{j\neq \ell}
	(-1)^{k-1+p_\ell(j)}
	\Omega_{r_k,c_j}
	\det\left(
	\mathbf{M}^{\mathbf{\Omega},k-2}_{R-r_k-r_i,\ C-c_j-c_\ell}
	\right).
	\]
	The signs agree after factoring out $(-1)^{i+\ell}$. Indeed, $
	(-1)^{k+j+i+p_j(\ell)}
	=
	(-1)^{i+\ell+k-1+p_\ell(j)}.$ If $j<\ell$, then $p_j(\ell)=\ell-1$ and $p_\ell(j)=j$. If
	$j>\ell$, then $p_j(\ell)=\ell$ and $p_\ell(j)=j-1$. In both cases the
	two exponents have the same parity. Therefore
	\[
	\begin{aligned}
		&\sum_{j\neq \ell}
		(-1)^{k+j+i+p_j(\ell)}
		\Omega_{r_k,c_j}
		\det\left(
		\mathbf{M}^{\mathbf{\Omega},k-2}_{R-r_k-r_i,\ C-c_j-c_\ell}
		\right) =
		(-1)^{i+\ell}
		\det\left(\mathbf{M}^{\mathbf{\Omega},k-1}_{R-r_i,C-c_\ell}\right).
	\end{aligned}
	\]
	It follows that
	\[
	\sum_{j=1}^{k}(-1)^{k+j}
	\Omega_{r_k,c_j}
	\tilde J_{R-r_k,C-c_j}(\alpha,\theta^0)
	=
	\sum_{i=1}^{k-1}\sum_{\ell=1}^{k}
	(-1)^{i+\ell}
	J^\alpha_{r_i,c_\ell}
	\det\left(\mathbf{M}^{\mathbf{\Omega},k-1}_{R-r_i,C-c_\ell}\right).
	\]
	Combining this with the first sum, which is the contribution with $i=k$,
	yields
	\[
	\tilde J_{R,C}(\alpha,\theta^0)
	=
	\sum_{i=1}^{k}\sum_{j=1}^{k}
	(-1)^{i+j}
	J^\alpha_{r_i,c_j}
	\det\left(\mathbf{M}^{\mathbf{\Omega},k-1}_{R-r_i,C-c_j}\right).
	\]
	This completes the induction.
\end{proof}

\subsection*{Proof of Lemma \ref{lem:span-signed-minors}:}

\begin{proof}
	The argument in this lemma is deterministic. The entries of the non-normalized
	rows of $\mathbf{F}$, collected in $\mathbf{F}_2$, are treated as unrestricted coordinates
	ranging over $\mathbb R^{(T-\bar r)\bar r}$. They are not treated as realized
	random draws throughout this proof.
	
	Now, suppose the claim of the Lemma is not true. Then there exists a nonzero vector
	$b=(b_1,\ldots,b_{\bar{r}+1})'$ such that $
	b'u_S(\mathbf{F})=0$
	for every value of $\operatorname{vec}(\mathbf{F}_2)\in\mathbb R^{(T-\bar r)\bar r}$. Indeed, let
	\[
	\mathbf{G}=
	\begin{pmatrix}
		b_1 & \mathbf{F}_{s_1}\\
		\vdots & \vdots\\
		b_{\bar{r}+1} & \mathbf{F}_{s_{\bar{r}+1}}
	\end{pmatrix}.
	\]
	Expanding $\det(\mathbf{G})$ along the first column gives
	\[
	\det(\mathbf{G})
	=
	\sum_{i=1}^{\bar r+1}
	(-1)^{i+1}b_i\det(\mathbf{F}_{S-s_i}).
	\]
	On the other hand, by the definition of $u_S(\mathbf{F})$,
	\[
	b'u_S(\mathbf{F})
	=
	\sum_{i=1}^{\bar r+1}
	b_i(-1)^i\det(\mathbf{F}_{S-s_i}).
	\]
	Therefore $
	\det(\mathbf{G})=-b'u_S(\mathbf{F})$. Thus $b'u_S(\mathbf{F})=0$ is equivalent to $
	\det(\mathbf{G})=0$ for every value of $\operatorname{vec}(\mathbf{F}_2)\in\mathbb R^{(T-\bar r)\bar r}$. Now, define the set of indices:
	\[
	I=\{i:s_i\le \bar{r}\},
	\qquad
	J=\{1,\ldots,\bar{r}+1\}\setminus I.
	\]
	The rows indexed by $I$ are fixed identity rows, while the rows indexed by $J$ are non-normalized rows (see Assumption \ref{as:2}). Since $S$ has $\bar{r}+1$ elements but there are only $\bar{r}$
	normalized rows, $J\neq\emptyset$ (meaning that in $u_S(\mathbf{F})$ there is at least one non-fixed factor).
	
	Here $i$ and $j$ index positions in the ordered set $S$, whereas
	$s_i$ and $s_j$ are the corresponding row indices in the original matrix
	$\mathbf{F}$. Thus $I$ and $J$ are subsets of $\{1,\ldots,\bar r+1\}$, while
	$\{s_i:i\in I\}$ and $\{s_j:j\in J\}$ are subsets of the original row-index
	set. In particular, for $i\in I$, the row $\mathbf{F}_{s_i}$ is normalized, while
	for $j\in J$, the row $\mathbf{F}_{s_j}$ is non-normalized.
	
	For each $j\in J$, write the non-normalized row as
	\[
	\mathbf{F}_{s_j}=(x_{j,1},\ldots,x_{j,\bar{r}}),
	\]
	where the coordinates $x_{j,1},\ldots,x_{j,\bar r}$ are unrestricted coordinates under the normalization. For each $i\in I$, the normalized row satisfies $
	\mathbf{F}_{s_i}=e_{s_i}'$.
	
	Thus row $i$ in $\mathbf{G}$ has the form $(b_i,e_{s_i}')$,
	with a $1$ in factor column $s_i$ and zeros in the other factor columns. Now define $\mathbf{B}$ to be the $(\bar r+1)\times(\bar r+1)$ matrix obtained after
	eliminating, from every non-normalized row $j\in J$, the entries in the normalized factor columns.
	Formally, for $i\in I$, define
	\[
	\mathbf{B}_{i,\cdot}
	=
	(b_i,e_{s_i}'),
	\]
	and, for $j\in J$, define $
	\mathbf{B}_{j,\cdot}
	=
	(b_j,\mathbf{F}_{s_j})
	-
	\sum_{i\in I}x_{j,s_i}(b_i,e_{s_i}')$. Equivalently, for $j\in J$,
	\[
	\mathbf{B}_{j,\cdot}
	=
	\left(
	b_j-\sum_{i\in I}x_{j,s_i}b_i,\,
	\mathbf{F}_{s_j}-\sum_{i\in I}x_{j,s_i}e_{s_i}'
	\right).
	\]
	Since $\mathbf{B}$ is obtained from $\mathbf{G}$ by adding multiples of
	rows to other rows, the determinant is unchanged. Therefore $
	\det(\mathbf{B})
	=
	\det
	(\mathbf{G})$. For $j\in J$, the first entry of row $j$ of $\mathbf{B}$ is
	\[
	B_{j,1}
	\equiv
	\tilde b_j
	=
	b_j-\sum_{i\in I}x_{j,s_i}b_i.
	\]
	For the remaining entries, recall that column $1+q$ of $\mathbf{B}$ corresponds
	to factor column $q$, $q=1,\ldots,\bar r$. Hence, for $j\in J$,
	\[
	B_{j,1+q}
	=
	x_{j,q}
	-
	\sum_{i\in I}x_{j,s_i}\mathbf 1\{q=s_i\}
	=
	\begin{cases}
		0, & q\in\{s_i:i\in I\},\\[0.3em]
		x_{j,q}, & q\notin\{s_i:i\in I\}.
	\end{cases}
	\]
	Thus the entries of row $j$ in the normalized factor columns are zero, while
	the entries in the remaining factor columns are unchanged.
	
	The rows indexed by $I$ still contain identity pivots in the factor columns
	$s_i$, $i\in I$, while the rows indexed by $J$ have zeros in those
	columns. We now expand explicitly along these identity pivots. Let $m=|I|$, and write $
	I=\{i_1,\cdots,i_m\}, \
	J=\{j_1,\cdots,j_{|J|}\}$,
	where we assume the sets are ordered such that $i_1<\cdots<i_m$ and $j_1<\cdots<j_{|J|}$ without any loss of generality. The first column of $\mathbf{B}$ is the $b$-column, and factor column $q$ is column $1+q$ of $\mathbf{B}$.
	
	Define the set of pivot columns in the augmented
	matrix by $
	P
	=
	\{1+s_i:i\in I\}
	=
	\{1+s_{i_1},\cdots,1+s_{i_m}\}$, with the elements ordered increasingly. After the row operations, the submatrix of $\mathbf{B}$ with rows $I$ and columns
	$P$ is the identity matrix, while the submatrix with rows $J$ and columns
	$P$ is zero:
	\[
	\mathbf{B}_{I,P}=\mathbf{I}_m,\qquad \mathbf{B}_{J,P}=0,
	\]
	where $m=|I|$, $\mathbf{B}_{I,P}$ is $m\times m$, and $\mathbf{B}_{J,P}$ is $|J|\times m$. By the Laplace expansion by complementary minors along the columns $P$,
	\[
	\det(\mathbf{B})
	=
	\sum_{\substack{\mathcal{T}\subseteq \{1,\ldots,\bar r+1\}\\ |\mathcal{T}|=m}}
	(-1)^{\sum_{t\in \mathcal{T}}t+\sum_{p\in P}p}
	\det(\mathbf{B}_{\mathcal{T},P})
	\det(\mathbf{B}_{\mathcal{T}^c,P^c}),
	\]
	where, for any $\mathcal{T}\subseteq\{1,\ldots,\bar r+1\}$, we have
	$\mathcal{T}^c=\{1,\ldots,\bar r+1\}\setminus \mathcal{T},
	\
	P^c=\{1,\ldots,\bar r+1\}\setminus P$.
	Observe that for any $\mathcal{T}\subseteq\{1,\ldots,\bar r+1\}$ with $|\mathcal{T}|=m$, the minor $\det(\mathbf{B}_{\mathcal{T},P})$ vanishes whenever $\mathcal{T}$ contains at least one element of $J$, because the corresponding rows of $\mathbf{B}_{\mathcal{T},P}$ are zero on the columns $P$ (since $\mathbf{B}_{J,P}=0$). The only subset of size $m$ that does not intersect $J$ is $\mathcal{T}=I$ itself; hence all terms in the Laplace sum are zero except the one with $\mathcal{T}=I$. Hence
	\[
	\det(\mathbf{B})
	=
	(-1)^{\sum_{i\in I}i+\sum_{i\in I}(1+s_i)}
	\det(\mathbf{B}_{I,P})
	\det(\mathbf{B}_{J,P^c})
	=
	(-1)^{\sum_{i\in I}i+\sum_{i\in I}(1+s_i)}
	\det(\mathbf{B}_{J,P^c}).
	\]
	Since $\det(\mathbf{B}_{I,P})=1$, the original determinant equals, up to a sign, $\det(\mathbf{B}_{J,P^c})$.
	Now define the remaining factor-column index set $
	H
	=
	\{1,\ldots,\bar r\}\setminus\{s_i:i\in I\}$.
	
	The columns $P^c$ consist of the first augmented column, namely the
	$\tilde b$-column, together with the factor columns in $H$. Therefore
	\[
	\mathbf{B}_{J,P^c}
	=
	\begin{pmatrix}
		\tilde b & \mathbf{Z}
	\end{pmatrix},
	\]
	where $
	\tilde b=(\tilde b_j)_{j\in J}\in\mathbb R^{|J|}$, and
	$\mathbf{Z}
	=
	\left(
	x_{j,\ell}
	\right)_{
		j\in J,\
		\ell\in H
	}$. The rows of $\mathbf{Z}$ are ordered according to the increasing order of $J$, and
	the columns of $\mathbf{Z}$ are ordered according to the increasing order of $H$.
	Since there are $|J|$ remaining rows and $
	|H|=\bar r-|I|=|J|-1$, we have $
	\mathbf{Z}\in\mathbb R^{|J|\times(|J|-1)}$.
	
	The set $H$ may be empty. In that case $|J|=1$, $\mathbf{Z}$ is a
	$1\times 0$ matrix, and $(\tilde b \; \mathbf{Z})$ is the $1\times 1$ matrix
	with entry $\tilde b_j$, where $J=\{j\}$.
	
	The entries of $\mathbf{Z}$ depend only on the coordinates $x_{j,\ell}$ with
	$j\in J$ and $\ell\in H$. These coordinates are unrestricted under the
	normalization imposed in Assumption \ref{as:2}. Hence, as the non-normalized rows vary over their admissible
	parameter space, $\mathbf{Z}$ ranges over $
	\mathbb R^{|J|\times(|J|-1)}$.
	
	Now, suppose first that $b_i\neq0$ for some $i\in I$. Fix such an index
	$i_0\in I$. Let $w=(w_j)_{j\in J}\in\mathbb R^{|J|}$ be an arbitrary vector. Holding
	all coordinates other than $\{x_{j,s_{i_0}}:j\in J\}$ fixed, set
	\[
	x_{j,s_{i_0}}
	=
	\frac{
		b_j-w_j-\sum_{i\in I,\ i\neq i_0}x_{j,s_i}b_i
	}{
		b_{i_0}
	},
	\qquad j\in J.
	\]
	Then
	\[
	\tilde b_j
	=
	b_j-\sum_{i\in I}x_{j,s_i}b_i,
	\qquad j\in J.
	\]
	Thus, as $\operatorname{vec}(\mathbf{F}_2)$ ranges over
	$\mathbb R^{(T-\bar r)\bar r}$, the vector $\tilde b$ can attain any value
	in $\mathbb R^{|J|}$.
	
	The coordinates used above to set $\tilde b$ do not enter $\mathbf{Z}$. Indeed,
	$\mathbf{Z}$ uses only factor columns in $
	H=\{1,\ldots,\bar r\}\setminus\{s_i:i\in I\}$, whereas $s_{i_0}\notin H$. Hence the choice of
	$\{x_{j,s_{i_0}}:j\in J\}$ imposes no restriction on the coordinates entering
	$\mathbf{Z}$. Since the coordinates entering $\mathbf{Z}$ range over
	$\mathbb R^{|J|\times(|J|-1)}$, there exists a value of
	$\operatorname{vec}(\mathbf{F}_2)\in\mathbb R^{(T-\bar r)\bar r}$ such that
	$\tilde b$ and the columns of $\mathbf{Z}$ form a basis of $\mathbb R^{|J|}$.
	
	For this value of $\operatorname{vec}(\mathbf{F}_2)$, $
	\det\begin{pmatrix}
		\tilde b & \mathbf{Z}
	\end{pmatrix}
	\neq0$, contradicting the assumption that $\det(\mathbf{G})=0$ for every value of
	$\operatorname{vec}(\mathbf{F}_2)\in\mathbb R^{(T-\bar r)\bar r}$. Hence this case is
	impossible, so $b_i=0$ for all $i\in I$.
	
	With $b_i=0$ for all $i\in I$, we have
	\[
	\tilde b_j
	=
	b_j-\sum_{i\in I}x_{j,s_i}b_i
	=
	b_j,
	\qquad j\in J.
	\]
	Thus $
	\tilde b=(b_j)_{j\in J}$. If this vector were nonzero, then, because $\mathbf{Z}$ ranges over
	$\mathbb R^{|J|\times(|J|-1)}$, there exists a value of the coordinates entering $\mathbf{Z}$ such that its $|J|-1$ columns complete the nonzero vector $\tilde b$ to a basis of $\mathbb R^{|J|}$. This would again imply $
	\det\begin{pmatrix}
		\tilde b & \mathbf{Z}
	\end{pmatrix}
	\neq0$, a contradiction. Therefore $b_j=0$ for all $j\in J$.
	
	We have shown that $b_i=0$ for all $i\in I$ and $b_j=0$ for all
	$j\in J$. Since $I\cup J=\{1,\ldots,\bar{r}+1\}$, this implies $b=0$,
	contradicting the choice of $b\neq0$.
	
	Then, we have shown that $
	\operatorname{span}
	\left\{
	u_S(\mathbf{F}): \operatorname{vec}(\mathbf{F}_2)\in\mathbb R^{(T-\bar r)\bar r}
	\right\}
	=
	\mathbb R^{\bar r+1}$.\end{proof}

\subsection*{Proof of Lemma \ref{lem:bilinear-nonzero-polynomial}}

\begin{proof}
	The argument is deterministic. The unrestricted entries of $\mathbf{F}_2$ are treated
	as coordinates ranging over their admissible parameter space, not as realized
	random draws. First, $u_R(\mathbf{F})'\mathbf{A}\,u_C(\mathbf{F})$ is a polynomial in the unrestricted entries of
	$\mathbf{F}_2$, because each component of $u_R(\mathbf{F})$ and $u_C(\mathbf{F})$ is a determinant
	of a submatrix of $\mathbf{F}$, and therefore is polynomial in the entries of $\mathbf{F}$.
	The normalized rows are fixed by Assumption \ref{as:2}, so the only variable
	entries are the unrestricted entries $\operatorname{vec}\left(\mathbf{F}_2\right)$.
	
	We now show that this polynomial is not identically zero. Suppose, toward a
	contradiction, that $
	u_R(\mathbf{F})'\mathbf{A}\,u_C(\mathbf{F})$ were the zero polynomial in $\mathbf{F}_2$. Then $
	u_R(\mathbf{F})'\mathbf{A}\,u_C(\mathbf{F})=0$ for every admissible deterministic assignment of those unrestricted entries. Because $R\cap C=\emptyset$, no row of $\mathbf F$ that enters any
	component of $u_R(\mathbf F)$ enters any component of $u_C(\mathbf F)$.
	Thus the unrestricted coordinates of $\mathbf F_2$ that affect
	$u_R(\mathbf F)$ are disjoint from those that affect $u_C(\mathbf F)$.
	Consequently, any attainable value of $u_R(\mathbf F)$ can be combined with
	any attainable value of $u_C(\mathbf F)$ in a single admissible assignment of
	$\mathbf F_2$.
	
	Hence, for every attainable value $u$ of $u_R(\mathbf F)$
	and every attainable value $v$ of $u_C(\mathbf F)$, there exists an
	admissible deterministic assignment of the unrestricted entries of
	$\mathbf F_2$ such that we can set
	$u_R(\mathbf F)=u$, and
	$u_C(\mathbf F)=v$.
	
	By the contradiction assumption, this implies $
	u'\mathbf A v=0$ for every attainable value $u$ of $u_R(\mathbf F)$ and every attainable
	value $v$ of $u_C(\mathbf F)$. By Lemma~\ref{lem:span-signed-minors}, the attainable values of $u_R(\mathbf{F})$
	span $\mathbb R^k$, and the attainable values of $u_C(\mathbf{F})$ also span
	$\mathbb R^k$. Therefore, for arbitrary $x,y\in\mathbb R^k$, we can write
	\[
	x=\sum_{a=1}^{m}\lambda_a u_a,
	\qquad
	y=\sum_{b=1}^{n}\mu_b v_b,
	\]
	where each $u_a$ is an attainable value of $u_R(\mathbf{F})$ and each $v_b$ is an
	attainable value of $u_C(\mathbf{F})$. By bilinearity,
	\[
	x'\mathbf{A}y
	=
	\sum_{a=1}^{m}\sum_{b=1}^{n}
	\lambda_a\mu_b\,u_a'\mathbf{A}v_b.
	\]
	Each term $u_a'\mathbf{A}v_b$ is zero by the preceding paragraph. Hence
	\[
	x'\mathbf{A}y=0
	\qquad
	\text{for all }x,y\in\mathbb R^k.
	\]
	Therefore $\mathbf{A}=0$, contradicting the assumption that $\mathbf{A}\neq0$. Hence $
	u_R(\mathbf{F})'\mathbf{A}\,u_C(\mathbf{F})$ is a nonzero polynomial in the unrestricted entries of $\mathbf{F}_2$.
\end{proof}

\end{document}